\newif\ifTR
\TRfalse
\TRtrue
    \documentclass[a4paper,UKenglish,cleveref, autoref, thm-restate,authorcolumns]{lipics-v2019}

    \usepackage{xcolor}
    \usepackage{xspace}
    \usepackage[utf8]{inputenc}
    \usepackage[T1]{fontenc}
    \usepackage{listings}
    \usepackage{textgreek}
    \usepackage{mathpartir}
    \usepackage{rotating}
    \usepackage{amsmath}
    \usepackage{amssymb}
    \usepackage{stmaryrd}
    \usepackage{mathtools}
    \usepackage{dashbox}
    \usepackage{array}
    \usepackage{arydshln}
    \usepackage[most]{tcolorbox}
    \usepackage[mathscr]{eucal}

\usepackage[utf8]{inputenc}
    \newcommand{\prophecy}[2]{\ensuremath{\mathsf{prophecy}~#1~\mathsf{obs}~#2}\xspace}
    \newcommand{\replace}[1]{\ensuremath{\mathsf{replace}~#1}\xspace}
    
    \newcommand{\pblock}[2]{\ensuremath{\fbox{\ensuremath{#1}}_{#2}}}
    \newcommand{\punblock}[2]{\ensuremath{\dbox{\ensuremath{#1}}_{#2}}}
    \makeatletter 
    \def\arcr{\@arraycr}
    \def\acm@copyrightmode{0}
    \makeatother

    \lstdefinelanguage{csharp}{
        morekeywords={try,public,yield,return,void,foreach,var,in,this,break,catch,while,finally,synchronized,lock},
        sensitive=true,
        basicstyle=\small\ttfamily,
        keywordstyle=\color{blue},
        morecomment=[f]{*},
        morecomment=[s]{/*}{*/},
        morecomment=[n]{/*}{*/},
        commentstyle=\color{teal}\small\ttfamily}
    \lstdefinelanguage{coq}{
        morekeywords={Inductive,Set},
        basicstyle=\small\ttfamily,
        keywordstyle=\color{blue}
    }

\lstdefinelanguage{Scheme}{
  morekeywords=[1]{define, define-syntax, define-macro, lambda, define-stream, stream-lambda},
  morekeywords=[2]{begin, call-with-current-continuation, call/cc, call/comp,
    set!, abort/cc,
    call-with-input-file, call-with-output-file, case, cond,
    do, else, for-each, if,
    let*, let, let-syntax, letrec, letrec-syntax,
    let-values, let*-values,
    and, or, not, delay, force,
    quasiquote, quote, unquote, unquote-splicing,
    define-syntax-rule, raise-syntax-error, define-syntax-parameter, 
    class, super-new, define/public, define/override, super, send,
    map, fold, syntax, syntax-rules, eval, environment, query },
  morekeywords=[3]{import, export, require},
  alsodigit=!\$\%&*+-./:<=>?@^_~,
  sensitive=true,
  morecomment=[s]{\#|}{|\#},
  morestring=[b]",
  basicstyle=\small\ttfamily,
  keywordstyle=\bf\ttfamily\color[rgb]{0,.3,.7},
  commentstyle=\color[rgb]{0.133,0.545,0.133},
  stringstyle={\color[rgb]{0.75,0.49,0.07}},
  upquote=true,
  breaklines=true,
  breakatwhitespace=true,
  literate=*{`}{{`}}{1}
}
    
    \makeatletter
    \let \MathparLineskip \mpr@lesslineskip
    \makeatother

    \ifTR
    \title{Sequential Effect Systems with Control Operators}
    \else
    \title{Lifting Sequential Effects to Control Operators}
    \fi
    \author{Colin S. Gordon}{Drexel University, USA \and \url{https://cs.drexel.edu/~csgordon}}{csgordon@drexel.edu}{https://orcid.org/0000-0002-9012-4490}{}

    \authorrunning{Colin S. Gordon}
    \ifTR
    \hideLIPIcs
    \fi
    \Copyright{Colin S. Gordon}

\ifTR
\else
\EventEditors{Robert Hirschfeld and Tobias Pape}
\EventNoEds{2}
\EventLongTitle{34th European Conference on Object-Oriented Programming (ECOOP 2020)}
\EventShortTitle{ECOOP 2020}
\EventAcronym{ECOOP}
\EventYear{2020}
\EventDate{July 13--17, 2020}
\EventLocation{Berlin, Germany}
\EventLogo{}
\SeriesVolume{166}
\ArticleNo{23}
\relatedversion{An extended version of this paper containing full proofs is available~\cite{controltr} on arXiv at
\url{https://arxiv.org/abs/1811.12285}.}
\fi

\keywords{Type systems, effect systems, quantales, control operators, delimited continuations}
\ccsdesc{Theory of computation~Semantics and reasoning}
\ccsdesc{Theory of computation~Type structures}

\newcommand{\opt}[1]{\ensuremath{\underline{#1}}}
\theoremstyle{remark}
\newtheorem{usecase}{Use Case}

\nolinenumbers
    
\begin{document}

    \maketitle
    \begin{abstract}
    Sequential effect systems are a class of effect system that exploits information about program order, rather than discarding it as traditional commutative effect systems do.  This extra expressive power allows effect systems to reason about behavior over time, capturing properties such as atomicity, unstructured lock ownership, or even general safety properties.
    While we now understand the essential denotational (categorical) models fairly well, application of these ideas to real software is hampered by the variety of source level control flow constructs and control operators in real languages.  
    \looseness=-1

    We address this new problem by appeal to a classic idea: macro-expression of commonly-used programming constructs in terms of control operators.
    We give an effect system for a subset of Racket's tagged delimited control operators, as a lifting of an effect system for a language without direct control operators.
    This gives the first account of sequential effects in the presence of general control operators.
    Using this system, we also re-derive the sequential effect system rules for control flow constructs previously shown sound directly, and derive sequential effect rules for new constructs not previously studied in the context of source-level sequential effect systems.
    This offers a way to directly extend source-level support for sequential effect systems to real programming languages.
    \end{abstract}
    
    \section{Introduction}
    Effect systems extend type systems to reason about not only the shape of data, and available operations --- roughly, what a computation produces given certain inputs --- but to also reason about \emph{how} the computation produces its result.  Examples include 
    ensuring data race freedom by reasoning about what locks a computation assumes held during its execution~\cite{Abadi2006,objtyrace99,boyapati01,boyapati02}, 
    restricting sensitive actions (like UI updates) to dedicated threads~\cite{ecoop13}, 
    ensuring deadlock freedom~\cite{safelocking99,tldi12,Abadi2006,suenaga2008type},
    checking safe region-based memory management~\cite{tofte1997region,lucassen88},
    or most commonly checking that a computation handles (or at least indicates) all errors it may encounter --- Java's checked exceptions~\cite{gosling2014java} are the most widely used effect system.
    
    Most effect systems discard information about program order: the same join operation on a join semilattice of effects is used to overapproximate different branches of a conditional or different subexpressions executed in sequence.
    Despite this simplicity, these traditional \emph{commutative} effect systems (where the combination of effects is always a commutative operation) are powerful.
    Still, many program properties of interest are sensitive to evaluation order.
    For example, commutative effect systems handle scoped \lstinline|synchronized| blocks as in Java with ease: the effect of (the set of locks required by) the \lstinline|synchronized|'s body is permitted to contain the synchronized lock, in addition to the locks required by the overall construct.
    But to support explicit lock acquisition and release operations that are not block-structured, an effect system must track whether a given expression acquires and/or releases locks, and must distinguish their ordering: releasing and then acquiring a given lock is not the same as acquiring before releasing.
    To this end, \emph{sequential} effect systems (so named by Tate~\cite{tate13}) reason about effects with knowledge of the program's evaluation order.
    
    Sequential effect systems are much more powerful than commutative effect systems, with examples extending through generic reasoning about program traces~\cite{skalka2008types,Skalka2008} and even propagation of \emph{liveness} properties from oracles~\cite{Koskinen14LTR} --- well beyond what most type systems support.
    The literature includes sequential effect systems for 
    deadlock freedom~\cite{tldi12,suenaga2008type,Abadi2006,boyapati02},
    atomicity~\cite{flanagan2003atomicity,flanagan2003tldi},
    trace-based security properties~\cite{skalka2008types,Skalka2008},
    safety of concurrent communication~\cite{amtoft1999,nielson1993cml}, general linear temporal properties with a liveness oracle~\cite{Koskinen14LTR}, and more.
    Yet for all the power of this approach, for years each of the many examples of sequential effect systems in the literature 
    individually rederived much structure common to all sequential effect systems.
    Recent years have seen efforts to unify understanding of sequential effect systems with general frameworks, first denotationally~\cite{tate13,katsumata14,mycroft16,atkey2009parameterised}, and recently as an extension to the join semilattice model~\cite{ecoop17}.
    These frameworks can describe the structure of established sequential effect systems from the literature.

    However, these generic frameworks stop short of what is necessary to apply sequential effect systems to real languages: they lack generic treatments of critical features of real languages that interact with evaluation order --- 
    control operators, including established features like exceptions and increasingly common features like generators~\cite{Coyle:1991:BAI:122179.122181}..  And with the exception of the effect systems used to track correct return types with delimited continuations (answer type modification~\cite{danvy1989functional,asai2007polymorphic,KoboriKK16}), there are no sequential effect systems that consider the interaction of control and sequential effects. 
    This means \emph{promising sequential effect systems~\cite{skalka2008types,Skalka2008,Koskinen14LTR,flanagan2003atomicity,tldi12,amtoft1999,nielson1993cml,suenaga2008type,boyapati02}
    cannot currently be applied directly to real languages like Java~\cite{gosling2014java}, Racket, C\#~\cite{csharpgen}, Python~\cite{pythongen}, or JavaScript~\cite{jsgen}}.
    \looseness=-1
    
    Control operators effectively reorder, drop, or duplicate portions of a program's execution at runtime, changing evaluation order.  
    In order to reason precisely about flexible rearrangement of evaluation order, a sequential effect system must reason about control operators.
    The classic example is again Java's \lstinline|try-catch|: if the body of a \lstinline|try| block both acquires and releases a lock this is good, but if an exception is thrown mid-block the release may need to be handled in the corresponding catch.
    Clearly, applying sequential effect systems to real software requires support for exceptions in a sequential effect system.  Working out just those rules is tempting, but exceptions interact with loops.  The effect before a throw inside a loop --- which a catch block may need to ``complete'' (e.g., by releasing a lock) --- depends on whether the throw occurs on the first or $n$th iteration.
    Many languages include more than simply \lstinline|try-catch|, for example with the \emph{generators} (a form of coroutine) now found in 
    C\#~\cite{csharpgen}, Python~\cite{pythongen}, and JavaScript~\cite{jsgen}.
    These interact with exceptions \emph{and} loops.  
    Treating each new control operator individually seems inefficient.

    An alternative to studying all possible combinations of individual control constructs in common languages is to study more general constructs, such as the very general delimited continuations~\cite{FelleisenF87,Felleisen88} present in Racket. These are useful in their own right (for Racket, or the project to add them to Java~\cite{projectloom}), and can macro-express many control flow constructs and control operators of interest, including loops, exceptions~\cite{Flatt2007}, coroutines~\cite{DBLP:conf/lfp/HaynesFW84,DBLP:journals/cl/HaynesFW86}, generators~\cite{Coyle:1991:BAI:122179.122181}, and more~\cite{danvy2006analytical}.
    Then general principles can be derived for the general constructs, which can then be applied to or specialized for the constructs of interest.
    This both solves the open question of how to treat general control operators with sequential effect systems, and leads to a basis for more compositional treatment of loops, exceptions, generators, and future additions to languages.  This is the avenue we pursue in this paper.
    \looseness=-1

    Delimited continuations solve the generality problem, but introduce new challenges since sequential effect systems can track evaluation order~\cite{skalka2008types,Skalka2008,quantalesjournal,Koskinen14LTR}.
    The effect of an expression that aborts out of a prompt depends on what was executed before the abort, but not after.  The body of a continuation capture (\lstinline|call/cc|) must be typed knowing the effect of the \emph{enclosing context} --- the code executing after, but not before (up to the enclosing prompt).  
    We lay the groundwork for handling modern control operators in a sequential effect system:
    \looseness=-1
    \begin{itemize}
        \item We give the first generic characterization of sequential effects for continuations, by giving a \emph{generic} lifting of a control-unaware sequential effect system into one that can support tagged delimited continuations.
        The construction we describe provides a way to automatically extend existing systems with support for these constructs, and likewise will permit future sequential effect system designers to ignore control operations initially and add support later for free (by applying our construction).
        As a consequence, we can transfer prior sequential effect systems designed \emph{without} control operators to a setting \emph{with} control operators.
        \item We give sequential effect system rules for \lstinline|while| loops, \lstinline|try-catch|, and generators by deriving them from their macro-expression~\cite{felleisen1991expressive} in terms of more primitive operators.  The loop characterization was previously known (and technically a control flow construct, not a general control operator), but was given as primitive.  The others are new to our work, and necessary developments in order to apply sequential effect systems to most modern programming languages. The derivation approach we describe can be applied to other control operators that are not explicitly treated in the paper.
        \item We demonstrate how prior work's notion of an iteration operator~\cite{ecoop17,quantalesjournal} derived from a closure operator on the underlying effect lattice is not specific to loops, but rather provides a general tool for solving recursive constraints in sequential effect systems.
        \item We prove syntactic type safety for a type system using our sequential control effect transformation with any underlying effect system.
    \end{itemize}

    \section{Background}
    \label{sec:bg}
    We briefly recall the details of 
    \ifTR
    standard type-and-effect systems, 
    \fi
    sequential type-and-effect systems, and tagged delimited continuations.  We emphasize the view of effect systems in terms of a \emph{control flow algebra}~\cite{mycroft16} --- an algebraic structure with operations corresponding to the ways an effect system might combine the effects from subexpressions in a program.
    
\ifTR
\subsection{Effect Systems}
    Traditional type-and-effect systems extend the typing judgment $\Gamma\vdash e : \tau$ for an additional component.  The extended judgment form $\Gamma\vdash e : \tau \mid \chi$ is read ``under local variable assumptions $\Gamma$, the expression $e$ evaluates to a value of type $\tau$ (or diverges), with effect $\chi$ during evaluation.''
    The last clause of that reading is vague, but carries specific meanings for specific effect systems.  For checked exceptions, it could be replaced by ``possibly throwing exceptions $\chi$ during evaluation'' where $\chi$ would be a set of checked exceptions.  For a data race freedom type system reasoning about lock ownership, it could be replaced by ``and is data race free if executed while locks $\chi$ are held.''
    
    In traditional effect systems the set of effects tracked forms a join semilattice: a partial order with a (binary) least-upper bound operation (join, written $\sqcup$), and a least element $\bot$.  As is standard for join semilattices, $\sqcup$ is commutative and associative.  Any time effects of subexpressions must be combined, they are mixed with this join.
    Functions introduce an additional complication that requires modifying function types: the effect of a function's body does not occur when a function (e.g., a lambda expression) is evaluated, but only when it is applied.  So the effect of a function expression itself (like other values) may simply be bottom.
    A function type then carries the \emph{latent effect} of the body --- the effect that does not ``happen'' until the function is actually invoked.
    For example, consider checked exceptions in Java.  The allocation of a class instance (such as what a lambda allocation there translates to) does not actually run any method(s) of the class --- invocation does.  So allocating a class instance throws no exceptions (assuming the constructor throws no exceptions).  But invoking a method with a \lstinline|throws| clause may --- the \lstinline|throws| clause is the latent effect of the method for Java's checked exceptions. 
    To make this more explicit, let us consider the standard type rules for lambda expressions and function application in a generic effect system:
    \begin{mathpar}
    \inferrule*[left=T-Lambda]{
        \Gamma,x:\tau \vdash e : \sigma \mid \chi
    }{
        \Gamma\vdash (\lambda x\ldotp e) : \tau\overset{\chi}{\rightarrow}\sigma \mid \bot
    }
    \and
    \inferrule*[left=T-App]{
        \Gamma\vdash e_1 : \tau\overset{\chi}{\rightarrow}\sigma \mid \chi_1\\
        \Gamma\vdash e_2 : \tau \mid \chi_2\\
    }{
        \Gamma\vdash e_1\;e_2 : \sigma \mid \chi_1\sqcup\chi_2\sqcup\chi
    }
    \end{mathpar}

    Key to note here are that the lambda expression's type carries the latent effect of the function body, but itself has only the bottom effect; and that when a function is applied, the overall effect of the expression is the combination (via join) of all subexprssions' effects \emph{and} the latent effect of the function.
    We call these effect systems \emph{commutative} not only to distinguish them from the broader class of systems we study in this paper, but also because all combinations of effects in such systems are commutative, and disregard evaluation order --- the only means to combine an effect in a commutative effect system is with the (commutative) join operation.  Other rules with multiple subexpressions, such as while loops, conditions, and more, similar join effects without regard to program order or repetition.
    
    By contrast, many effect systems use a richer structure to reason about cases where evaluation order is important.  
\fi
    This includes effect systems for atomicity~\cite{flanagan2003atomicity,flanagan2003tldi}, deadlock freedom~\cite{suenaga2008type,tldi12,boyapati02,Abadi2006}, race freedom with explicit lock acquisition and release~\cite{suenaga2008type,ecoop17}, message passing concurrency safety~\cite{nielson1993cml,amtoft1999}, security checks~\cite{skalka2008types}, and (with the aid of an oracle for liveness properties) general linear-time properties~\cite{Koskinen14LTR}.
    Tate labels these systems \emph{sequential} effect systems~\cite{tate13}, as their distinguishing feature is the use of an additional sequencing operator to join effects where one is known to be evaluated before another.  Consider the sequential rules for functions, function application, conditionals, and while loops:
    \looseness=-1

    \vspace{-1em}
    \begin{mathpar}
    \inferrule[T-App]{
        \Gamma\vdash e_1 : \tau\overset{\chi}{\rightarrow}\sigma \mid \chi_1\\
        \Gamma\vdash e_2 : \tau \mid \chi_2\\
    }{
        \Gamma\vdash e_1\;e_2 : \sigma \mid \chi_1\rhd\chi_2\rhd\chi
    }
    \and
    \inferrule[T-While]{
        \Gamma\vdash e_c : \mathsf{boolean} \mid \chi_c\\
        \Gamma\vdash e_b : \tau \mid \chi_b\\
    }{
        \Gamma\vdash \mathsf{while}\;e_c\;e_b : \mathsf{unit} \mid \chi_c\rhd(\chi_b\rhd\chi_c)^*
    }
    \and
    \inferrule[T-Lambda]{
        \Gamma,x:\tau \vdash e : \sigma \mid \chi
    }{
        \Gamma\vdash (\lambda x\ldotp e) : \tau\overset{\chi}{\rightarrow}\sigma \mid I
    }
    \and
    \inferrule[T-If]{
        \Gamma\vdash e_c : \mathsf{bool} \mid \chi_c\\
        \Gamma\vdash e_t : \tau \mid \chi_t\\
        \Gamma\vdash e_f : \tau \mid \chi_f
    }{
        \Gamma\vdash \mathsf{if}~e_c~e_t~e_f : \tau \mid \chi_c\rhd(\chi_t\sqcup\chi_f)
    }
    \end{mathpar}
    The sequencing operator $\rhd$ is associative but \emph{not} (necessarily) commutative.  Thus the effect in the new \textsc{T-App} reflects left-to-right evaluation order: first the function position is reduced to a value, then the argument, and then the function body is executed.
    The conditional rule reflects the execution of the condition followed by \emph{either} (via commutative join) the true or false branch.
    The while loop uses an iteration operator $(-)^*$ to represent 0 or more repetitions of its argument; we will return to its details later.  The effect of \textsc{T-While} reflects the fact that the condition will always be executed, followed by 0 or more repetitions of the loop body and checking the loop condition again.
    The rule for typing lambda expressions switches from a bottom element, to a general unit effect: identity for sequential composition.

    To formalize the intuition above,
    Gordon~\cite{ecoop17} proposed \emph{effect quantales} as a model that captures prior effect systems' structure:
    
    \begin{definition}[Effect Quantale]
        \label{def:eq}
        An \emph{effect quantale} is a join-semilattice-ordered monoid with nilpotent top.  That is, it is a structure $(E,\sqcup,\top,\rhd,I)$ where:
        \begin{itemize}
        \item $(E,\sqcup,\top)$ is an upper-bounded join semilattice
        \item $(E,\rhd,I)$ is a monoid
        \item $\top$ is nilpotent for sequencing ($\forall x\ldotp x\rhd\top=\top=\top\rhd x$)
        \item $\rhd$ distributes over $\sqcup$ on both sides: $a\rhd(b\sqcup c)=(a\rhd b)\sqcup(a\rhd c)$ and
        $(a\sqcup b)\rhd c = (a\rhd c)\sqcup(b\rhd c)$
        \end{itemize}
    \end{definition}
    The structure extends a join semilattice with a sequencing operator, a designated error element to model possibly-undefined combinations, and laws specifying how the operators interact.
    Top ($\top$) is used as an indication of a type error, for modeling partial join or sequence operators: expressions with effect $\top$ are rejected. $\sqcup$ is used to model non-deterministic joins (e.g., for branches) as in the commutative systems, and $\rhd$ is used for sequencing.  The default effect of ``uninteresting'' program expressions (including values) becomes the unit $I$ rather than a bottom element (which need not exist).
    As a consequence of the distributivity laws, it follows that $\rhd$ is also monotone in both arguments, for the standard partial order derived from a join semilattice: $x\sqsubseteq y \equiv x\sqcup y = y$.

    Gordon~\cite{ecoop17} also showed how to exploit \emph{closure operators}~\cite{birkhoff,blyth2006lattices,fuchs2011partially} to impose a well-behaved notion of iteration (the $(-)^*$ operator from \textsc{T-While}) that coincides with manually-derived versions for the effect quantales modeling prior work for many effect quantales.
    Gordon~\cite{quantalesjournal} recently generalized the construction, and showed that large general classes of effect quantales meet the criteria to have such an iteration operator.
    The effect quantales for which the generalized iteration is defined are called \emph{laxly iterable}.
    An effect quantale is \emph{laxly iterable} if for every element $x$, the set of subidempotent elements ($\{s \mid s\rhd s\sqsubseteq s\}$) greater than both $x$ and $I$ has a least element.
    This is true of all known effect quantales corresponding to systems in the literature.
    \looseness=-1

    The iteration operator for an iterable effect quantale takes each effect $x$ to the least subidempotent effect greater than or equal to $x\sqcup I$ (which exists, by the definition of laxly iterable).
    This iteration operator satisfies 5 essential properties for any notion of iteration~\cite{quantalesjournal}, which we will find useful when deriving rules for loops.
    Iteration operators are \emph{extensive} ($\forall e\ldotp e \sqsubseteq e^*$), \emph{idempotent} ($\forall e\ldotp (e^*)^*=e^*$), \emph{monotone} ($\forall e,f\ldotp e\sqsubseteq f \Rightarrow e^*\sqsubseteq f^*$), \emph{foldable} ($\forall e\ldotp e\rhd e^* \sqsubseteq e^*~\textrm{and}~e^*\rhd e\sqsubseteq e^*$), and \emph{possibly-empty} ($\forall e\ldotp I \sqsubseteq e^*$).
    Another useful property of iteration that we will sometimes use is that $\forall x,y\ldotp x^*\sqcup y^*\sqsubseteq (x\sqcup y)^*$\ifTR\footnote{$x^*\sqcup y^*\sqsubseteq (x\sqcup y)^*\sqcup y^*$ by monotonicity and $x\sqsubseteq x\sqcup y$, so for $y$ similarly $\ldots\sqsubseteq (x\sqcup y)^*\sqcup(x\sqcup y)^*=(x\sqcup y)^*$}\fi. 
    Gordon~\cite{ecoop17,quantalesjournal} gives more details on closure operators and the derivation of iteration.  We merely require its existence and properties.
    \looseness=-1

    For our intended goal of giving a transformation of any arbitrary sequential effect system into one that can use tagged delimited continuations, we require \emph{some} abstract characterization. We choose effect quantales as the abstraction for lifting for several reasons.  First, they characterize the structure of a range of concrete systems from prior work~\cite{ecoop17,quantalesjournal}, while other proposals omit structure that is important to these concrete systems. Second, while effect quantales are not maximally general, they remain very general: the motivating example for Tate's work~\cite{tate13} (which \emph{is} maximally general) can be modeled as an effect quantale.  Third, we would like to check whether our derived rules are sensible; effect quantales are the only abstract characterization for which imperative loops have been investigated, offering appropriate points of comparison.  Finally, the iteration construction on effect quantales offers a natural approach to solving recursive constraints on effects, which we will use in deriving closed-form derived rules for macro-expressed control flow constructs and control operators.
    \looseness=-1

\ifTR
    Gordon~\cite{ecoop17} gives an effect quantale for enforcing data race freedom in the presence of unstructured locking, where the elements are pairs of multisets of locks --- a multiset counting (recursive) lock aquisitions assumed before an expression, and a multiset counting (recursive) lock claims after an expression.  Using pre- and post-multisets rather than simply tracking acquisitions and releases makes it possible to enforce data race freedom: an atomic action (e.g., field read) accessing data guarded by a lock $\ell$ can have effect $(\{\ell\},\{\ell\})$, which indicates the guarding lock must be held before and remains held.
    The idempotent elements for this effect quantale are the error element, plus those where the pre- and post-multisets are the same --- so iterating an expression where the lock claims are loop-invariant does nothing, while iterating an action that acquires and/or releases locks (in aggregate, not internally) yields an error.
\fi

As a running example throughout the paper, we will use
    a simplification of various trace or history effect systems~\cite{skalka2008types,Skalka2008,Koskinen14LTR}.  For a set (alphabet) of events $\Sigma$, consider the non-empty subsets of $\Sigma^*$ --- the set of possibly-empty strings of letters drawn from $\Sigma$ (the strings, not the subsets, may be empty). 
    This gives an effect quantale $\mathcal{T}(\Sigma)$ whose elements are these subsets or an additional top-most error element $\mathsf{Err}$.  Join is simply set union lifted to propagate $\mathsf{Err}$.  Sequencing is the double-lifting of concatenation, first to sets ($A\cdot B=\{xy \mid x\in A\land y\in B\}$), then again to propagate $\mathsf{Err}$.  The unit for sequencing is the singleton set of the empty string, $\{\epsilon\}$.
    If $\Sigma$ is a set of events of interest --- e.g., security events --- then effects drawn from this effect quantale represent sets of possible finite event sequences executed by a program.
    Effects drawn from this effect quantale show the possible sequences of operations code may execute, which will allow us to show explicitly how fragments of program execution are rearranged when using control operators.
    For our examples, we will assume a family of language primitives $\mathsf{event}[\alpha]$ with effect $(\emptyset,\emptyset,\{\alpha\})$ (similar to Koskinen and Terauchi~\cite{Koskinen14LTR}), where $\alpha$ is drawn from a set $\Sigma$ of possible events.
    The key challenge we face in this paper is, viewed through the lens of $\mathcal{T}(\Sigma)$, to ensure that when continuations are used, the effect system does not lose track of events of interest or falsely claim a critical event occurs where it may not.
    \looseness=-1

\subsection{Tagged Delimited Continuations}
Control operators have a long and rich history, reaching far beyond what we discuss here.
     Many different control operators exist, and many are macro-expressible~\cite{felleisen1991expressive} in terms of each other (i.e., can be translated by direct syntactic transformation into another operator), though some of these translations require the assumption of mutable state, for example.  
    But a priori there is no single most general construct to study which obviously yields insight on the source-level effect typing of other constructs.
    A suitable starting place, then, is to target a highly expressive set of operators that see use in a real language.
    If the operators are sufficiently expressive, this provides not only a sequential type system for an expressive source language directly, but also supports deriving type rules for \emph{other} languages' control constructs, based on their macro-expression in terms of the studied control operators.
    
    We study a subset of the tagged delimited control operators~\cite{FelleisenF87,Felleisen88,Sitaram1993,SitaramF90a,SitaramF90b} present in Racket~\cite{Flatt2007}, shown in 
\ifTR
    Figure \ref{fig:opsem}.\footnote{Racket afficianados familiar with Flatt et al.'s work may skip ahead while noting we omit continuation marks and \lstinline|dynamic-wind|, deferring these to future work.
    Continuation marks are little-used outside Racket.  \lstinline|dynamic-wind| is the heart of constructs like Java's \lstinline|finally| block or the unconditional lock release of a \lstinline|synchronized| block, but we leave them to future work.
    Composable continuations extend their context, rather than replacing it.  We give their semantics below for comparison and because they are necessary for completeness in some models~\cite{SitaramF90a,SitaramF90b}, but leave their treatment in a sequential effect system to future work as well.}
\fi
    The semantics include both local ($\rightarrow$) and global ($\Rightarrow$) reductions on configurations consisting of a state $\sigma$ and expression $e$.
    All continuations in Racket are delimited, and tagged.  There is a form of \emph{prompt} that limits the scope of any continuation capture:
    \lstinline|(% tag e e2)| is a tagged prompt with tag \lstinline|tag|, body \lstinline|e|, and abort handler \lstinline|e2|.  
    Without tags, different uses of continuations --- e.g., error handling or concurrency abstractions --- can interfere with each other~\cite{Sitaram1993}; as a small example, if loops and exceptions were both implemented with \emph{un}delimited continuations, throwing an exception from inside a loop inside a try-catch would jump to the loop boundary, not the catch.  Thus prompts, the continuation-capturing primitives \lstinline|call/cc| and \lstinline|call/comp|, and the \lstinline|abort| primitive all specify a tag, and only prompts with the specified tag are used to interpret continuation and abort boundaries. This permits jumping over unrelated prompts (e.g., so exceptions find the nearest \emph{catch}, not merely the nearest control construct).
    In most presentations of delimited continuations, tags are ignored (equivalently, all tags are equal), while most implemenations retain them for the reasons above.  Here the tags are essential to the theory as well: an abort that ``skips'' a different prompt must be handled differently by our type-and-effect system.
    
\begin{figure*}\small
\[
    E ::= \bullet \mid (E~e) \mid (v~E) \mid (\%~t~E~v) \mid (\mathsf{call/cc}~t~E) \mid (\mathsf{call/comp}~t~E) \mid (\mathsf{abort}~t~e)
\]
\begin{mathpar}
\fbox{$\sigma;e\overset{q}{\rightarrow} \sigma;e$}
\quad
\inferrule*[left=\small E-App]{ }{
    \sigma;((\lambda x\ldotp e)~v)\overset{I}{\rightarrow}\sigma;e[v/x]
}
\ifTR
\and
\inferrule*[left=\small E-PrimApp]{
    \delta(\sigma,p~\overline{e})=(\sigma',v,\chi)\\
    \mathsf{Values}(\overline{e})
}{
    \sigma;(p~\overline{e})\overset{\chi}{\rightarrow}\sigma',v
}
\fi
\quad
\inferrule*[left=\small E-PromptVal]{ }{
    \sigma;(\%~\ell~v~h)\xRightarrow{I}\sigma;v
}
\\
    \fbox{$\sigma;e\overset{q}{\Rightarrow}\sigma;e$}
    \quad
    \inferrule*[left=\small E-Context]{ \sigma;e\overset{q}{\rightarrow}\sigma';e'}{\sigma;E[e]\overset{q}{\Rightarrow}\sigma';E[e']}
    \quad
    \inferrule*[left=\small E-Abort]{ E'~\textrm{contains no prompts for}~\ell }{
        \sigma;E[(\%~\ell~E'[(\mathsf{abort}~\ell~v)]~h)]
        \xRightarrow{I}
        \sigma;E[h~v]
    }
    \and
    \inferrule*[left=\small E-CallCC]{ E'~\textrm{contains no prompts for}~\ell }{
        \sigma;E[(\%~\ell~E'[(\mathsf{call/cc}~\ell~k)]~h)]
        \xRightarrow{I}
        \sigma;E[(\%~\ell~E'[(k~(\mathsf{cont}~\ell~E'))]~h)]
    }
    \and
    \inferrule*[left=\small E-InvokeCC]{ E'~\textrm{contains no prompts for}~\ell }{
        \sigma;E[(\%~\ell~E'[((\mathsf{cont}~\ell~E'')~v)]~h)]
        \xRightarrow{I}
        \sigma;E[(\%~\ell~E''[v]~h)]
    }
\end{mathpar}
        \vspace{-2em}
\caption{Operational semantics
        \vspace{-1em}
}
\label{fig:opsem}
\end{figure*}
    \lstinline|call/cc tag f| is the standard (delimited) call-with-current-continuation: \lstinline|f| is invoked with a delimited continuation representing the current continuation up to the nearest prompt with tag \lstinline|tag| (\textsc{E-CallCC}).  Invoking that continuation (\textsc{E-InvokeCC}) replaces the context up to the nearest dynamically enclosing prompt with the same tag, leaving the delimiting prompt in place.
    Both capture and replacement are bounded by the nearest enclosing prompt for the specified tag.
    The surrounding captured or replaced context ($E'$ in both rules) may contain prompts for other tags, but not the specified tag.
    Racket also includes \lstinline|(abort t e)| (absent in many formalizations of continuations), which evaluates \lstinline|e| to a value, then replaces the enclosing prompt (of the specified tag \lstinline|t|) with an invocation of the handler applied to that value (\textsc{E-Abort}).
    Racket's rules differ from some uses of \lstinline|abort| in the literature. Figure \ref{fig:opsem}'s rules are Flatt et al.'s rules~\cite{Flatt2007} without continuation marks and \lstinline|dynamic-wind|.  Flatt et al.\ formalized Racket's control operators in Redex~\cite{felleisen2009semantics,klein2012redex}, including showing they passed the Racket implementation tests for those features.  We have verified the rules above continue to pass the relevant tests in Redex (see supplementary material~\cite{controltr}).
    \looseness=-1

    We chose this set of primitives, over related control operators~\cite{shan2007static} such as \lstinline|shift/reset| or \lstinline|shift0/reset0| which can simulate these primitives, for several reasons.  First, they are general enough to use for deriving rules for higher-level constructs like generators from their macro-expansion.  Second, the control operators we study are implemented as primitives in a real, mature language implementation (Racket), used in real software~\cite{KrishnamurthiHMGPF07}.  And finally, it is known~\cite{Flatt2007} how these control operators interact with other useful control operators like \lstinline|dynamic-wind|~\cite{Haynes:1987:ECP:29873.30392} (relevant to \lstinline|finally| or \lstinline|synchronized| blocks) and continuation marks~\cite{clements2001modeling}.  Thus our Racket subset is a suitable basis for future extension, while we are unaware of established extensions of \lstinline|shift/reset|, \lstinline|shift0/reset0|, etc. with continuation marks or \lstinline|dynamic-wind|.
    \looseness=-1

    The operators we study can express loops, exceptions, coroutines~\cite{DBLP:conf/lfp/HaynesFW84,DBLP:journals/cl/HaynesFW86}, and generators~\cite{Coyle:1991:BAI:122179.122181}.
    Racket also includes \emph{compositional} continuations, whose application extends the current context rather than discarding it, giving completeness with respect to some denotational models~\cite{SitaramF90b}, and alleviating space problems when using \lstinline|call/cc| to simulate other families of control operators (it is known to macro-express another popular form of delimited continuations, the combination of \lstinline|shift| and \lstinline|reset|~\cite{Flatt2007}).
    \ifTR
    We excise discussion of compositional continuations to Appendix \ref{apdx:compositional}.
    \fi

    One final point about the semantics worth noting is the presence of effect annotations on the reduction arrows.
    These semantics are further adapted from Flatt et al.~\cite{Flatt2007} to ``emit'' the primitive effect of the reduction, which is typical of syntactic type safety proofs for effect systems, including ours (Section \ref{sec:summary_soundness}).  They do not influence evaluation, but only mark a relationship between the reduction rules and static effects.
    \ifTR
    Non-unit (non-$I$) effects arise from primitives, via \textsc{E-PrimApp}, which is taken from Gordon's parameterized system~\cite{ecoop17}; it assumes a function $\delta$ giving state transformer semantics and a primitive effect to raise $\chi$ when primitives are applied to a correct number of values.
    \fi
    
    \section{Growing Sequential Effects: Control, Prophecies, and Blocking}
    \label{sec:usecases}
    To build intuition for our eventual technical solution, we motivate its components through a series of progressively more sophisticated use cases.
    We will use $\mathcal{T}(\Sigma)$ in all of our examples, because we find traces to be an effective way of explaining the difficulties the effect system must address related to program fragments (i.e., events) being repeated, skipped, or reordered.  A reader may choose to impart security-specific meanings to these events (as Skalka et al.~\cite{skalka2008types} do) or as any other protocol of personal interest (e.g., lock acquisition and release).  However, our development in Section \ref{sec:formal} is \emph{not} specific to this effect quantale, but instead parameterized over an arbitrary effect quantale.
    Our goal is to develop a sequential effect system based on transforming an \emph{underlying} base effect quantale $Q$ into a structure we will call $\mathcal{C}(Q)$ with sequencing and join operations, and a unit effect.
    This ensures our transformation works for \emph{any} valid effect quantale, which includes all sequential effect systems we are aware of~\cite{ecoop17,quantalesjournal}.

    \begin{usecase}[Control-Free Programs] Since programs are not required to use control operators, our solution must include a restriction equivalent to the class of underlying effects to reason about.  For example, if 
    \lstinline|event[$\alpha$]| has effect $\{\alpha\}\in\mathcal{T}(\Sigma)$ and 
    \lstinline|event[$\beta$]| has effect $\{\beta\}\in\mathcal{T}(\Sigma)$,  we should expect the effect of \lstinline|event[$\alpha$]; event[$\beta$]| to be somehow equivalent to sequencing those underlying effects --- $\{\alpha\}\rhd\{\beta\}=\{\alpha\beta\}$.  This suggests underlying effects should be at least a component of continuation-aware effects.
    \end{usecase}

    \begin{usecase}[Aborting Effects]
        \label{use:basicabort}
         The simplest control behavior we can use is to abort to a handler, and this interacts with both sequencing and conditionals. Consider:
    \begin{lstlisting}
        (% t ((if c (event[$\alpha$]) (abort t 3)); event[$\beta$]) ($\lambda$n. event[$\gamma$]) )
    \end{lstlisting}
    Assuming \lstinline|c| is a variable (i.e., pure), there are two paths through this term:
    \begin{itemize}
        \item If \lstinline|c| is true, the code will emit events $\alpha$ and $\beta$ (in order), and not execute the handler
        \item If \lstinline|c| is false, the code will abort to the handler, which will emit event $\gamma$.
    \end{itemize}
    So intuitively, the effect for this term should contain those traces; ideally the effect would be $\{\alpha\beta,\gamma\}$, containing only those traces. For an effect system to validate this effect for this term, it must not only track ordinary underlying effects, but also two aspects of the \lstinline|abort| operation's behavior: it causes some code to be discarded (the \lstinline|event[$\beta$]|), and it causes the handler to run. We can track this information by making effects pairs of two components: a set of behaviors up to an abort (which we will call the \emph{control effect set}, since it tracks effects due to non-local control transfer), and an underlying effect for when no abort occurs.  We could then give the body of the prompt the effect $(\{\mathsf{abort}(\{\epsilon\})\},\{\alpha\beta\})$ to indicate that it either executes normally producing a trace $\alpha\beta$, or it aborts to the nearest handler after doing nothing (more precisely, after performing actions with the unit effect for $\mathcal{T}(\Sigma)$). The type rule for prompts can then recognize the body may abort, and for each possible prefix of an aborting execution, add into the overall (underlying) effect of the prompt the result of sequencing that prefix with the \emph{handler}'s effect: here, $(\{\epsilon\}\rhd\{\gamma\})\sqcup\{\alpha\beta\}=\{\alpha\beta,\gamma\}$.

    Above the body effect was given as a whole from intuition, but in general this must be built compositionally from subexpression effects, motivating further questions.
    First, what is the effect of the subterm \lstinline|abort t 3|? In particular, what is its underlying effect?  One sound choice would be $I$ ($\{\epsilon\}$ for $\mathcal{T}(\Sigma)$), but this would introduce imprecision: it would produce effects suggesting it was possible to execute only event $\beta$ (since the conditional's underlying effect would include both $\alpha$ and the empty trace $\epsilon$, sequenced with $\beta$).  Instead, we will make the underlying component \emph{optional}, writing $\bot$ when it is absent.  We will continue to use metavariables $Q$ to indicate a definitely present element of the underlying effect quantale, but will use the convention of underlining metavariables (e.g., $\opt{Q}$) when they may be $\bot$.  This permits the conditional's underlying effect to only contain the trace $\alpha$, because joining the branches' effects can simply ignore the missing underlying effect from the aborting branch.

    Second, we should consider how the sequencing and join operations interact with abort effects. While component-wise union/join is a natural (and working) starting point, sequencing is less obvious. 
    Prefixing the last example's body with an extra event is instructive:
    \begin{lstlisting}
        (% t (event[$\delta$]; (if c (event[$\alpha$]) (abort t 3)); event[$\beta$]) ($\lambda$n. event[$\gamma$]) )
    \end{lstlisting}
    Execution could generate two traces: $\delta\alpha\beta$ and $\delta\gamma$. So \emph{both} traces in the effect for the last body should gain this $\delta$ prefix: not only the underlying effect component, but also the portion related to the \lstinline|abort|. This suggests the following definition of sequencing:\\
    \centerline{$(C_1,\opt{Q_1})\rhd(C_2,\opt{Q_2}) = (C_1\cup(\opt{Q_1}\rhd C_2), \opt{Q_1}\rhd \opt{Q_2})$}
    Assuming there is a lifting (given later) of the underlying sequencing to possibly-absent underlying effects ($\opt{Q_1}\rhd\opt{Q_2}$), this reflects the natural ways of combining paths through these terms: the \emph{control effect set} collecting abort behaviors should include both the abort behaviors from $C_1$ (an execution corresponding to one of those behaviors means nothing from the second effect will execute), as well as the result of first executing \emph{normal} behaviors of the first effect (e.g., $\{\delta\}$) before the aborting behaviors of the second --- $\opt{Q_1}\rhd C_2$. So the effect of this new example's prompt body should be the join of the two branches ($(\{\mathsf{abort}(\{\epsilon\})\},\bot)\sqcup(\emptyset,\{\alpha\})=(\{\mathsf{abort}(\{\epsilon\})\},\{\alpha\})$), sequenced between the effects of event $\delta$ and event $\beta$: 
    $(\emptyset,\{\delta\})\rhd (\{\mathsf{abort}(\{\epsilon\})\},\{\alpha\})\rhd(\emptyset,\{\beta\})=(\emptyset,\{\delta\})\rhd (\{\mathsf{abort}(\{\epsilon\})\},\{\alpha\beta\})=(\{\mathsf{abort}(\{\delta\})\}, \{\delta\alpha\beta\})$.
    Repeating our informal prompt handling above gives us the new expected underlying effect $\{\delta\alpha\beta,\delta\gamma\}$.
    We refer to the way the control effect set accumulates prefixes on the left as \emph{left-accumulation} of effects.
    This definition is associative, and distributes over the component-wise union/join.
    For now we continue with the simplifying assumption that all tags are \lstinline|t| (equivalent to \emph{untagged} continuations), and consider issues with multiple tags in Use Case \ref{use:tags}.
    A final detail deferred to Section \ref{sec:formal} is that the value thrown by an \lstinline|abort| must be of the type expected by the corresponding handler.
    \looseness=-1

    This is already part-way to our goal of deriving type rules for common control operators from delimited continuations:
    the work so far supports basic checked exceptions: 
    \\\centerline{$
        \llbracket\mathsf{try}\;e\;\mathsf{catch}\;{C\Rightarrow e_c}\rrbracket = (\%~C~\llbracket{e}\rrbracket~\llbracket{e_c}\rrbracket)
        \qquad
        \llbracket\mathsf{throw}_C e\rrbracket = (\mathsf{abort}\;t_C\;\llbracket{e}\rrbracket)
    $}
    The formal rules for prompts and aborts will directly dictate the rules for these macro-expressions of checked exceptions, just as the informal discussion here transfers directly to these macro-expressed checked exceptions.
\end{usecase}

\begin{usecase}[{Invoking Simple Continuations}]
    \label{use:basiccallcont}
    Operationally, invoking an existing continuation is similar to using \lstinline|abort| in that it discards the surrounding context up to the nearest prompt.
    But unlike uses of \lstinline|abort|, invoking a continuation does not cause handler execution --- instead (by \textsc{E-InvokeCC}) the prompt and handler remain in place.
    For the type rule for prompts to treat this additional mechanism, effects must indicate that invocation may occur.  We can do this by extending effects slightly: instead of $C$ in the effect $(C,\opt{Q})$ only containing prefixes of aborting computations, it should also include prefixes of continuation invocations --- this is why elements of $C$ were labeled $\mathsf{abort}(-)$ in the previous example, and we can now include effects tagged by $\mathsf{replace}(-)$ to indicate control behaviors that \emph{replace} the current continuation with a new one.  
    Considering a term invoking a continuation \lstinline|k|:
    \begin{lstlisting}
        (% t (event[$\delta$]; (if c (event[$\alpha$]) (k ())); event[$\beta$]) ($\lambda$n. event[$\gamma$]) )
    \end{lstlisting}
    As in the previous example, one possible trace of this program is $\delta\alpha\beta$ when \lstinline|c| is true.  When \lstinline|c| is false, however, this program will emit event $\delta$, then invoke \lstinline|k|, replacing the rest of the body with \lstinline|k|'s captured continuation (with unit in its hole).  
    Thus typing this requires also knowing additional information about \lstinline|k|, not required in the \lstinline|abort| case.  We require continuations to carry a latent effect similar to functions: while the latent effect of a function $(\lambda x. e)$ describes the effect of the term obtained by substituting an appropriately-typed value $v$ into $e$ --- the effect of $e[v/x]$ --- the latent effect of a continuation $(\mathsf{cont}^\tau~\ell~E)$ describes the effect of plugging an appropriately-typed value $v$ into $E$'s hole --- the effect of $E[v]$.

    Assuming \lstinline|k| has only underlying effects --- e.g., if $k=(\mathsf{cont}^\mathsf{unit}~t~(\bullet; \mathsf{event}[\eta]))$, then we should expect the type rule for invocation to take that underlying effect from \lstinline|k|'s latent effect and move it to the control effect set, under $\mathsf{replace}(-)$.  
    So if \lstinline|k|'s latent effect is $(\emptyset,\{\eta\})$, the effect of \lstinline|(k ())| should be $(\{\mathsf{replace}(\{\eta\})\},\bot)$.
    Sequential composition should be extended to treat \textsf{replace} control effects similar to \textsf{abort} control effects, by accumulating on the left.  With that adjustment, we can conclude the body above has effect $(\{\mathsf{replace}(\{\delta\eta\})\},\{\delta\alpha\beta\})$.
    The type rule for prompts can treat these similarly to \textsf{abort} control effects, but without sequencing them with the handler (which doesn't execute in this case), producing an overall effect $(\emptyset,\{\delta\alpha\beta,\delta\eta\})$.
    Notice that this is slightly different from $\mathsf{abort}(-)$ control effects: those track \emph{only} the prefix effect before the abort, but the approach just outlined would have $\mathsf{replace}(-)$ control effects include prefix effects before invoking the continuation \emph{and} the underlying effect of behaviors \emph{after} the control operation. 
    This corresponds to the location of the behavior that executes after the control operation --- remote for aborts (the handler is non-local) and local for continuation invocation (the continuation is at the call site).
\end{usecase}

\begin{usecase}[Invoking Continuations that Abort or Invoke Other Continuations]
    \label{use:nestedcont}
    The example above assumed \lstinline|k| had only underlying effects, but in general \lstinline|k|'s body might use \lstinline|abort| or invoke other continuations.  In such cases, \lstinline|k|'s latent effect would be some pair $(C,\opt{Q})$ for non-empty control effect set $C$.  It turns out simply treating those na\"ively --- including them into the control effect set for invoking \lstinline|k| --- is adequate for now (we revisit this when considering multiple tags).  If $k=(\mathsf{cont}^\mathsf{unit}~t~(\bullet; \mathsf{event}[\eta]; \mathsf{abort}~t~3))$ in the previous example, then the latent effect will be $(\{\mathsf{abort}(\{\eta\})\},\bot)$, and simply making the effect of \lstinline|(k ())| be the same (dropping the absent underlying latent effect since the continuation can only return via control operators, but making the application's underlying effect $\bot$ because by definition it does not return directly) gets the expected result at the prompt, including the trace $\delta\eta\gamma$ from emitting $\delta$, invoking \lstinline|k|, emitting $\eta$ (from \lstinline|k|'s restored body) and aborting to the handler.
    Because the latent control effects from \lstinline|k|'s body already includes the prefixes from the start of \lstinline|k|'s body to the \lstinline|abort|, the existing left accumulation in our definition of $\rhd$ correctly accumulates prefixes from the site of continuation invocation, into the continuation, to its uses of control operations.  
    In general (for a single tag), invoking a continuation with latent effect $(C,{Q})$ has effect $(C\cup\{\mathsf{replace}(Q)\},\bot)$, though this assumes a non-empty underlying effect for the continuation --- an assumption the final type rules will need to relax, along with extension for multiple tags, and issues with the continuation's argument and result type.
\end{usecase}

\begin{usecase}[{Capturing Continuations}]
    \label{use:capture}
    Typing uses of \lstinline|call/cc| is the most complex problem this effect system must address. For a term \lstinline|(call/cc t ($\lambda$k. e))|, the rule must type the body function, which means choosing a type for the variable \lstinline|k| that will be bound to the continuation. 
    \ifTR
    As just discussed, this type includes argument and return types, as well as a latent effect.  When invoked, the argument provided takes the place of the \lstinline|call/cc| term itself in the surrounding context, so the argument type must be equal to the result type of the term at hand.  Unfortunately, the others must be consistent with parts of the program that are \emph{not subterms of the \lstinline|call/cc| use}.  The result type of the continuation must agree with the result type of the prompt within which it is invoked.  
    We will defer discussing the continuation return type and focus on the question of captured continuations' latent effects.
    \fi
    The latent effect of the continuation parameter depends on the effect of the code in the captured context --- code that is (at runtime) ``between'' the \lstinline|call/cc| and the (dynamically) nearest prompt.
    Consider applying a purely local type rule for \lstinline|call/cc| to this simple example:
    \begin{lstlisting}
        (% t ((call/cc t ($\lambda$k. e)); (foo 3)) ...)
    \end{lstlisting}
    Here the context captured will clearly be \lstinline|($\bullet$; (foo 3))|.
    This is awkward when typing the subterm \lstinline|(call/cc t ($\lambda$k. e))|, because that context is not a subterm of what is being typed. 
    \ifTR
    It may be tempting to push the surrounding context's latent effect ``down'' into subterms, but the continuation invocation above may have arisen from reducing a function call at runtime --- the source level type checker may not see a unique context for each \lstinline|call/cc|.  
    \fi

    We will take a ``guess and check'' approach\footnote{In the sense that a human constructing a typing derivation involving \lstinline|call/cc| would need to informally guess a prophecy, then use the type rules to check that it was a sound guess.} to typing \lstinline|call/cc|, assuming a certain latent effect and ensuring it can be checked elsewhere when additional information is available (in our case, in the type rule for prompts), essentially a form of \emph{prophecy}~\cite{abadi1988existence,abadi1991existence}.
    We track prophecies in a third (final) component, the \emph{prophecy set}, resulting in three-component effects $(P,C,\opt{Q})$.  We call these three-part effects \emph{continuation effects}, using metavariable $\chi$.
    An individual prophecy records the assumed latent effect of a continuation. Since that continuation may capture further continuations, the prophecy must predict a full continuation effect, making prophecies and continuation effects mutually recursive.

    Prophecies alone are only local guesses about non-local phenomena; the effect system requires a way to validate them.
    Intuitively, a prophecy that a \lstinline|call/cc| captures a continuation with latent effect $(P,C,Q)$ is valid if in any dynamic context the \lstinline|call/cc| is evaluated, the effect of the program text ``between'' the capture and the nearest prompt has an effect less than $(P,C,Q)$ in the partial order (to be defined). 
    \ifTR
    This is complicated as usual because the dynamic context is always \emph{outside} the \lstinline|call/cc| itself (not a subterm), and because the context likely does not appear explicitly in the program text. 
    While the full context captured is visible in
    \begin{lstlisting}
        (% t ((call/cc t ($\lambda$k. e)); event[$\alpha$]) ($\lambda$n. ...))
    \end{lstlisting}
    it is not synactically obvious in
    \begin{lstlisting}
        let f = ($\lambda$x. (call/cc t ($\lambda$k. k x))) in
        (% t ((f ()); event[$\alpha$]) ($\lambda$n. ...))
    \end{lstlisting}
    and if $f$ were invoked in multiple contexts, it becomes more difficult to predict.
    Yet in either case, it is safe for the typing of the \lstinline|call/cc| to assume the (underlying) latent effect of the continuation is $\{\alpha\}$, while adding a second \lstinline|event[$\beta$]| to the end of the prompt body would invalidate such an assumption, so the validation depends not only on what prophecy was made, but also on \emph{where in the term} the \lstinline|call/cc| occurs.
    \fi

    We can again turn to the idea of accumulation, but \emph{to the right}.
    Thus we write individual prophecies (again, ignoring tags and argument and result types of continuations) as $\mathsf{prophecy}~\chi~\mathsf{obs}~\chi'$ --- indicating that for a certain \lstinline|call/cc|, an effect of $\chi$ was prophecized, meaning that the term whose effect contains this prophecy assumed a latent effect $\chi$ for the continuation, and the effect of the term fragment between the \lstinline|call/cc| and the boundary of the term has effect $\chi'$, which we call the \emph{observation}.  We will extend sequential composition to accumulate effects to the right into the observation.  When type checking a prompt, there is no more to accumulate because the prompt is the boundary of the continuation capture: at that point, the observation reflects the \emph{actual} effect of the code \emph{statically} on control flow paths between \lstinline|call/cc| and the prompt, so the prophecy was sound if the observed effect is less than the prophecy effect.
    \looseness=-1

    With this prophecy-and-observation approach, let us consider:
    \begin{lstlisting}
        (% t (((call/cc t ($\lambda$k. e)); event[$\alpha$]); event[$\beta$]) ($\lambda$n. ...))
    \end{lstlisting}
    The type rule for \lstinline|call/cc| can assume a surrounding context effect of $(\emptyset,\emptyset,\{\alpha\beta\})$.  Let us assume for simplicity the inner body of the \lstinline|call/cc| is pure (unit effect).  Then the effect of the \lstinline|call/cc| term can be 
    $(\{\mathsf{prophecy}~(\emptyset,\emptyset,\{\alpha\beta\})~\mathsf{obs}~(\emptyset,\emptyset,I)\},\emptyset,I)$
     (writing the unit effect as simply $I$ for readability).  Then the effect of sequencing that with the \lstinline|event[$\alpha$]| (i.e., $(\emptyset,\emptyset,\{\alpha\})$) should both update the underlying effect \emph{and the observation} of the prophecy:
    $(\{\mathsf{prophecy}~(\emptyset,\emptyset,\{\alpha\beta\})~\mathsf{obs}~(\emptyset,\emptyset,\{\alpha\})\},\emptyset,\{\alpha\})$.
     Repeating this for the next event, we would get
    $(\{\mathsf{prophecy}~(\emptyset,\emptyset,\{\alpha\beta\})~\mathsf{obs}~(\emptyset,\emptyset,\{\alpha\beta\})\},\emptyset,\{\alpha\beta\})$.
    At that point, this has typed the entire body of the prompt, and the type rule for the prompt can check that the observed effect is no greater than the prophecized effect --- in this case trivially since they are equal.

\end{usecase}

\begin{usecase}[{Trouble with Tags}]
    \label{use:tags}
    Work with delimited continuations typically formalizes work without tags, but then adds them to the implementation as a ``straightforward'' extension.  There are, however, several ways in which handling multiple tags is \emph{non-trivial} in this context, warranting an explicit treatment.  
    First, nested continuations result in more subtlety when handling control effect sets in the prompt rule. An obvious change is for a prompt tagged \lstinline|t| to ignore aborts and continuation invocations (elements of the control effect set) that target a different tag (which requires tracking the target tag for each \textsf{replace} or \textsf{abort} effect).
    However, this is insufficient.
    A continuation \lstinline|k| that restores up to tag $t$ may include a latent abort to tag $t2$ --- but if the nearest prompt is tagged $t2$, this is subtle.
    Consider the continuation $k=(\mathsf{cont}^\mathsf{unit}~t~(\bullet; \mathsf{abort}~\mathit{t2}~3))$ in the context of:
    \begin{lstlisting}[escapechar=`]
        (% t2 (% t `\colorbox{white}{\tt (\% t2 (\colorbox{lightgray}{(k ())}; event[$\alpha$]) ...)}` ...) ($\lambda$n. event[$\beta$]))
    \end{lstlisting}
    \lstinline|k| is invoked nested inside multiple prompts of different tags: when \lstinline|k| is restored, it will discard and replace the inner \lstinline|t2| prompt (white background) entirely.
    This is important: after context restoration, the \lstinline|abort| inside \lstinline|k| will be executed, passing 3 to the \emph{outermost} handler and emitting event $\beta$: 
    the innermost \lstinline|t2| prompt contains an abort to \lstinline|t2|, yet the innermost handler will not execute; the \emph{outermost} handler will.
    \begin{lstlisting}[escapechar=`]
(% t2 (% t ((); abort t2 3) ...) ($\lambda$n. event[$\beta$])) $\Rightarrow^*$ ($\lambda$n. event[$\beta$]) 3 
    \end{lstlisting}
    Our initial effect for invoking continuations with further control effects (Use Case \ref{use:nestedcont}) would have the inner prompt's body effect contain $\mathsf{abort}~\mathit{t2}~\{\epsilon\}$, which our suggested prompt handling would then join into the \emph{innermost} prompt due to the matching tag.  This has two problems: it spuriously suggests the innermost prompt's handler could execute, introducing a kind of imprecision; and because that abort effect would no longer be present in the effect of the outermost prompt's body, it would be missed that the outer handler \emph{could} run, making the approach unsound. We must refine our approach for these ``jumps.''
    \looseness=-1

    We resolve this by adding one final nuance to control effect sets.  We allow control effects to be either basic (the \textsf{replace} or \textsf{abort} effects we have already seen, with target tags) or \emph{blocked} $\pblock{c}{t}$ for a control effect $c$.  A blocked control effect $\pblock{c}{t}$ is ignored by prompts (left in the control effect set) until a prompt tagged $t$ is reached --- that prompt will \emph{unblock} the effect, leaving $c$ in the control effect set of the prompt.
    When invoking a continuation, instead of simply including the latent control effects in the effect of the invocation, the type rule will include the latent control effects \emph{blocked} until the target tag of the continuation.  So in the example above, the inner prompt's body effect would instead include $\pblock{\mathsf{abort}~\mathit{t2}~\{\epsilon\}}{t}$.  This would be ignored (propagated) by the inner prompt's typing (so the inner handler would not be spuriously considered), unblocked by the middle prompt's typing, and finally resolved in the outer prompt's typing (which would trigger consideration of the outer handler).  Because the overall effect of any execution path that triggers such blocked control effects must still execute code along the way to the continuation invocation, blocked control effects still accumulate on the left.

    Prophecies raise similar issues that lead to similar introduction of blocked prophecies: if \lstinline|k| above instead captured a continuation up to a prompt tagged $\mathit{t2}$, the white-background portion of the term discarded when \lstinline|k| was invoked would not be part of the captured continuation, so it should not be incorporated into the observation part of a prophecy.  Because of this, blocked prophecies \emph{must not accumulate while blocked}.  The difference is this: blocked control effects continue to accumulate on the left because control operations do not discard code that occurs on the way to a control effect, while blocked prophecies must ``skip over'' the code discarded by control operations, which always appears to the right (later in source program order).
\end{usecase}

    \section{Continuation Effects}
    \label{sec:formal}
    \begin{figure}\small
        \centerline{$
            \begin{array}{rlp{0.46\textwidth}}
                c\in\textrm{ControlEffect}  ::= & \mathsf{replace}~\ell:Q\leadsto\tau & \textrm{Invoked continuation up to $\ell$, with prefix \& underlying effect $Q$, continuation result type $\tau$}\\
                 \mid & \mathsf{abort}~\ell~Q\leadsto\tau & \textrm{Abort up to $\ell$ after effects $Q$, throwing value of type $\tau$}\\
                 \mid & \pblock{c}{\ell} & \textrm{Blocked control effect, frozen until nearest prompt for $\ell$ (triggered inside restored continuation targeting $\ell$)}\\
                p\in\textrm{Prophecy}  ::= &\mathsf{prophecy}~\ell~\chi~\leadsto\tau~\mathsf{obs}~\chi' & \textrm{Prophecy of latent continuation effect $\chi$, effect $\chi'$ observed since point of prophecy (\lstinline|call/cc|)}\\
                 \mid & \pblock{p}{\ell} & \textrm{Prophecy blocked until $\ell$}\\
                \chi\in\textrm{ContinuationEffect}   = & 
                    \multicolumn{2}{l}{\mathsf{set}~\mathsf{Prophecy}\times\mathsf{set}~\mathsf{ControlEffect}\times\mathsf{option}~\mathsf{UnderlyingEffect} }
            \end{array}
        $}
        \caption{Grammar of continuation effects over an existing effect quantale $Q\in\textrm{UnderlyingEffect}$}
        \vspace{-1.5em}
        \label{fig:contgrammar}
    \end{figure}
    This section makes the outlines of the previous section precise, and fills in missing details (such as coordinating the type of a value thrown with the argument type of a handler).
    Figure \ref{fig:contgrammar} defines the effects of $\mathcal{C}(Q)$ derived from the examples above, for underlying effect quantale $Q$.  Continuation-aware effects of an underlying effect quantale $Q$ are effects $\chi$ of three components: a prophecy set $P$, a control effect set $C$, and an optional underlying effect $\opt{Q}$.
    Basic control effects include effects representing aborts to a tagged prompt ($\mathsf{abort}~\ell~Q\leadsto\tau$) or invoking continuations that replace the context up the nearest tagged prompt ($\mathsf{replace}~\ell:Q\leadsto\tau$), as suggested by Use Cases \ref{use:basicabort} and \ref{use:basiccallcont}.
    These versions are additionally tagged with the specific prompt tag they target ($\ell$), and each carries a type $\tau$ --- for replacement this is the result type of the restored continuation, and for aborts this is the type of the value thrown (each intuitively a kind of result type for each control behavior).
    Both control effects and prophecies may also be blocked until a prompt with a certain tag if they originate inside a continuation that was invoked (per Use Case \ref{use:tags}).  Note that blocking constructors may nest arbitrarily deeply, because one restored continuation may restore another continuation which may restore another\ldots and so on.
    For a term with effect $(P,C,\opt{Q}$):
    \begin{itemize}
        \item The prophecy set $P$ contains prophecies for all uses of \lstinline|call/cc| within the term (some possibly blocked if introduced by restoring a continuation).
        \item The control effect set $C$ describes all possible exits of the term via control operations (abort or continuation invocation).  For aborts of continuation invocation that occur directly, $C$ will contain basic control effects.  For aborts or continuation invocation that may occur in the body of a restored continuation, there will be blocked control effects.
        \item The underlying effect $\opt{Q}$ describes (an upper bound on) the underlying effect of any execution of the term that does not exit via control operator.
    \end{itemize}

    To define sequencing and join formally, we must first lift the underlying effect quantale's operators to deal with missing effects:\\
    \centerline{$
            \opt{Q_1}\rhd\opt{Q_2} = \left\{
                \begin{array}{ll}
                    \top & \textrm{if}~\opt{Q_1}=\top\vee\opt{Q_2}=\top\\
                    \bot & \textrm{if}~\opt{Q_1}=\bot\vee\opt{Q_2}=\bot\\
                    Q_1\rhd Q_2 & \textrm{otherwise}
                \end{array}
                \right.
            \;
            \opt{Q_1}\sqcup\opt{Q_2} = \left\{
                \begin{array}{ll}
                    \top & \textrm{if}~\opt{Q_1}=\top\vee\opt{Q_2}=\top\\
                    Q_1 & \textrm{if}~\opt{Q_2}=\bot\\
                    Q_2 & \textrm{if}~\opt{Q_1}=\bot\\
                    Q_1\sqcup Q_2 & \textrm{otherwise}
                \end{array}
                \right.
    $}
    We also require a way to prefix control effects with an underlying effect, to implement left accumulation (recursively), extending the ideas of Use Cases \ref{use:basicabort}, \ref{use:basiccallcont}, and \ref{use:tags} to the final definition:\\
    \centerline{$\begin{array}{rcl}
            Q\rhd \pblock{c}{\ell} &=& \pblock{Q\rhd c}{\ell}\\
             Q\rhd\mathsf{replace}~\ell:Q'\leadsto\tau&=&\mathsf{replace}~\ell:(Q\rhd Q')\leadsto\tau \\
            Q\rhd \mathsf{abort}~\ell~Q'\leadsto\tau &=& \mathsf{abort}~\ell~(Q\rhd Q')\leadsto\tau
    \end{array}$}
    and we will lift this to operate on control effect \emph{sets} and possibly-absent underlying effects:\\
    \centerline{$
            \opt{Q}\rhd C = \mathsf{if}~(\opt{Q}=\bot)~\mathsf{then}~\emptyset~\mathsf{else}~(\mathsf{map}~(Q\rhd\_)~C)
    $}
    Likewise, we must define a means of right-accumulating in prophecies:\\
    \centerline{$
            \begin{array}{l}
            \pblock{\mathsf{prophecy}~\ell~(P,C,\opt{Q})\leadsto\tau~\mathsf{obs}~(P',C',\opt{Q'})}{\ell'}\blacktriangleright\chi'' =
                \pblock{\mathsf{prophecy}~\ell~(P,C,\opt{Q})\leadsto\tau~\mathsf{obs}~(P',C',\opt{Q'})}{\ell'}\arcr
            \mathsf{prophecy}~\ell~(P,C,\opt{Q})\leadsto\tau~\mathsf{obs}~(P',C',\opt{Q'})\blacktriangleright(P'',C'',\opt{Q''}) \arcr
            \qquad=\mathsf{prophecy}~\ell~(P,C,\opt{Q})\leadsto\tau~\mathsf{obs}~((P'\blacktriangleright(P'',C'',\opt{Q''}))\cup P'',C'\cup(\opt{Q'}\rhd C''),\opt{Q'}\rhd \opt{Q''})
            \end{array}
    $}
    which we also lift to operate on prophecy \emph{sets} (not shown, but analogous to the lifting of left-accumulation).
    Finally, this is enough to define sequencing and join:\\
    \centerline{$
\begin{array}{rcl}
            (P_1,C_1,\opt{Q_1})\sqcup(P_2,C_2,\opt{Q_2})&=&(P_1\cup P_2, C_1\cup C_2, \opt{Q_1}\sqcup \opt{Q_2})
            \\
            (P_1,C_1,\opt{Q_1})\rhd(P_2,C_2,\opt{Q_2})&=&((P_1\blacktriangleright(P_2,C_2,\opt{Q_2}))\cup P_2, C_1\cup(\opt{Q_1}\rhd C_2), \opt{Q_1}\rhd \opt{Q_2})
            \end{array}
    $}
    Joins are implemented component-wise, using set union on prophecy or control effect sets, and the (option-lifted) join from the underlying effect quantale.  The sequencing operator, and its relation to the accumulation on prophecies, is a bit complex and warrants some further explanation.
    The sequence operator $\rhd$ is defined according to the ideas driven by Use Cases \ref{use:basicabort}, \ref{use:basiccallcont}, and \ref{use:capture}. Underlying effects are sequenced by reusing the underlying effect quantale.
    Control effects are handled by left-accumulating.  In fact, in the case where there are no prophecies (\lstinline|call/cc|s) involved, the handling of the control effect sets and underlying effects is exactly as in Use Case \ref{use:basicabort}, just extended for the additional control effects.
    Deferring the right-accumulation $\blacktriangleright$ for one further moment, the full sequencing operator produces a resulting prophecy set as the union of prophecies from the second effect (which are unaffected by the first) with the result of the first effect's prophecies accumulating effects from the right, since that effect will (in the type rules) correspond to behavior that will be part of the captured continuation.
    The right accumulation for prophecies implemented by $\blacktriangleright$ essentially just implements sequential composition of the observation with the accumulated effect --- the recursive use of $\blacktriangleright$ could equivalently just be $\rhd$, but direct recursion is easier to prove things about than mutual recursion.  As suggested by Use Case \ref{use:tags}, blocked prophecies do not accumulate.
    It is easy to confirm that the unit element for $\rhd$ is $(\emptyset,\emptyset,I)$, where $I$ is the identity of the underlying effect quantale.
    
    \begin{remark}[Sets of Underlying Effects]
        \label{rmk:unions}
        We have described most of the structure of lifting an effect quantale to support delimited control: sequencing and join operators that distribute over each other, along with a unit for sequencing.
        We have not yet stated whether the result $\mathcal{C}(Q)$ of the lifting is an effect quantale.
        As described so far, it is not quite an effect quantale: there is no single distinguished top in the partial order induced by $\sqcup$: for any effect $(P,C,\opt{Q})$, a larger effect can be obtained by adding new control effects or prophecies. And because the underlying $\top$ can appear in multiple ways, conceptually many different incomparable effects should be considered erroneous.
        Introducing a distinguished top element $\mathsf{Err}$, and wrapping the sequencing and join definitions above with an additional operation producing $\mathsf{Err}$ any time the operations above produce effects containing underlying top.  This would produce an effect quantale, but we take an alternative approach that also enhances flexibility, without adding special cases for $\top$ throughout the system.  

        Using \emph{sets} of structures containing underlying effects can lead to the extra set structure being ``too picky'' in distinguishing effects, in the sense of distinguishing intuitively equivalent effects.
        Consider the following join:
        \\\centerline{$(\emptyset,\{\mathsf{abort}~\ell~\{\alpha\}\},\bot)\sqcup(\emptyset,\{\mathsf{abort}~\ell~\{\beta\}\},\bot) = (\emptyset,\{\mathsf{abort}~\ell~\{\alpha\},\mathsf{abort}~\ell~\{\beta\}\},\bot)$}
        which could be the effect of \lstinline|if c (event[$\alpha$]; abort $\ell$ 3) (event[$\beta$]; abort $\ell$ 3)|.
        Their join is \emph{not} the very similar
        $(\emptyset,\{\mathsf{abort}~\ell~\{\alpha,\beta\}\})$, which also indicates an abort after one of the two same underlying effects, and is the effect of a minor rewrite of the last expression: \lstinline|(if c (event[$\alpha$]) (event[$\beta$])); abort $\ell$ 3|.
        Both effects indicate executing $\alpha$ or $\beta$ before aborting, and replacing one with the other inside a prompt body will not change the \emph{prompt}'s effect (per Use Cases \ref{use:basicabort} and \ref{use:basiccallcont}): the prompt will sequence each with the handler effect, and join them together). Yet because $\sqcup$ is defined using set union, they are incomparable in the induced partial order $x\sqsubseteq y\leftrightarrow x\sqcup y=y$, because joining them yields a \emph{third} effect:
        $(\emptyset,\{\mathsf{abort}~\ell~\{\alpha\},\mathsf{abort}~\ell~\{\beta\},\mathsf{abort}~\ell~\{\alpha,\beta\}\},\bot)$. 
        This is not a valuable distinction. We would like at least
        {$(\emptyset,\{\mathsf{abort}~\ell~\{\alpha\},\mathsf{abort}~\ell~\{\beta\}\},\bot)\sqsubseteq(\emptyset,\{\mathsf{abort}~\ell~\{\alpha,\beta\}\})$}
        because every abort prefix on the left is over-approximated by an abort prefix on the right.\footnote{While in this $\mathcal{T}(\Sigma)$ example the effects are equivalent, examples in other effect quantales only justify $\sqsubseteq$.}
        One way to achieve this would be to replace the set union of control effect sets or prophecy sets in the sequencing and join operators with operators that also recursively joined the underlying effects of all aborts to the same tag, joined the underlying effects of all replacements to the same tag, and joined the observed effects of all prophecies to the same tag with the same prediction (plus joining types). 
        The partial order induced by this modification would establish this desirable order.
        \looseness=-1

        Directly extending $\sqcup$ and $\rhd$ as given to \emph{both} recursively join when combining sets and lift underlying $\top$ to a special $\mathsf{Err}$ would yield a proper effect quantale, but add significant complexity that is orthogonal to the key ideas of our approach.
        \emph{Instead}, we make two adjustments.  First, we define a direct partial order on continuation effects in Figure \ref{fig:directPO}, where an effect $\chi$ is less than another effect $\chi'$ if every prophecy observation, abort effect, replacement effect, or underlying effect in $\chi$ is over-approximated by \emph{some} corresponding component in $\chi'$.  This captures the intuition behind the desirable ordering outlined above. We also define a corresponding equivalence relation $\approx$ on continuation effects, which equates all continuation effects containing an underlying $\top$ (anywhere in the effect) and otherwise uses the partial order to induce equivalence.
        In the type rules considered in Section \ref{sec:rules}, this is the notion of subeffecting used (rather than the traditional partial order derived from a join), and all effects are considered modulo the equivalence relation.
        Quotienting $\mathcal{C}(Q)$ by the equivalence $\approx$ yields a proper effect quantale, equivalent to the direct but verbose version outlined above. See 
        \ifTR
        Section \ref{sec:modulo}
        \fi
        for further details.
    \end{remark}

    \begin{figure*}\small
        \begin{mathpar}
            C_1\sqsubseteq C_2 = \bigwedge\left\{\begin{array}{l}
                \forall\ell,Q,\tau\ldotp \mathsf{replace}~\ell:Q\leadsto \tau\in C_1\Rightarrow \exists Q'\ldotp \mathsf{replace}~\ell:Q'\leadsto\tau\in C_2 \land Q\sqsubseteq Q'\\
                \forall\ell,Q,\tau\ldotp \mathsf{abort}~\ell~Q\leadsto\tau\in C_1\Rightarrow \exists Q'\ldotp \mathsf{abort}~\ell~Q'\leadsto\tau\in C_2 \land Q\sqsubseteq Q'\\
            \end{array}\right.
            \\
            P_1\sqsubseteq P_2 = \begin{array}{l}
                (\forall\ell,\chi,\chi',\tau\ldotp\mathsf{prophecy}~\ell~\chi~\leadsto\tau~\mathsf{obs}~\chi'\in P_1\Rightarrow
                \exists \chi''\ldotp\mathsf{prophecy}~\ell~\chi~\leadsto\tau~\mathsf{obs}~\chi''\in P_2\land \chi'\sqsubseteq\chi'')\land
                \\
                (\forall\ell,\ell',\chi,\chi',\tau\ldotp \pblock{\mathsf{prophecy}~\ell~\chi~\leadsto\tau~\mathsf{obs}~\chi'}{\ell'}\in P_1\Rightarrow
                \exists \chi''\ldotp\pblock{\mathsf{prophecy}~\ell~\chi~\leadsto\tau~\mathsf{obs}~\chi''}{\ell'}\in P_2\land \chi'\sqsubseteq\chi'')
            \end{array}
            \\
    (P,C,\opt{Q})\sqsubseteq(P',C',\opt{Q'}) \leftrightarrow P\sqsubseteq P' \land C\sqsubseteq C' \land \opt{Q}\sqsubseteq\opt{Q'} 
    \\
        \chi\approx\chi' =
        (\mathsf{NoUnderlyingTop}(\chi,\chi')\land \chi\sqsubseteq\chi'\land\chi'\sqsubseteq\chi) 
        \vee
        (\mathsf{UnderlyingTop}(\chi)\land \mathsf{UnderlyingTop}(\chi'))
        \end{mathpar}
        \vspace{-2em}
        \caption{Direct partial order and equivalence on continuation effects.
        \vspace{-1em}
        }
        \label{fig:directPO}
    \end{figure*} 

    \label{sec:rules}
    We consider a type-and-effect system for a language with the constructs from Figure \ref{fig:opsem}.  Our extended technical report~\cite{controltr} extends these results for composable continuations.  Our expressions and types are:
    \looseness=-1
    \vspace{-1ex}
    \[\begin{array}{r@{\ }l}
    \textrm{Expressions} & e ::= p \mid \lambda x\ldotp e \mid (e\;e)\mid (\%\;\ell\;E\;v) \mid (\mathsf{call/cc}\;\ell\;e) \mid (\mathsf{abort}\;\ell\;e) \mid (\mathsf{cont}\;\ell\;E) \mid \mathsf{if}\;e\;e\;e\\
    \textrm{Values} & v ::= (\lambda x\ldotp e) \mid \mathsf{cont}\;\ell\;E \mid v_p\\
    \textrm{Types} & \tau,\gamma ::= \mathsf{unit} \mid \mathsf{bool} \mid
        \tau\xrightarrow{\chi}\tau \mid (\mathsf{cont}\;\ell\;\tau\;\chi\;\tau) \mid \mu X\ldotp \tau
    \end{array}\]
    $p$ and $v_p$ are parameters of the system following Gordon's work~\cite{ecoop17,quantalesjournal}: primitives (which can include operations such as locking primitives or \lstinline|event[$\alpha$]|) and primitive values (e.g., for encoding locations). 
    Gordon's soundness framework also parameterizes operational semantics by an abstract notion of state, and semantics for primitives manipulating state; we assume (and later use) a similar framework, which admits a range of concrete examples.

    Types include common primitive types, function types with latent effects, equirecursive types (needed for typing loops), as well as a type for continuation values that we discuss with the type rule for invoking continuations.
    The type rules for lambda abstraction, function application, and conditionals are as in Section \ref{sec:bg} (though using continuation effects), so we do not discuss them further.   
    Typing uses of primitives requires additional rules and parameters to define the additional types (e.g., lock or location types) and their relationship to operational primitives, following Gordon~\cite{ecoop17,quantalesjournal}.  
    For intuition, readers may assume $p$ is simply the \lstinline|event[_]| primitive from our running example.
    We include subtyping ($<:$), including the standard type-and-effect subsumption rule, and function subtyping that is covariant in the body's latent effect. 
    Figure \ref{fig:controlops} gives central type rules for this paper, for prompts, aborts, continuation capture, and continuation invocation.

    \begin{figure*}[t]\small
    \begin{mathpar}
        \inferrule*[left=V-Effects]{
            \forall\tau',Q\ldotp(\mathsf{abort}~\ell~Q\leadsto\tau')\in \punblock{C}{\ell}\Rightarrow  
            \tau'<:\sigma 
            \quad
            \forall Q,\tau'\ldotp(\mathsf{replace}~\ell:Q\leadsto\tau')\in \punblock{C}{\ell}\Rightarrow  
            \tau'<:\tau
            \\
            \left(\begin{array}{l}
                \forall \chi_{proph},\tau',P_p,C_p,\opt{Q_p}\ldotp 
                    \mathsf{prophecy}~\ell~\chi_{proph}\leadsto\tau'~\mathsf{obs}~(P_p,C_p,\opt{Q_p})\in \punblock{P}{\ell}\Rightarrow
                \arcr
                \qquad 
                \punblock{P_p}{\ell}\sqsubseteq\punblock{P_{proph}}{\ell}\land
                \punblock{C_p}{\ell}\sqsubseteq \punblock{C_{proph}}{\ell}\land \opt{Q_p}\sqsubseteq \opt{Q_{proph}}
                \land 
                \tau<:\tau'
            \end{array}\right)
        }{
            \mathsf{validEffects}(P,C,\opt{Q},\ell,\tau,\sigma)
        }
        \and
        \inferrule*[left=T-Prompt]{
            \Gamma\vdash e : \tau \mid (P,C,\opt{Q})\\
            \Gamma\vdash h : \sigma\xrightarrow{(\emptyset,\emptyset,Q_h)}\tau\mid I\\
            \mathsf{validEffects}(P,C,\opt{Q},\ell,\tau,\sigma)
        }{
            \Gamma\vdash (\%~\ell~e~h) : \tau \mid (\punblock{P}{\ell}\setminus^{Q_h}\ell, \punblock{C}{\ell}\setminus\ell, \opt{Q}\sqcup\left(\bigsqcup \punblock{C}{\ell}|^{Q_h}_\ell\right) )
        }
        \and
        \inferrule*[left=T-CallCont]{
            \mathsf{NonTrivial}(\chi_k)\\
            \Gamma\vdash e : (\mathsf{cont}\;\ell\;\tau\;\chi_k\;\gamma)\overset{\chi}{\rightarrow}\tau \mid \chi_e
        }{
            \Gamma\vdash (\mathsf{call/cc}\;\ell\;e): \tau \mid (\chi_e\rhd\chi)\rhd(\{\prophecy{\ell~\chi_k\leadsto\gamma}{(\emptyset,\emptyset,I)}\},\emptyset,I)
        }
        \and
        \inferrule*[left=T-AppCont]{
            \Gamma\vdash k : \mathsf{cont}\;\ell\;\tau'\;(P,C,\opt{Q})~\tau'' \mid \chi_k\\
            \Gamma\vdash e : \tau' \mid \chi_e
        }{
            \Gamma\vdash (k\;e) : \tau \mid \chi_k\rhd\chi_e\rhd(\pblock{P}{\ell},\pblock{C}{\ell}\cup(\opt{Q}\rhd\{\replace{\ell:I\leadsto\tau''})\},\bot)
        }
        \and
        \inferrule*[left=T-Abort]{
            \Gamma\vdash e : \tau \mid \chi_e
        }{
            \Gamma\vdash \mathsf{abort}~\ell~e : \sigma \mid \chi_e\rhd(\emptyset,\{\mathsf{abort}~\ell~I\leadsto\tau\},\bot)
        }
    \end{mathpar}
    \vspace{-2em}
    \caption{Typing control operators with continuation effects.
    \vspace{-1em}
    }
    \label{fig:controlops}
    \end{figure*}

    We will discuss the type rules in relation to the Use Cases from Section \ref{sec:usecases}.

    \textsc{T-Abort} handles Use Case \ref{use:basicabort}.  As suggested in that discussion, the effect of an abort is to introduce a control effect signifying an abort to the targeted label, with an absent underlying effect.  The initial prefix of the abort is ``empty'' --- the underlying unit effect $I$ --- and the control effect tracks the type of the thrown value.  The overall effect of an \lstinline|abort| expression sequences this after the effect of reducing $e$ to a value, since this occurs earlier in evaluation order than the abort operation itself. A subtlety worth noting, is that this also ensures a \lstinline|call/cc| inside $e$ will correctly accumulate the pending abort in the captured context.
    \looseness=-1

    \textsc{T-AppCont} handles Use Cases \ref{use:basiccallcont} and \ref{use:nestedcont}.
    First consider the basic case where the invoked continuation has only simple underlying effects (i.e., $P$ and $C$ in the continuation's latent effect are both $\emptyset$).  As with \textsc{T-Abort}, subterms are reduced to values before the control behavior occurs, so those effects are sequenced before the control behavior itself.  With that taken care of, we may temporarily assume both $k$ and $e$ are already syntactic values to simplify discussion of effect.  In this case the rule simplifies to exactly what the example in Use Case \ref{use:basiccallcont} suggested: the underlying effect is absent (because the term does not return normally, but via a control behavior), and a control effect is introduced reflecting that a continuation with underlying effect $\opt{Q}$ is invoked.  A subtlety here is that we cannot simply write $\mathsf{replace}~\ell:\opt{Q}\leadsto\tau''$, because $\opt{Q}$ may be $\bot$, which would make that control effect invalid.  Instead we (ab)use the left-accumulation operator on control effects: prefixing the \emph{unit-effect} replacement effect with $\opt{Q}$ will give the expected result when $\opt{Q}$ is present, and otherwise result in the empty set.
    The replacement effect also records the result type of the invoked continuation, and the rule also ensures the argument provided to the continuation is the expected type.
    \looseness=-1

    \textsc{T-AppCont} also handles Use Case \ref{use:nestedcont}, adjusted per Use Case \ref{use:tags}: the latent control effects of the continuation are included, but \emph{blocked until $\ell$}, to ensure no prompt rule (discussed shortly) resolves those behaviors before the behaviors escape a prompt tagged $\ell$.
    Unlike the discussion in Section \ref{sec:usecases}, this finished rule also supports the case where the restored continuation contains (possibly-nested) uses of \lstinline|call/cc|, blocking the latent prophecies as well. 
    \looseness=-1

    The conclusion of \textsc{T-AppCont} critically overloads the syntax for constructing blocked prophecies and control effects, to block prophecy \emph{sets} and control effect \emph{sets}.  This overload \emph{almost} maps the appropriate blocking constructor over each set --- however, it first checks that this will not result in a control effect of the form $\pblock{\pblock{c}{\ell}}{\ell}$ for some tag $\ell$ (similarly for prophecies).  Such a control effect would represent a control effect that should propagate directly through \emph{two} prompts tagged $\ell$, but this is dynamically impossible: any number of nested restorations of continuations to the same tag remains within the same prompt, e.g.:
    \begin{lstlisting}
(% $\ell$ ((cont $\ell$ ($\bullet$; ((cont $\ell$ ($\bullet$; abort $\ell$ 5)) 4))) 3) h)
$\Rightarrow$ (% $\ell$ (3; ((cont $\ell$ ($\bullet$; abort $\ell$ 5)) 4)) h) $\Rightarrow$ (% $\ell$ ((cont $\ell$ ($\bullet$; abort $\ell$ 5)) 4) h)
$\Rightarrow$ (% $\ell$ (4; abort $\ell$ 5) h) $\Rightarrow$ (% $\ell$ (abort $\ell$ 5) h) $\Rightarrow$ (h 5)
    \end{lstlisting}
    Na\"ively mapping the blocking constructor would yield the control effect set 
    $\{\pblock{\pblock{\mathsf{abort}~\ell~I\leadsto\mathsf{nat}}{\ell}}{\ell}\}$ for the body; our discussion with Use Case \ref{use:tags} about each prompt removing one layer of blocking (which we will see does inform \textsc{T-Prompt}) would then not pair the abort with the local handler that is invoked.  With the modified mapping, the simplification of the control effect set is instead $\{\pblock{\mathsf{abort}~\ell~I\leadsto\mathsf{nat}}{\ell}\}$ (only one layer of blocking), which will match the \lstinline|abort| to the correct handler.

    \textsc{T-CallCont} is a bit more subtle.
    Standard for any type rule for \lstinline|call/cc|, the rule ensures the result type of the expression itself ($\tau$) is also the return type of the body function (for executions that return normally) and the argument type assumed for the continuation (since the location of the \lstinline|call/cc| becomes the hole the argument replaces at invocation).
    As suggested in the discussion of Use Case \ref{use:capture}, the effect arising from the use of \lstinline|call/cc| itself is a prophecy effect, recording the assumed latent effect of the captured continuation and the assumed result type of the captured continuation --- both of which make their way into the assumed type of the argument to the \lstinline|call/cc| body.
    The initial observation is empty, because this effect corresponds intuitively to the point in the execution from which the prediction begins --- but this point is the heart of another key subtlety.
    As with other rules, the subterm $e$ must be reduced before anything else, so its effect is sequenced before others.  
    But na\"ively ordering the body's (latent) effect would place it after the prophecy, as dynamically the continuation is captured before the body is evaluated (since the continuation becomes the argument in \textsc{E-CallCC}).  However, this would result in the prophecy effect observing the body of the \lstinline|call/cc| --- which is incorrect, as that behavior will not be part of the captured continuation.  Thus we place the prophecy \emph{after} the body's latent effect.

    For a small variation on Use Case \ref{use:capture}'s first example, this gives us:
    \\\centerline{$
        (\%\; t\; \underbrace{(\overbrace{(\mathsf{call/cc}\;t\;(\lambda k. e))}^{\chi_e\rhd(\{\mathsf{prophecy}\;t\;(\emptyset,\emptyset,\{\alpha\}\;\mathsf{obs}\;(\emptyset,\emptyset,I))\},\emptyset,I)};\;\overbrace{\mathsf{event}[\alpha]}^{(\emptyset,\emptyset,\{\alpha\})})}_{\chi_e\rhd(\{\mathsf{prophecy}\;t\;(\emptyset,\emptyset,\{\alpha\}\;\mathsf{obs}\;(\emptyset,\emptyset,\{\alpha\}))\},\emptyset,\{\alpha\})}\; ...)
    $}
    The individual effects of the \lstinline|call/cc| and \lstinline|event| expressions (given above the term) simplify to the effect below the term. The prophecy from the \lstinline|call/cc| observes the event that would be captured in its context.  Because $e$'s effect $\chi_e$ is to the \emph{left} of the resulting prophecy, $e$'s non-captured behavior is \emph{not} observed, and the prophecy under the term will appear in the prophecy set of the overall prompt's body.

    \textsc{T-CallCont} has one extra antecedent, constraining the prophecy to predict a \emph{non-trivial} effect, to avoid degeneracy.  The simplest problematic prediction would be to predict an effect of $(\emptyset,\emptyset,\bot)$.  Consider the effect of the code \lstinline|event[$\alpha$]; (k ())|.  The effect would be $(\emptyset,\emptyset,\{\alpha\})\rhd(\emptyset,\emptyset,\bot)=(\emptyset,\emptyset,\bot)$.  This is problematic: the term \emph{does} have a behavior, but the effect reflects \emph{no} behavior.  This could be introduced by a circularity:
    \begin{lstlisting}
        (% t (let k = (call/cc t ($\lambda$k. k)) in (event[$\alpha$]; (k ()))) ...)
    \end{lstlisting}
    This term is in fact the macro-expansion for an infinite loop that executes the event forever.  Assuming the degenerate latent effect in the \lstinline|call/cc| gives \lstinline|k| the degenerate effect, which gives the body of the the let-expression --- the context captured by \lstinline|call/cc| --- a degenerate effect, allowing the observed effect to match the prophecy (both degenerate).  Requiring a non-empty control effect set or underlying effect avoids this collapse. (This is not a termination-sensitivity issue; a terminating while loop has the same challenge, but a larger term.)
    \looseness=-1

    \begin{figure*}[t]\small
        \begin{mathpar}\scriptsize
            P\setminus^Q\ell = \{ p\setminus^Q\ell \mid p\in P \land \mathsf{OuterTag}(p)\neq\ell \}
            \quad
            C\setminus\ell =\{c\in C\mid \mathsf{OuterTag}(c)\neq\ell \} 
            \\
            C|^Q_\ell = \{Q'\rhd Q \mid \mathsf{abort}~\ell~Q'\leadsto\_\in C\}\cup\{ Q' \mid\mathsf{replace}~\ell:Q'\leadsto\_\in C\}
            \\
            \begin{array}{r@{~=~}l}
                \punblock{\mathsf{prophecy}~\ell'~(P,C,Q)\leadsto\tau~\mathsf{obs}~(P',C',Q')}{\ell} & \mathsf{prophecy}~\ell'~(P,C,Q)\leadsto\tau~\mathsf{obs}~(\punblock{P'}{\ell},\punblock{C'}{\ell},Q')\arcr
                \punblock{\pblock{\mathsf{prophecy}~\ell'~(P,C,Q)\leadsto\tau~\mathsf{obs}~(P',C',Q')}{\ell}}{\ell} & \mathsf{prophecy}~\ell'~(P,C,Q)\leadsto\tau~\mathsf{obs}~(\punblock{P'}{\ell},\punblock{C'}{\ell},Q')\arcr
                \punblock{
                    \pblock{\mathsf{prophecy}~\ell'~(P,C,Q)\leadsto\tau~\mathsf{obs}~(P',C',Q')}{\ell''}
                }{\ell} & \pblock{\mathsf{prophecy}~\ell'~(P,C,Q)\leadsto\tau~\mathsf{obs}~(P',C',Q')}{\ell''}~\quad\textrm{(if $\ell\neq\ell''$)}
            \end{array}\\
            \begin{array}{r@{~=~}l}
                (\mathsf{prophecy}~\ell'~\chi\leadsto\tau~\mathsf{obs}~(P',C',Q'))\setminus^Q\ell &
                \mathsf{prophecy}~\ell'~\chi\leadsto\tau~\mathsf{obs}~(
                    P'\setminus^Q\ell,
                    C'\setminus\ell,
                    Q'\sqcup(C'|^Q_\ell)
                )
                \arcr
                (\pblock{\mathsf{prophecy}~\ell'~\chi\leadsto\tau~\mathsf{obs}~(P',C',Q')}{\ell})\setminus^Q\ell &
                \mathsf{prophecy}~\ell'~\chi\leadsto\tau~\mathsf{obs}~(
                    P'\setminus^Q\ell,
                    C'\setminus\ell,
                    Q'\sqcup(C'|^Q_\ell)
                )
                \arcr
                (\mathsf{prophecy}~\ell'~\chi\leadsto\tau~\mathsf{obs}~(P',C',Q')~\mathsf{until}~\ell'')\setminus^Q\ell &
                \mathsf{prophecy}~\ell'~\chi\leadsto\tau~\mathsf{obs}~(
                    P',
                    C',
                    Q')
                )~\mathsf{until}~\ell''\quad(\textrm{if}~\ell\neq\ell'')
            \end{array}
        \end{mathpar}
        \vspace{-1em}
        \caption{Auxilliary definitions used by type rules.}
        \label{fig:aux}
    \end{figure*}

    \textsc{T-Prompt} is the most complex type rule in the system, because it serves many roles.  In addition to giving an overall type and effect to the prompt, it must check that:
    \begin{itemize}
        \item Any \lstinline|abort| to this handler throws values whose type is a subtype of the handler's argument.
        \item Any continuation invoked in the body that will restore up to this prompt has a result type that is a subtype of the prompt's own result type.
        \item Any \lstinline|call/cc| that captures a continuation delimited by this prompt was typed assuming the result was a \emph{supertype} of the prompt's actual result type.
        \item Any \lstinline|call/cc| that captures a continuation delimited by this prompt was typed assuming a valid latent effect for that continuation (i.e., that prophecies targeting that tag are upper-bounds on the observations).
    \end{itemize}
    The first three checks are handled by the subtyping constraints in the auxiliary judgment \textsc{V-Effects}.  Notice that the antecedents of \textsc{V-Effects} quantify over elements of \emph{unblocked} control effect sets and prophecy sets.  Unblocking, written $\punblock{-}{\ell}$ is defined for prophecies in Figure \ref{fig:aux}: it maps a corresponding per-element unblocking operation, which is a no-op on elements that were not blocked, a no-op on elements blocked until a \emph{different} tag, and strips off a block constructor for $\ell$ if that is the outermost constructor. Unblocking for control effect sets is defined analogously.  Intuitively, this corresponds to the fact that any control effect or prophecy blocked until a prompt for $\ell$ has now \emph{reached} a prompt for $\ell$.

    The last of the checks above is handled by the prophecy-related antecedent of \textsc{V-Effects}, which \emph{nearly} checks that the observations for a given prophecy are a subeffect of what was predicted --- that would be $(P_p,C_p,Q_p)\sqsubseteq\chi_{proph}$, which would be sound but overly conservative.  Instead, the comparison of the prophecy and observation, compares the \emph{unblocked} versions of prophecy and control effect sets: $\punblock{P_p}{\ell}\sqsubseteq\punblock{P_{proph}}{\ell}$ and $\punblock{C_p}{\ell}\sqsubseteq\punblock{P_{proph}}{\ell}$.
    This is useful because if a context captured by \lstinline|call/cc| contains an invocation of itself, the prophecy arising from \textsc{T-CallCont} will observe a \emph{blocked} version of \emph{its own prophecy} (from \textsc{T-AppCont}), making the naive subeffect check too conservative: no prophecy can predict a blocked version of itself. Rather than being an esoteric concern, this is actually quite practical: encodings of loops using delimited continuations do exactly this.
    Unblocking for the prompt's tag in \textsc{V-Effects} resolves this: just like the quantifications unblock because those sets have now reached the corresponding prompt, the prophecy validation must reflect that the observation has now ``reached'' the corresponding prompt.

    Let us revisit the example of Use Case \ref{use:capture} discussed above with \textsc{T-CallCont}. The prophecy in that derivation contains no blocked prophecies, so \textsc{T-Prompt} will effectively check that the observation is less than the prophecy --- that $(\emptyset,\emptyset,\{\alpha\})\sqsubseteq(\emptyset,\emptyset,\{\alpha\})$, which trivially succeeds because the prophecy predicted its context's behavior exactly.

    We can see the role of unblocking more clearly by revisiting the motivating example for requiring prophecized effects to be non-trivial (eliding types for brevity).
    \[
        (\%\; t\; (\mathsf{let}\; k = \overbrace{(\mathsf{call/cc}\; t\; (\lambda k\ldotp k))}^{(\{\mathsf{prophecy}\;t\;(\emptyset,\{\mathsf{replace}\;t\;\alpha^*\},\bot)\;\mathsf{obs}\;(\emptyset,\emptyset,I))\},\emptyset,I)}\; \mathsf{in}\; (\overbrace{\mathsf{event}[\alpha]}^{(\emptyset,\emptyset,\{\alpha\})}; (\overbrace{k\; ()}^{(\emptyset,\{\pblock{\mathsf{replace}\;t\;\alpha^*}{t}\},\bot)}))) ...)
    \]
    The effect of the prompt's body is the sequencing (with $\rhd$) of the three individual subterm effects written above the term fragments, writing $\alpha^*$ for the set of all finite traces consisting of only $\alpha$.  Simplifying the sequencing of the two right-most effects first will result in a control effect set $\{\pblock{\mathsf{replace}\;t\;(\{\alpha\}\rhd\alpha^*)}{t}\}$ (invoke the continuation after an $\alpha$ event, with aggregate behavior of some non-zero finite number of $\alpha$ events). 
    Simplifying again, the prophecy will observe this, and \textsc{T-Prompt} will check (via \textsc{V-Effects}) that
    $\punblock{\{\pblock{\mathsf{replace}\;t\;(\{\alpha\}\rhd\alpha^*)}{t}\}}{t}\sqsubseteq\punblock{\{\mathsf{replace}\;t\;\alpha^*\}}{t}$,
    which is equivalent to checking
    ${\{{\mathsf{replace}\;t\;(\{\alpha\}\rhd\alpha^*)}\}}\sqsubseteq{\{\mathsf{replace}\;t\;\alpha^*\}}$,
    which is true because $\{\alpha\}\rhd\alpha^*\sqsubseteq\alpha^*$ in $\mathcal{T}(\Sigma)$ (the set of non-empty finite traces containing only $\alpha$ is a subset of the set of possibly-empty finite traces containing only $\alpha$).

    Finally, \textsc{T-Prompt} must give a type and effect to the prompt expression itself. The underlying effect must include (1) the body's underlying effect, (2) the effect of any possible paths through the body that abort to the local handler and execute it (including those resulting from invoking continuations up to this prompt prior to aborting), and (3) the effect of any possible paths through the body that invoke one or more continuations up to this prompt before returning normally.  (1) is simply $\opt{Q}$. (2) is the result of sequencing any abort prefix for $\ell$ in the (unblocked) control effect set (which will incorporate aborts resulting from continuation invocation as well).  (3) is simply the set of replacement effects in the (unblocked) $C$.  (2) and (3) are computed via a projection operator $-|^{Q_h}_\ell$ applied to the unblocked control effect set, defined in Figure \ref{fig:aux}.  The superscript on the operator is the underlying effect of the handler (constrained to have no control effects), and the subscript is the choice of relevant prompt tag.  The resulting set of underlying effects is joined together with the body's underlying effect.

    The prophecy and control effect sets of the overall prompt should propagate prophecies and control effects targeting other prompts, and remove those targeting the prompt at hand.  To this end, the conclusion of \textsc{T-Prompt} unblocks both sets and then removes (1) prophecies related to this prompt (which have now been validated) and (2) control effects related to this prompt (which have been incorporated into the prompt's underlying effect, because they address control behaviors scoped to this prompt).  The (unblocked) control effects are simply filtered with $-\setminus\ell$ (Figure \ref{fig:aux}), which retains only basic control effects targeting other tags and control effects still blocked until other tags. ({Thus, the nested control effect from Use Case \ref{use:tags} appears blocked in the effect of the innermost prompt.})
    The prophecy set is filtered similarly, but the filtered results are also transformed --- the remaining observations must be adapted to model the changes to effects from going ``through'' this prompt.  Thus the filtering operation $-\setminus^Q\ell$ on prophecy sets (Figure \ref{fig:aux}) selects those prophecies related to other prompts, then recursively transforms their observations --- recursively filtering (and transforming) the observed prophecy and control effect sets, and joining transformations of their content into the observations' underlying effect, just as in the conclusion of \textsc{T-Prompt} itself.
    This is important due to interactions between tags: a prophecy to an outer prompt may observe in its context aborts to an inner prompt: $-\setminus^Q\ell$ joins such abort prefixes with the handler that would run, and joins that into the underlying effect of the prophecy.
    \looseness=-1

    \ifTR
    \subsection{$\mathcal{C}(Q)$ is Almost an Effect Quantale}
    \label{sec:modulo}
    Figure \ref{fig:directPO} defines an explicit preorder on continuation effects rather than inheriting one directly from $\sqcup$ as effect quantales do; this essentially permits certain underlying effects tracked by the continuation effects to be over-approximated by ``greater'' effects from the underlying effect quantale, rather than from the naive partial order induced by tracking sets for prophecies and control effects.  
    For example, this permits a closure with a latent effect indicating it aborts after underlying effect $A$ to be passed as an argument whose formal parameter type has a latent effect indicating the closure should abort after underlying effect $B$ as long as $A\sqsubseteq B$.\footnote{Inheriting $\sqsubseteq$ from $\sqcup$ would require the parameter type to list both $A$ and $B$.}
    \ifTR
    This is preferable to requiring the formal parameter to declare multiple control effects for both $A$ and $B$ as acceptable, when any pre-abort computation with effect $A$ also already has effect $B$; this relaxation reflects natural subsumption of underlying effects into the continuation-aware effects.\footnote{We could instead induce a partial order from the join, as we do with the underlying effect quantale, but this would be too ``picky'': code aborting after $A$ should count as code aborting after $B$, when $A\sqsubseteq B$ in the underyling effect quantale.}
    \fi
    The figure also defines an equivalence relation $\approx$ as a subset of mutual over-approximation: two effects are equivalent if each over-approximates the other \emph{or} both contain underlying $\top$, indicating that both contain a type error with respect to the underlying effect quantale.  

    \ifTR
    The most natural characterization of continuation effects as a generalization of effect quantales is as an \emph{effect quantale modulo equivalence}:
    \begin{definition}[Effect Quantale Modulo Equivalence]
        \label{def:eqme}
        An \emph{effect quantale modulo equivalence} is a structure $(E,\approx,\sqcup,T,\rhd,I)$ where:
        $\approx$ is an equivalence relation on $E$;
        $(E,\sqcup)$ is a join semilattice up to $\approx$, with greatest element (equivalence class) $T$;
        $(E,\rhd,I)$ is a monoid up to $\approx$;
        membership in $T$ is propagated by sequencing ($\forall x,t\ldotp t\in T\Rightarrow (x\rhd t)\in T \land (t\rhd x)\in T$);
        and $\rhd$ distributes over $\sqcup$ on both sides ($a\rhd(b\sqcup c)\approx(a\rhd b)\sqcup(a\rhd c)$ and
        $(a\sqcup b)\rhd c \approx (a\rhd c)\sqcup(b\rhd c)$).
    \end{definition}
    
    In particular, $\rhd$ is associative (up to equality) with unit element $(\emptyset,\emptyset,I_Q)$; $\sqcup$ is commutative and associative (up to equality), and always produces an element in the equivalence class (by $\approx$) of the least upper bound according to $\sqsubseteq$ (which is used in $\approx$ in such a way that the partial order induced by $\sqcup$ is compatible with the explicit $\sqsubseteq$ in Figure \ref{fig:directPO}).
    $\approx$ is not degenerate unless the underlying effect quantale is.
    $\rhd$ and $\sqcup$ respect $\approx$ and the distributivity requirements.
    $T$ is simply the set of all 
    Note that $\sqsubseteq$ does not respect equivalence (it is used to define equivalence), but we will later treat it as such due to side-conditions forbidding effects containing underlying $\top$.

    We will generally refer to $\mathcal{C}(Q)$ as simply an effect quantale, rather than an effect quantale modulo equivalence, because the latter can always be collapsed to the former:
    \begin{lemma}[Effect Quantales Modulo Equivalence Collapse]
        \label{lem:eqme_collapse}
        Given an effect quantale modulo equivalence $(E,\approx,\sqcup,T,\rhd,I)$, the result of quotienting $E$ by the equivalence relation --- $(E/\approx,\sqcup,T,\rhd,\{I\})$ --- is an effect quantale.
    \end{lemma}
    \else
    $\sqcup$ and $\rhd$ (and the type rules introduced shortly) all respect this equivalence relation, and in fact a true effect quantale can be constructed by quotienting $\mathcal{C}(Q)$ by $\approx$ (the elements become equivalence classes with respect to $\approx$; the details appear in the supplementary material~\cite{controltr}).
    Because of this, we will continue to refer to $\mathcal{C}(Q)$ as an effect quantale despite its relaxation by an equivalence relation.
    \fi
    In principle it would be possible to modify $\mathcal{C}(Q)$ to produce a standard effect quantale directly, which directly incorporates the sorts of over-approximation we want, rather than dealing with the equivalence relation (e.g., by computing a single over-approximation of each abort or replace effect, and explicitly lifting underlying $\top$ into a distinguished top element).  But this would add further significant complexity to metatheory, exposition, and derivation of type rules.
    \fi

    \ifTR
    \subsection{Design Notes}
    Readers familiar with other effects dealing with enclosing contexts~\cite{neamtiu2008contextual} might wonder why we are placing this prophecy information, and the machinery for validation, into effects rather than modifying the shape of the type judgment to, for example, include the contextual effect out to the next prompt (for each tag) and check at continuation invocation that the ``outer context effect'' was less than the assumed latent effect.  Such an alternative is a natural reflection of the duality observed by Crary et al.~\cite{crary1999typed}, that effect systems can be framed as ``pushing contextual information down'' into judgments of subterms, or ``passing analyzed behavior up'' to enclosing contexts as we do.  There are several reasons.  First, this approach leads to further invasive changes when the \lstinline|call/cc| is hidden inside a function.  It must be possible to typecheck function bodies in isolation, and reuse them in multiple contexts; taking a ``push the context into the type judgment'' approach would require regular function types carrying their own assumptions about ``acceptable contexts'' and this would likely require similar pervasive changes.  In contrast, one of the points of the original effect quantales work was that this sort of context management could be pushed into the structure of effects themselves --- essentially what we have done.  Second, doing this allows us to pose the entire tracking system as an effect quantale, which in addition to being theoretically pleasant (a transformation on effect quantales) has the practical advantage of allowing reuse of the iteration work for effect quantales, which becomes important when deriving rules for macro-expressed control operators and control flow operations. Making bespoke modifications to the shape of the type judgment would require generating a new approach to iteration in the presence of control operators, which would be non-trivial even with the insights of prior work to guide us.  We prefer to directly reuse proven constructions.
    \fi

\section{Iterating Continuation Effects}
    \label{sec:contiter}
    Prior work on effect quantales~\cite{ecoop17,quantalesjournal} introduced the notion of lax iterability to introduce a loop iteration operator, as outlined in Section \ref{sec:bg}.
    We would like to reuse this operator construction for two reasons.  First, we would like to check that if we macro-express loop constructs and derive rules for them as we proposed earlier, that they are consistent with manually-derived rules from prior work, which use the iteration operator.  Second, the iteration operator has properties that make it useful for solving recursive constraints over effects, such as those that arise in building derived rules for control flow constructs and control operators later in the paper.
    Of course, lax iterability and the construction are defined on standard effect quantales,\footnote{Lax iterability is defined for a version of effect quantales using partial operators~\cite{quantalesjournal} instead of the distinguished $\top$ used here to simulate partiality.  But the definitions simplify to those used here when given total operations and using side-conditions to reject effects containing $\top$.}
\ifTR
    not the effect quantales modulo equivalence, which is the structure we give.  Fortunately these are closely related.  Lemma \ref{lem:eqme_collapse} gives a standard effect quantale for each effect quantale modulo equivalence.  Lemma \ref{lem:eqme_collapse}'s quotient construction preserves lax iterability of the underlying effect quantale, meaning the existing iteration construction applies to the quotient effect quantale.  This construction takes each effect $X$ to the least subidempotent ($y\rhd y\sqsubseteq y$) effect greater than both $X$ and $I$; lax iterability ensures the least such element always exists.  Since we are applying this to a quotient construction, this naturally takes the form of an operation on elements of the effect quantale modulo equivalence, which respects the equivalence relation $\approx$.
\fi

    \begin{theorem}[Lax Iterability with Continuations]
        \label{thm:laxiter}
        For a laxly iterable underlying effect quantale $Q$, the effect quantale $\mathcal{C}(Q)/\approx$ is also laxly iterable, with the closure operator given by lifting the following operator
        from elements of $\mathcal{C}(Q)$ to the corresponding equivalence class.
        \[(P,C,\opt{Q})^*=(\bigcup_{i\in\mathbb{N}}P\blacktriangleright(P,C,\opt{Q})^i,\opt{Q}^*\rhd C, \opt{Q}^*)\]
    \end{theorem}
    \ifTR
    \begin{proof}
        To prove this is the closure operator, we must prove that the right hand side is the minimum subidempotent element greater than both $I$ and $(P,C,Q)$, or more precisely that it respects $\approx$ and when applied to equivalence classes of this quotient construction it gives the least equivalence class with respect to $\sqsubseteq$.

        Subidempotence follows directly from the infinite union of prophecies, and the properties of of the underlying effect quantale's iteration.  Being greater than the original input and identity is straightforward. So it remains to prove minimality.
        By contradiction.  Assume there is a lesser such element than the above, $(P',C',\opt{Q'})$.  The fact that this is supposedly less than the result of $(P,C,\opt{Q})^*$ defined above requires that each component be ordered less, and at least one of these component-wise inequalities must be strict (else they could be equivalent:
        \begin{itemize}
            \item $\opt{Q'}\sqsubseteq \opt{Q}^*$
            \item $C'\sqsubseteq \opt{Q}^*\rhd C$
            \item $P'\sqsubseteq \bigcup_{i\in\mathbb{N}}P\blacktriangleright(P,C,\opt{Q})^i$
        \end{itemize}
        The first is only possible if $\opt{Q'}=\opt{Q}^*$, since $\opt{Q}^*$ is the minimal subidempotent within the $\bot$-extended underlying effect quantale.
        The second degenerates to an equality requirement for the same reason, so the final constraint must be a strict $\sqsubset$.
        The final component-wise constraint requires every prophecy in $P'$ to be over-approximated by some prophecy in the infinite union term.  
        So for the assumed $(P',C',\opt{Q'})$ to be strictly less than $(P,C,\opt{Q})^*$, either at least one prophecy in $P'$ is strictly over-approximated wherever it is over-approximated in the infinite union, or at least one prophecy in the infinite union is not necessary to over-approximate an element of $P'$.  
        The latter case is straightforwardly not possible: since the infinite union contains exactly all prophecies obtainable by finite iteration of $(P,C,\opt{Q})$, omitting any such prophecy from $P'$ would mean that for some $m$, $(P',C',\opt{Q'})\rhd(P,C,\opt{Q})^m$ would contain a prophecy not contained in $P'$, even though because $(P',C',\opt{Q'})$ is an iteration result and subidempotent, it should be the case $(P',C',\opt{Q'})\rhd(P,C,\opt{Q})^m\sqsubseteq(P',C',\opt{Q'})$; so this is not possible.
        The former case is similar: it is not possible for a prophecy $p\in P'$ to be only \emph{strictly} over-approximated by any larger elements of the infinite union, because that would require $p$ to be a prophecy that could not be generated by finite iteration of $(P,C,\opt{Q})$. 
    \end{proof}
    \fi

    \ifTR
    The requirements for lax iterability dictate exactly the operator above for iteration, but we can consider the relationship between this operator and various representative cases that may arise in typechecking to understand why this is not only mathematically necessary, but actually corresponds in a sensible way to the runtime semantics.
    \fi

    \ifTR
    We can build some intuition for the operator above by considering two special cases, then discussing the general case.

    \begin{example}[Control-Free Iteration]
        In the case where an iterated effect has no (escaping) prophecies or control effects, it behaves exactly as the iteration from the $\bot$-extended underlying effect quantale: $(\emptyset,\emptyset,\opt{Q})^*=(\emptyset,\emptyset,\opt{Q}^*)$.
    \end{example}
    \begin{example}[Prophecy-Free Iteration]
        In the case where the prophecies are empty --- where there are no unresolved continuation captures (such as throwing exceptions from within a loop) --- the results correspond to the intuitive idea that the control effects would occur after 0 or more non-exceptional runs of the underlying effect --- that any exceptional control action in $C$ would occur only after repeating $\opt{Q}$ some (possibly-zero) number of times: $(\emptyset,C,\opt{Q})^*=(\emptyset,\opt{Q}^*\rhd C, \opt{Q}^*)$.
    \end{example}
    \fi
    While these examples ``merely'' drop certain components of Theorem \ref{thm:laxiter}, it helps to work from the simplest case up to the more complex versions, since the examples above correspond intuitively to various execution paths.
    The infinite union in the prophecy set is the most subtle part of the operation to explain.
    Consider an expression with the structure
    $ \mathsf{while}~c~(\ldots~(\mathsf{call/cc}~t\ldots)~\ldots) $:
    Assume the tag $t$ for the continuation captured inside the loop does not occur elsewhere inside the loop --- in particular, that the captured continuation would extend \emph{outside} the loop.  Considering the runtime execution, in some sense the prophecy captured by the \emph{first} loop iteration must predict not only the regular execution and exceptional executions of future iterations, but even the need for more prophecies to be generated by the \lstinline|call/cc|s in future iterations as well!  This is why the set of prophecies must still be sequenced with some form of themselves, rather than just some subset.
    During static typechecking, we must therefore conservatively overapproximate the number of iterations following a prophecy.  It may be 0, 1, 2, \ldots or any number.  So the approximation must consider all of those possibilities, hence the infinite union of finite repetitions following the prophecies. This requires prophecy sets to be possibly-infinite, but only countably so.

\section{Type Safety}
\label{sec:summary_soundness}
We have proven syntactic type soundness for the type system presented in Section \ref{sec:formal}.
We continue to reuse Gordon's parameterization for soundness~\cite{ecoop17,quantalesjournal}, making the proof generic over a choice of abstract states (ranged over by $\sigma$) and related parameters subject to some restrictions.
Progress is uninteresting (if primitives satisfy progress), in the sense that effects play no essential role (they are merely ``pushed around'' and the proof looks otherwise like standard progress proofs).
Preservation is similar to the common formulation for syntactic type soundness results of sequential effect systems~\cite{tldi12,ecoop17,skalka2008types,Skalka2008}. 
It follows from single-step preservation: informally, for a well-typed runtime state $\sigma$ and term $e$, if $\sigma;e\xrightarrow{q}\sigma';e'$, then $\sigma'$ and $e'$ are also well-typed, and moreover if the static effects of $e$ and $e'$ are $\chi$ and $\chi'$ respectively, then $(\emptyset,\emptyset,q)\rhd\chi'\sqsubseteq\chi$ --- that is, sequencing the actual effect of the reduction with the residual effect of the reduced expression is soundly bounded by the effect of the original expression. This is the typical form of syntactic type safety proofs for sequential effect systems~\cite{tldi12,flanagan2003atomicity}.

Explaining the formal statement requires explaining the full details of parameterization.  
\ifTR
Full details appear in the appendix, and here for brevity
\fi
we offer the formal statement of single-step preservation specialized for our running example of $\mathcal{T}(\Sigma)$ and $\mathsf{event}[-]$ with trivial (i.e., unit) state\footnote{Meaningful loops obviously require more meaningful state.} for which many of the conditions simplify to \textsf{True}:
\begin{corollary}[Single-Step Preservation for $\mathcal{T}(\Sigma)$]
    If $\Gamma\vdash e : \tau \mid \chi$ and $();e\xRightarrow{q}();e'$, then there exists a $\tau'<:\tau$ and a $\chi'$ such that $\Gamma\vdash e' : \tau' \mid \chi'$ and $q\rhd\chi'\sqsubseteq\chi$.
\end{corollary}

A key lemma for soundness of the rule for \lstinline|call/cc| precisely relates prophecy observations to typing continuations. Informally, we prove that for a term $E[e]$, if the effect of $e$ contains a prophecy $\mathsf{prophecy}\;\ell\;\chi\leadsto\tau;\mathsf{obs}\;(\emptyset,\emptyset,I)$ and $E$ contains no prompts for tag $\ell$, then (1) the effect of $E[e]$ contains a prophecy  $\mathsf{prophecy}\;\ell\;\chi\leadsto\tau\;\mathsf{obs}\;\chi'$ (accumulating some observation), and (2) plugging an appropriately-typed value $v$ into $E$ produces a term $E[v]$ with static effect $\chi'$ (exactly matching the observation). The assumptions of the lemma are exactly the conditions when considering \textsc{E-CallCC} in the preservation proof: $E[e]$ is the body of the delimiting prompt and $e$ is a use of \lstinline|call/cc|, so by \textsc{T-CallCont} an appropriate prophecy exists in $e$'s effect ensuring the prophecy is gives sound latent effect for typing the captured continuation.
\looseness=-1

\begin{remark}[{Syntactic vs.\ Semantic Soundness}]
Our proof imparts no semantic meaning to effects beyond syntactically relating the dynamic and static effects --- it does not check that a certain effect enforces what it is meant to (e.g., deadlock freedom), unless like some finite trace effects~\cite{skalka2008types,Koskinen14LTR} the relation between static and dynamic effects is already the intent.
This is common to any syntactic type safety approach for generic effects~\cite{marino09}.
Gordon~\cite{quantalesjournal} has extended his proof approach for semantic pre- and post-condition type properties, such as ensuring locking effects accurately describe lock acquisition and release, but limited to safety properties; in principle this should be adaptable to our setting.
Denotational approaches to abstract effect systems~\cite{katsumata14,mycroft16,tate13,atkey2009parameterised} inherently give actual semantics, and therefore can ensure liveness properties.

\ifTR
Ignoring the syntactic nature of soundness leads to counterintuitive misunderstandings.  Consider an effect quantale with 3 elements --- $\mathsf{Total}\sqsubseteq\mathsf{Partial}\sqsubseteq\top$ --- intended to model total or partial computations.  If sequencing simply takes least-upper bound with respect to the partial order ($\rhd=\sqcup$), this is a valid effect quantale with $\mathsf{Total}$ as the identity.  But Gordon's iteration operator will set $\mathsf{Total}^*=\mathsf{Total}$, suggesting that infinite loops of ``$\mathsf{Total}$'' actions are $\mathsf{Total}$.  This is because the soundness proof does not account for what each effect should mean, and the syntactic effect $\mathsf{Total}$ is not semantically tied to termination.
Notice that this effect quantale is isomorphic to one that simply expresses whether or not a computation uses reflection: $\mathsf{NoReflection}\sqsubseteq\mathsf{Reflection}\sqsubseteq\top$.
\fi
\end{remark}

    \section{Deriving Sequential Effect Rules}
    \label{sec:deriving_rules}
    Section \ref{sec:formal} developed the core type rules which give sequential effects to programs making direct use of tagged delimited control.
    As we have discussed, most programs do not use the full power of delimited control, and instead use only control flow constructs or weaker control operators.
    This section uses the new type type-and-effect rules to \emph{derive} sequential effect rules for a range of control flow constructs and weaker control operators macro-expressed in terms of prompts.
    \looseness=-1

\ifTR
    Our examples fall into two groups.  First, we consider checking consistency of \emph{derived rules} for typical control flow constructs with those hand-designed in prior work, for infinite loops (Section \ref{sec:infloop}) and while loops (Section \ref{sec:while}). 
    Second, we consider derived rules for constructs that are common in most programming languages, yet \emph{never addressed in prior work on sequential effect systems}: exceptions (Section \ref{sec:realexceptions}).
    Finally, we consider expressing a weaker control operator, a form of generator close to an encoding given by Coyle and Crogono~\cite{Coyle:1991:BAI:122179.122181}.
\fi

    In each case, we give a derived type rule for the construct of interest.  While we are most explicit in Section \ref{sec:infloop}, in each case our process for deriving the rule is the following:
    \begin{enumerate}
        \item Assume closed typing derivations for subexpressions (e.g., loop bodies)
        \item Apply the type rules from Section \ref{sec:formal} to give a closed-form rule for the macro's expansion to be well-typed under the assumed subexpression types.  Typically these have several undetermined choices for metavariables representing effects, with non-trivial constraints to close the typing derivation.
        \item Simplify the type rule by giving solutions to the constrained-but-undetermined metavariables in terms of the subexpressions' effects.  This gives type rules that are simpler, and possibly less general, but given entirely in terms of the subexpressions' effects.  The simplifications are typically involve rewriting by the laws satisfied by $\mathcal{C}(Q)$, and using the iteration operator from prior work~\cite{ecoop17,quantalesjournal} to solve recursive constraints on undetermined effects.
        \looseness=-1
    \end{enumerate}

    \subsection{Infinite Loops}
    \label{sec:infloop}
    Consider a simple definition of an infinite loop using the constructs we have derived here:
    \[
        \llbracket\mathsf{loop}~e\rrbracket = (\%\;\ell\;(\mathsf{let}\;cc = (\mathsf{call/cc}~\ell~(\lambda k\ldotp k))~\mathsf{in}~(\llbracket{e}\rrbracket; cc\;cc))\;(\lambda\_\ldotp \mathsf{tt}))
    \]
    The term above executes $e$ repeatedly, forever (assuming $e$ does not abort).
    Thus, its effect \emph{ought} to indicate that $e$'s effect, which we take to be $(\emptyset,\emptyset,Q_e)$,\footnote{Note the non-$\bot$ underlying effect; well-typed expressions do not have degenerate effects.} is repeated arbitrarily many times.  We take this expansion as the body of a macro $\llbracket\mathsf{loop}~e\rrbracket$.
    This program can be well-typed in our system, with an appropriate effect (assuming the underlying effect quantale is \emph{laxly iterable} per Section \ref{sec:bg}).  The body of the \textsf{call/cc} is pure, but for the expression to be well-typed, the \textsf{call/cc}'s own effect must prophecize some effect $(\emptyset,C_p,\opt{Q_p})$ of the enclosing continuation up to the prompt for $\ell$ (because no \lstinline|call/cc| occurs in the continuation of another, the prophecy set can be empty).
\ifTR
    \[
        \begin{array}{c}
        \underbrace{
        (
            \underbrace{
                \mathsf{let}\;cc = \underbrace{(\mathsf{call/cc}~\ell~(\lambda k\ldotp k))}_{(\{\mathsf{prophecy}~\ell~(\emptyset,C_p,\opt{Q_p})\leadsto\mathsf{unit}~\mathsf{obs}~(\emptyset,\emptyset,I)\},\emptyset,I)}~\mathsf{in}~(\underbrace{\llbracket{e}\rrbracket}_{(\emptyset,\emptyset,Q_e)}}_{(\{\mathsf{prophecy}~\ell~(\emptyset,C_p,\opt{Q_p})\leadsto\mathsf{unit}~\mathsf{obs}~(\emptyset,\emptyset,Q_e)\},\emptyset,Q_e)};
                 \underbrace{cc\;cc}_{(\emptyset,\pblock{C_p}{\ell}\cup((Q_e\rhd \opt{Q_p})\rhd\{\mathsf{replace}~\ell:I\leadsto\mathsf{unit}\}),\bot)}))
        }_{(\{\mathsf{prophecy}~\ell~(\emptyset,C_p,\opt{Q_p})\leadsto\mathsf{unit}~\mathsf{obs}~(\emptyset,C,\bot)\},C,\bot)}
        \\
        \textrm{where}~C=(Q_e\rhd \pblock{C_p}{\ell})\cup((Q_e\rhd \opt{Q_p})\rhd\{\mathsf{replace}~\ell:I\leadsto\mathsf{unit}\})
        \end{array}
    \]
\fi

    The right-accumulator of the prophecy effect, initially $(\emptyset,\emptyset,I)$, eventually accumulates a control set $(Q_e\rhd \pblock{C_p}{\ell})\cup((Q_e\rhd \opt{Q_p})\rhd\{\mathsf{replace}~\ell:I\leadsto\mathsf{unit}\})$ and underlying effect $\bot$, because between capturing the continuation and the prompt, the program evaluates $e$ (underlying effect $Q_e$) and then invokes the captured continuation (prophecized effect $(\emptyset,C_p,\opt{Q_p})$, underlying effect $\bot$). This is also the resulting control effect set for the body; we will refer to it as $C$. 
    The type rule for the prompt itself removes all $\ell$-related prophecies and control effects, leaving both empty (since we assume no control effects escape $e$, $C_p$ should only contain $\ell$-related effects, while the prophecy set contains the single prophecy from the \lstinline|call/cc|).
    For the underlying effect, \textsc{T-Prompt} joins the immediate underlying effect $Q_e$ (from the overall judgment, not the prophecy) with all $\ell$-related behaviors in $C$ --- $e$ has no escaping control effects, and the macro-expanded loop contains no aborts, so $C_p$ ought to have only \textsf{replace} effects, meaning $C$ contains only replace effects, and $Q_e\sqcup\bigsqcup \punblock{C}{\ell}|^I_\ell$ will join the underlying effect of all continuations invoked by the body.
    \textsc{T-Prompt} also performs some checking of result types (which all hold trivially since all types involved are \textsf{unit}), and prophecy validity checks that yield constraints we can solve to derive a closed-form type rule for the loop.
    
    Completing a typing derivation with final underlying effect $Q_\ell=Q_e\sqcup\bigsqcup \punblock{C}{\ell}|^I_\ell$ is possible given the solutions to the effect-related constraints imposed by \textsf{validEffects}:
        $\bot\sqsubseteq \opt{Q_p}$, and 
        \mbox{$(Q_e\rhd \pblock{C_p}{\ell})\cup((Q_e\rhd \opt{Q_p})\rhd\{\mathsf{replace}~\ell:I\leadsto\mathsf{unit}\}) \sqsubseteq C_p$}.
    These could be read off a hypothetical derivation 
    \ifTR
    (for example, see Figure \ref{fig:infloop} in Appendix \ref{apdx:dloop})
    \fi
    yielding the derived rule
    \vspace{-0.5em}
    \begin{mathpar}
        \inferrule{
            \Gamma\vdash e : \tau \mid (\emptyset,\emptyset,Q_e)\quad
            \bot\sqsubseteq \opt{Q_p}\quad
            (Q_e\rhd \pblock{C_p}{\ell})\cup((Q_e\rhd \opt{Q_p})\rhd\{\mathsf{replace}~\ell:I\leadsto\mathsf{unit}\}) \sqsubseteq C_p
        }{
            \Gamma\vdash \llbracket\mathsf{loop}~e\rrbracket : \mathsf{unit} \mid 
            (\emptyset,\emptyset,Q_e\sqcup\bigsqcup \punblock{C}{\ell}|^I_\ell)
        }
    \end{mathpar}
    However, this rule is more complex than we would like for a simple infinite loop (note we have not expanded $C=(Q_e\rhd \pblock{C_p}{\ell})\cup((Q_e\rhd \opt{Q_p})\rhd\{\mathsf{replace}~\ell:I\leadsto\mathsf{unit}\})$), and also exposes details of the continuation-aware effects --- which is undesirable if the goal is to derive closed rules for using the loop by itself, without developer access to full continuations, and there is an additional requirement that the prophecy used in the derivation is non-trivial (from \textsc{T-Prompt}).
    These constraints can be satisfied by $\opt{Q_p}=Q_e^*$ (thus not $\bot$, ensuring a non-trivial prophecy), with $C_p=\{\mathsf{replace}~\ell:Q_e^*\leadsto\mathsf{unit}\}$ (so $\pblock{C_p}{\ell}=C_p$).
    The choice for $C_p$ ensures than any ``unrolling'' of the loop to include any number of $Q_e$ prefixes (as generated by the left operand of the union in the last constraint) is in fact less than the replacement effect ($Q_e\rhd Q_e^*\sqsubseteq Q_e^*$).
    This then implies that $Q_e\sqcup\bigsqcup \punblock{C}{\ell}|^I_\ell\sqsubseteq Q_e\sqcup(Q_e^*)\sqsubseteq Q_e^*$, by properties of Gordon's iteration operator~\cite{ecoop17,quantalesjournal} (Section \ref{sec:bg}).  Assuming $cc\not\in\Gamma$ (or hygienic macros) and applying subsumption, this leads us to the pleasingly simple derived rule:
    \begin{mathpar}
        \inferrule*[left=D-InfLoop]{
            \Gamma\vdash e : \tau \mid (\emptyset,\emptyset,Q_e)
        }{
            \Gamma\vdash \llbracket\mathsf{loop}~e\rrbracket : \mathsf{unit} \mid (\emptyset,\emptyset,Q_e^*)
        }
    \end{mathpar}

    \ifTR
        \subsection{While Loops}
    \label{sec:while}
    While loops can similarly be macro-expressed via continuations\footnote{An alternative is to capture the continuation before any conditional. This also works, but the (sound) rule derived from this does not match that of prior work~\cite{flanagan2003atomicity,flanagan2003tldi,ecoop17}.}:
    \[
        \begin{array}{l}
        \llbracket\mathsf{while}~c~e\rrbracket =\\
        (\%\;\ell\;(\mathsf{if}(c)~(\mathsf{let}~cc=(\mathsf{call/cc}~id)~\mathsf{in}~(e;\mathsf{if}(c)~(cc~cc)~(tt)))~(tt))~
        \\
        ~~~~(\lambda\_\ldotp tt))   
        \end{array}
    \]

    Assume no other control effects escape $e$ and $c$ (i.e., $\Gamma\vdash e : \tau \mid (\emptyset,\emptyset,Q_e)$ and $\Gamma\vdash c : \mathsf{bool} \mid (\emptyset,\emptyset,Q_c)$).

    \ifTR
    Writing only the underlying effects as shorthand for the case where no control behaviors appear, the effect of the prompt's body is detailed and simplified in Figure \ref{fig:whilebody}.
    \begin{figure*}\scriptsize
    \[
        \begin{array}{ll}
            & Q_c\rhd(((\{\mathsf{prophecy}~\ell~(\emptyset,C_p,Q_p)\leadsto\mathsf{unit}~\mathsf{obs}~(\emptyset,\emptyset,I)\},\emptyset,I)\rhd Q_e\rhd Q_c\rhd \colorbox{lightgray}{$((\emptyset,C_p\cup\{\mathsf{replace}~\ell:Q_p\leadsto\mathsf{unit}\},I)\sqcup I)$}) \sqcup I)
            \\
            \equiv & Q_c\rhd(((\{\mathsf{prophecy}~\ell~(\emptyset,C_p,Q_p)\leadsto\mathsf{unit}~\mathsf{obs}~(\emptyset,\emptyset,I)\},\emptyset,I)\rhd \colorbox{lightgray}{$ Q_e\rhd Q_c\rhd (\emptyset,C_p\cup\{\mathsf{replace}~\ell:Q_p\leadsto\mathsf{unit}\},I)$}) \sqcup I)
            \\
            \equiv & Q_c\rhd(((\{\mathsf{prophecy}~\ell~(\emptyset,C_p,Q_p)\leadsto\mathsf{unit}~\mathsf{obs}~(\emptyset,\emptyset,I)\},\emptyset,I)\rhd (\emptyset,(Q_e\rhd Q_c\rhd C_p)\cup\{\mathsf{replace}~\ell:(Q_e\rhd Q_c\rhd Q_p)\leadsto\mathsf{unit}\},Q_e\rhd Q_c)) \sqcup I)
            \\
            \equiv & Q_c\rhd((\{\mathsf{prophecy}~\ell~(\emptyset,C_p,Q_p)\leadsto\mathsf{unit}~\mathsf{obs}~(\emptyset,C,Q)\},C,Q) \sqcup I)
            \\&\qquad\qquad\textsf{where}~C\overset{def}{=}(Q_e\rhd Q_c\rhd C_p)\cup\{\mathsf{replace}~\ell:(Q_e\rhd Q_c\rhd Q_p)\leadsto\mathsf{unit}\} \qquad Q\overset{def}{=}Q_e\rhd Q_c
            \\
            \equiv & (\{\mathsf{prophecy}~\ell~(\emptyset,C_p,Q_p)\leadsto\mathsf{unit}~\mathsf{obs}~(\emptyset,C,Q)\},Q_c\rhd C,Q_c\rhd Q) \sqcup Q_c
            \\
            \equiv & (\{\mathsf{prophecy}~\ell~(\emptyset,C_p,Q_p)\leadsto\mathsf{unit}~\mathsf{obs}~(\emptyset,C,Q)\},Q_c\rhd C,(Q_c\rhd Q)\sqcup Q_c)
        \end{array}
    \]
    \caption{Simplifying the body effect for a while loop (without control effects escaping subexpressions).}
    \label{fig:whilebody}
    \end{figure*}

    As in the infinite loop case, this body effect along with \textsc{T-Prompt} imposes a set of constraints which, if satisfied, allows the while loop to be well-typed in our type system.  In short, a derived type rule requires some underlying effect $Q_\ell=((Q_c\rhd Q_e\rhd Q_c)\sqcup Q_c)\sqcup\bigsqcup \punblock{(Q_c\rhd C)}{\ell}|^I_\ell$, and a choice for the prophecized control effect set $C_p$ and underlying $Q_p$ where:\\
    \begin{itemize}
        \item $Q_e\rhd Q_c\sqsubseteq Q_p$ (since $Q_e\rhd Q_c\rhd(\bot\sqcup I) = Q_e\rhd Q_c$)
        \item $(Q_e\rhd Q_c\rhd \pblock{C_p}{\ell})\cup((Q_e\rhd Q_c\rhd Q_p)\rhd\{\mathsf{replace}~\ell:I\leadsto\mathsf{unit}\})\sqsubseteq C_p$
    \end{itemize}
    (We have jumped to assuming $Q_p$ is defined, not $\bot$, since it must be greater than $I$; this ensures non-triviality of the prophecy involved.)
    This leads to another complex derived rule, which can be further simplified:
    \begin{mathpar}
        \inferrule{
            \Gamma\vdash c : \mathsf{bool} \mid (\emptyset,\emptyset,Q_c)\\
            \Gamma\vdash e : \tau \mid (\emptyset,\emptyset,Q_e)\\
            Q_e\rhd Q_c\sqsubseteq Q_p\\
            (Q_e\rhd Q_c\rhd \pblock{C_p}{\ell})\cup((Q_e\rhd Q_c\rhd Q_p)\rhd\{\mathsf{replace}~\ell:I\leadsto\mathsf{unit}\})\sqsubseteq C_p
        }{
            \Gamma\vdash \llbracket\textsf{while}~c~e\rrbracket : \mathsf{unit} \mid ((Q_c\rhd Q_e\rhd Q_c)\sqcup Q_c)\sqcup\bigsqcup \punblock{(Q_c\rhd C)}{\ell}|^I_\ell
        }
    \end{mathpar}

    As in the infinite loop case, a simpler solution is available as long as the underlying effect quantale has an iteration operator, and because the \lstinline|callcc| does not capture other \lstinline|callcc|s, we start from the assumption that the prophecy predicts no other prophecies.
    In this case, we set 
    $Q_p=(Q_e\rhd Q_c)^*$, and $C_p=\{\mathsf{replace}~\ell:(Q_e\rhd Q_c)^*\leadsto\mathsf{unit}\}$ (again, $\pblock{C_p}{\ell}=C_p$).
    Then $Q_\ell$ simplifies using properties of effect quantales and the iteration operator:
    \[
        \begin{array}{rcl}
        Q_\ell& =&((Q_c\rhd Q_e\rhd Q_c)\sqcup Q_c)\sqcup\bigsqcup \punblock{(Q_c\rhd C)}{\ell}|^I_\ell\\
        &=& (Q_c\rhd ((Q_e\rhd Q_c)\sqcup I)\sqcup\bigsqcup \punblock{(Q_c\rhd C)}{\ell}|^I_\ell\\
        &=& Q_c\rhd (((Q_e\rhd Q_c)\sqcup I)\sqcup(Q_e\rhd Q_c)^*)\\
        &=& Q_c\rhd (Q_e\rhd Q_c)^*
        \end{array}
    \]
    \else
    A suitable rule can be derived by following the same approach as in Section \ref{sec:while}: using the assumed effects of the condition and body, give the effect of the prompt's body, and use the \textsf{validEffects} constraints of the prompt rule to derive a a closed rule with some constraints on effect prophecies.  These constraints can then be simplified by appealing to the assumed restrictions on $e$ and $c$, and the iteration operator of the underlying effect quantale.  For full details of the derivation, see the supplementary material.
    \fi

    This justifies the following derived rule:
    \begin{mathpar}
        \inferrule*[left=D-While]{
            \Gamma\vdash e : (\emptyset,\emptyset,Q_e)\\
            \Gamma\vdash c : (\emptyset,\emptyset,Q_c)
        }{
            \Gamma\vdash\llbracket\mathsf{while}~c~e\rrbracket : \mathsf{unit} \mid (\emptyset,\emptyset,Q_c\rhd(Q_e\rhd Q_c)^*)
        }
    \end{mathpar}

    This derived rule is an important consistency check against prior work.
    Setting aside the additional enforcement that no other control effects escape $e$ or $c$ (as they would not in languages where control operators were used only for loops), this is nearly identical to the given rule for typing while loops with sequential effects recalled in Section \ref{sec:bg}, as in prior work~\cite{ecoop17,flanagan2003atomicity,flanagan2003tldi} where the rule's soundness was proven directly.
    The only difference is the presence of the (empty) behavior sets for other control effects and prophecies from working in a continuation-aware effect quantale.

    \subsection{While Loops Without Subexpression Prophecies}
    \label{sec:whileabort}
    \ifTR
    Thus far we have only shown derived rules for simple loops.  The infinite and while loops are limited in ways beyond simply being expected based on prior work that addressed them directly: they also ignore the potential for ``improper'' nesting of control operators --- the cases studied thus far assume subexpressions that are not part of the macro expansion do not involve further unresolved control effects --- we have not seen the interaction of loops with aborts, invoking continuations for prompts \emph{outside} a loop, or prophecies from a loop body that need to observe the presence of iteration.  Here we remedy the first two limitations, and in the next subsection address iteration of loop bodies with \emph{arbitrary} control effects.

    We first study iteration under the assumption that loop components may have aborts or continuation invocations that would exit the loop.
    While this stops short of the full generality of our system, it still encompasses many languages whose control flow constructs and operators --- when expressed in terms of tagged delimited continuations --- do not nest the capture of continuations inside macro arguments.
    This includes loops and exceptions.\footnote{By contrast, generators are a counterexample, as we see in Section \ref{sec:simplegen}.}  In these cases, subexpressions of $c$ and $e$ that capture continuations occur under prompts that are themselves within $c$ or $e$, so $P_e$ and $P_c$ above would be $\emptyset$.  In this case, $\punblock{P_\chi}{\ell}\setminus\ell=\emptyset$. 
    Figure \ref{fig:simpl_Cchiell} simplifies $\punblock{C_\chi}{\ell}\setminus\ell$. For simplicity, we assume $Q_e$ and $Q_c$ are defined (not $\bot$); if $e$ or $c$ does have underlying effect $\bot$, its effect can be coerced by subtyping to an effect with non-$\bot$ underyling effect.  The derivation could be done explicitly permitting $\bot$ underlying effects.
    \begin{figure*}
    \[
        \begin{array}{rl}
        \punblock{C_\chi}{\ell}\setminus\ell\equiv & \punblock{[C_c\cup (Q_c\rhd C)]}{\ell}\setminus\ell\\
        \equiv & [C_c\cup (Q_c\rhd (C_e\cup(Q_e\rhd C_c)\cup(Q_e\rhd Q_c\rhd \pblock{C_p}{\ell})\cup\{\mathsf{replace}~\ell:\ldots\leadsto\mathsf{unit}\}))]\setminus\ell\\
        & \textrm{because}~\ell\not\in C_c\land \ell\not\in C_e~\textrm{and unblocking for $\ell$ cancels blocking for $\ell$}\\
        \equiv & [C_c\cup (Q_c\rhd (C_e\cup(Q_e\rhd C_c)\cup(Q_e\rhd Q_c\rhd C_p\setminus\ell)))]\\
        & \textrm{because}~\ell\not\in C_c\land \ell\not\in C_e~\textrm{and definition of $-\setminus\ell$}\\
        \equiv & [C_c\cup (Q_c\rhd (C_e\cup(Q_e\rhd C_c)\cup(Q_e\rhd Q_c\rhd ((Q_p\rhd C_e)\cup (Q_p\rhd Q_e\rhd C_c)))))]\\
        & \textrm{after substituting for $C_p$}\\
        \sqsubseteq & [C_c\cup (Q_c\rhd (C_e\cup(Q_e\rhd C_c)\cup((Q_p\rhd C_e)\cup (Q_p\rhd Q_e\rhd C_c))))]\\
        & \textrm{because}~Q_e\rhd Q_c\rhd Q_p\sqsubseteq (Q_e\rhd Q_c)^*\rhd Q_p\sqsubseteq Q_p~\textrm{for choice of}~Q_p=(Q_e\rhd Q_c)^*\\
        \sqsubseteq & [C_c\cup (Q_c\rhd ((Q_e\rhd C_c)\cup(Q_p\rhd C_e)\cup (Q_p\rhd Q_e\rhd C_c)))]\\
        & \textrm{because}~C_e=I\rhd C_e\sqsubseteq Q_p\rhd C_e
        \end{array}
    \]
    \caption{Simplifying aborting while loop body effect}
    \label{fig:simpl_Cchiell}
    \end{figure*}
    \else
    If we permit the condition and loop body to contain aborts (i.e., $C_c$ and $C_e$ may contain abort effects), the same process carried out above leads to a nautral extension of the rule.
    \fi
    The closed derived rule under these assumptions then becomes:
    \begin{mathpar}
        \inferrule[D-AbortingWhile]{
            l\not\in C_c \\ l\not\in C_e\\
            \Gamma\vdash c : (\emptyset,C_c,Q_c)\\
            \Gamma\vdash e : (\emptyset,C_e,Q_e)
        }{
            \Gamma\vdash\llbracket\mathsf{while}_\ell~c~e\rrbracket : \mathsf{unit} \mid (\emptyset,\punblock{C_\chi}\setminus\ell,Q_c\rhd(Q_e\rhd Q_c)^*)
        }
    \end{mathpar}
    The full expansion of the control effect set naturally corresponds to the four intuitive cases where control may exit the loop by means other than the condition resolving to false:
    \begin{itemize}
        \item The first time the condition is executed
        \item The first time the body is executed (after first executing the condition)
        \item A subsequent condition execution (after executing the initial condition, then repeating the body and condition some number of times)
        \item A subsequent body execution (after executing the initial condition, then repeating the condition and body some number of times, followed by a normal execution of the condition)
    \end{itemize}

    \subsection{Infinite Loops with Control}
    \label{sec:infcontrol}
    While it is possible to derive a fully general rule for while loops that admit the full flexibility of our effect system --- without restricting the effects of subexpressions --- the details are quite verbose.  We instead demonstrate the principles on the slightly simpler example of the infinite loop; following the same process for the while loop yields a similar but correspondingly more complex result (in particular, the control effect set is the same as in Section \ref{sec:whileabort}).  Showing the example of the infinite loop also demonstrates quite clearly that the fact that our transformation preserving lax iterability (Section \ref{sec:contiter}) is not only a theoretical nicety, but useful.

    \ifTR
    The effect $\chi$ of the prompt body in this case, for $\chi_e=(P_e,C_e,Q_e)$ and $\chi_{invk}=(\pblock{P_p}{\ell},\pblock{C_p}{\ell}\cup\{\mathsf{replace}~\ell:Q_p\leadsto\mathsf{unit}\},\bot)$, is
    given in Figure \ref{fig:infcontrol}.
    \begin{figure*}
    \[\scriptsize\begin{array}{r@{~\equiv~}l}
        \chi & (\{\mathsf{prophecy}~\ell~(P_p,C_p,Q_p)\leadsto\mathsf{unit}~\mathsf{obs}~(\emptyset,\emptyset,I)\},\emptyset,I)\rhd \chi_e\rhd \chi_{invk}\\
        & (P_e\cup\{\mathsf{prophecy}~\ell~(P_p,C_p,Q_p)\leadsto\mathsf{unit}~\mathsf{obs}~(P_e,C_e,Q_e)\},C_e,Q_e)\rhd \chi_{invk}\\
        & (\pblock{P_p}{\ell}\cup((P_e\cup\{\mathsf{prophecy}~\ell~(P_p,C_p,Q_p)\leadsto\mathsf{unit}~\mathsf{obs}~(P_e,C_e,Q_e)\})\blacktriangleright\chi_{invk}),C_e\cup\pblock{(Q_e\rhd C_p)}{\ell}\cup\{\mathsf{replace}~\ell:(Q_e\rhd Q_p)\leadsto\mathsf{unit}\},Q_e)\\
    \end{array}\]
    \caption{Prompt body effect for infinite loops with arbitrary nested control.}
    \label{fig:infcontrol}
    \end{figure*}
    This requires a choice of prophecy, satisfying (after simplifying $\chi_p$ above):
    \begin{itemize}
        \item $\bot\sqsubseteq Q_p$
        \item $\punblock{C_e\cup(Q_e\rhd \pblock{C_p}{\ell})\cup\{\mathsf{replace}~\ell:(Q_e\rhd Q_p)\leadsto\mathsf{unit}\}))}{\ell}
                \sqsubseteq \punblock{C_p}{\ell}$
        \item $ \punblock{P_e\blacktriangleright (\pblock{P_p}{\ell},\pblock{C_p}{\ell}\cup\{\mathsf{replace}~\ell:Q_p\leadsto\mathsf{unit}\},\bot)}{\ell} \sqsubseteq \punblock{P_p}{\ell}$
    \end{itemize}
    Unsurprisingly, the underlying and control constraints suggest a choice of $Q_p=Q_e^*$ and $C_p=(Q_e^*\rhd C_e)\cup\{\mathsf{replace}~\ell:(Q_e^*)\leadsto\mathsf{unit}\}$.
    So a solution to the prophecy constraint would be given by a solution to
    \[ \punblock{P_e\blacktriangleright (\pblock{P_p}{\ell},\pblock{(Q_e^*\rhd C_e)\cup\{\mathsf{replace}~\ell:(Q_e^*)\leadsto\mathsf{unit}\}}{\ell},\bot)}{\ell} \sqsubseteq \punblock{P_p}{\ell} \]
    which, since unblocking a prophecy recursively unblocks observed prophecies and control effects, and unblocking something blocked for the same tag cancels, simplifies to
    \[ P_e\blacktriangleright {P_p},(Q_e^*\rhd C_e)\cup\{\mathsf{replace}~\ell:(Q_e^*)\leadsto\mathsf{unit}\},\bot) \sqsubseteq \punblock{P_p}{\ell} \]
    We would like a solution without prophecies blocked for $\ell$, in which case we may solve
    \[ P_e\blacktriangleright (P_p,(Q_e^*\rhd C_e)\cup\{\mathsf{replace}~\ell:(Q_e^*)\leadsto\mathsf{unit}\},\bot) \sqsubseteq P_p \]
    because unblocking prophecies that are already unblocked has no effect.
    This is solved for 
    \[
        P_p = \bigcup_{i\in\mathbb{N}} P_e\blacktriangleright (P_e,(Q_e^*\rhd C_e)\cup\{\mathsf{replace}~\ell:(Q_e^*)\leadsto\mathsf{unit}\},\bot)^i
    \]
    which is also a solution to the original constraint.

    Solving for the final effect of the prompt expression itself:
    \begin{itemize}
        \item $Q = Q_e\sqcup\bigsqcup \punblock{C_{\chi}}{\ell}|^I_\ell = Q_e\sqcup (Q_e^*) = Q_e^*$
        \item $C = \punblock{C_{\chi}}{\ell}\setminus^I\ell = Q_e^*\rhd C_e$
        \item $P = \punblock{P_{\chi}}{\ell}\setminus^I\ell$\\
        $ =  (\punblock{\pblock{P_p}{\ell}}{\ell}\cup\punblock{((P_e\cup\{\mathsf{prophecy}~\ell~(P_p,C_p,Q_p)\leadsto\mathsf{unit}~\mathsf{obs}~(P_e,C_e,Q_e)\})\blacktriangleright\chi_{invk})}{\ell})\setminus^I\ell$\\
        $={P_p}\setminus^I\ell\cup \punblock{(P_e\blacktriangleright\chi_{invk})}{\ell}\setminus^I\ell$\\
         $={P_p}\setminus^I\ell\cup \punblock{(P_e\blacktriangleright(\pblock{P_p}{\ell},\pblock{C_p}{\ell}\cup\{\mathsf{replace}~\ell:Q_p\leadsto\mathsf{unit}\},\bot))}{\ell}\setminus^I\ell
        $\\
         $={P_p}\setminus^I\ell\cup (P_e\blacktriangleright(\punblock{\pblock{P_p}{\ell}}{\ell},\punblock{\pblock{C_p}{\ell}}{\ell}\cup\{\mathsf{replace}~\ell:Q_p\leadsto\mathsf{unit}\},\bot))\setminus^I\ell$\\
         $={P_p}\setminus^I\ell\cup (P_e\blacktriangleright({P_p},{C_p}\cup\{\mathsf{replace}~\ell:Q_p\leadsto\mathsf{unit}\},\bot))\setminus^I\ell$\\
        $={P_p}\setminus^I\ell\cup (P_e\blacktriangleright({P_p}\setminus^I\ell, Q_e^*\rhd C_e,Q_e^*))
        $
    \end{itemize}
    A thorough reader will notice that because
    ${P_p}\setminus^I\ell = \bigcup_{i\in\mathbb{N}}P_e\blacktriangleright(P_e,(Q_e^*\rhd C_e), Q_e^*)$,
    this is less than  $(P,C,Q)\sqsubseteq(P_e,C_e,Q_e)^*$ using Section \ref{sec:contiter}'s notion of iteration, licensing the following derived rule that accounts for arbitrary body effects despite its superficial simplicity:
    \else
    Following the approach taken for other rules --- gathering the constraints for the macro-expansion to be well-typed, then simplifying the constraints with judicious application of subsumption and finding (not-necessarily-minimal) solutions to constraints --- leads to the following derived rule relying on Section \ref{sec:contiter}'s lifted iteration operator:
    \fi
    \begin{mathpar}
        \inferrule*[left=D-FullInfLoop]{
            \Gamma\vdash e : \tau_e \mid \chi_e
        }{
            \Gamma\vdash \llbracket\mathsf{loop}~e\rrbracket : \mathsf{unit} \mid \chi_e^*
        }
    \end{mathpar}
    When the prophecy set is empty, this rule simplifies (by unfolding the definition of $(-)^*$) to \textsc{D-InfLoop} from Section \ref{sec:infloop}.  If the process above is followed for the while loop expansion used in Sections \ref{sec:while} and \ref{sec:whileabort}, the resulting rule simplifies to \textsc{D-AbortingWhile} and \textsc{D-While} when assuming the same constraints as in those sections.

    \subsection{Exceptions}
    \label{sec:realexceptions}

    In Use Case \ref{use:basicabort}, we informally considered typing a macro expansion of basic exception handling facilities:
    \[
        \begin{array}{r@{~=~}l}
        \llbracket\mathsf{try}\;e\;\mathsf{catch}\;\overline{C_i\Rightarrow e_i}^n\rrbracket & (\%~C_1~\ldots~(\%~C_n~e~e_n)~\ldots~e_1)
        \\
        \llbracket\mathsf{throw}_C e\rrbracket & \textrm{\lstinline|(abort C e)|}
        \end{array}
    \]
    The earlier discussion focused on the need to track what effects occurred before a throw vs.\ after.  Now that we have discussed the type rules for prompts and aborts, this mapping is clear, and derived rules for the simple case (no escaping control effects) follow easily from the rules for prompt and abort.  Assuming there is a designated prompt label $\ell_C$ corresponding to every thrown type $C$:

    \begin{mathpar}
        \inferrule[D-TryCatch]{
            \Gamma\vdash e : \tau \mid (\emptyset,\{\mathsf{abort}~\ell_C~Q\leadsto C\},Q_e)\\
            \Gamma\vdash h : C\xrightarrow{(\emptyset,\emptyset,Q_h)}\tau \mid (\emptyset,\emptyset,I)
        }{
            \Gamma\vdash\llbracket\mathsf{try}~e~\mathsf{catch}~C\Rightarrow h\rrbracket : \tau \mid (\emptyset,\emptyset,Q_e\sqcup(Q\rhd Q_h))
        }
        \and
        \inferrule*[left=D-Throw]{
            \Gamma\vdash e : C \mid (\emptyset,\emptyset,Q_e)
        }{
            \Gamma\vdash \llbracket\mathsf{throw}_C~e\rrbracket : \tau \mid (\emptyset,\{\mathsf{abort}~\ell_C~Q_e\leadsto C\},\bot)
        }
    \end{mathpar}
    Iterating the construction for \textsc{D-TryCatch} while permitting other aborts in the body effect and still preventing control effects in each handler gives a similar rule for an arbitrary set of exceptions.  Mimicking the exact semantics of Java- or C\# style exceptions with multiple catch blocks per try --- specifically, that a throw within one catch block is not handled by catch blocks for the same try --- requires sum types and re-throwing, which is possible but does not illuminate the details of our continuation effects.

    \subsection{Generalized Iterators}
    \label{sec:simplegen}

    Here we consider a simple encoding of generators in terms of delimited continuations.  Our encoding is similar to Coyle and Crogono's~\cite{Coyle:1991:BAI:122179.122181}, but written independently (we first gave an encoding ourselves, then figuring it was unlikely to be new, located a reference with a similar approach).

    The design of our encoding is as follows: the code that traverses some data structure (or lazily enumerates a sequence) is given as a function taking two arguments: a function to pass a value to a consumer (often a primitive named \lstinline|yield| in many implementations, like C\#), and a function to indicate that iteration is complete (we will call it \lstinline|done|, but it is sometimes given other names such as \lstinline|yield break| in C\#).  Given such code, the function \lstinline|iterate| that we define returns a stateful procedure, each invocation of which returns either the \emph{next} value from the iterator, or a value indicating completion (via an option type).  Only one invocation of \lstinline{iterate} is required; then each time client code is ready for the next value, it invokes the same function returned from the one call to \lstinline|iterate|.

    \begin{figure}
\begin{lstlisting}[language=Scheme,basicstyle=\scriptsize\ttfamily]
(define iterate
  (lambda (f)
    (let* ([tag (new-prompt)]
           [resumption '*]
           [get-next 
              ($\lambda$ () (% (resumption '()) ($\lambda$ (v) v) #:tag tag))]
           [yield ($\lambda$ (val) (call/cc ($\lambda$ (res) 
                                                 (set! resumption res) 
                                                 (abort/cc tag 
                                                    `(Some ,val))) tag))]
           [finish ($\lambda$ () (set! resumption ($\lambda$ () (error))) 
                              (abort/cc tag 'None))]
           )
      (% (begin (call/cc ($\lambda$ (k) (set! resumption k) (abort/cc tag get-next)) tag) 
                (f yield finish) 
                (finish))
         ($\lambda$ (v) v) #:tag tag ))
      ))
\end{lstlisting}
\caption{Racket code (untyped) for a basic generator.}
\label{fig:gen1}
\end{figure}

    Figure \ref{fig:gen1} gives untyped Racket code for a simple generator.
    Compared to our core language, Racket names several primitives slightly differently, moves the tag to the last argument position (it is optional in Racket), and uses a keyword argument \lstinline|#:tag| to specify the prompt tag.
    \lstinline|iterate| takes as an argument a two-parameter function \lstinline|f|.  It allocates a fresh prompt tag \lstinline|tag|, and a placeholder for the resumption continuation --- initially \lstinline|'*| --- which will be used to store the continuation that will produce the remaining items to be generated.  \lstinline|get-next| assmes that placeholder has been initialized: when invoked, it creates a new prompt, and invokes the resumption context.
    \lstinline|yield| captures the enclosing context up to the nearest prompt for \lstinline|tag|, stores it into \lstinline|resumption|, and then aborts with (an option of) the generated value.  The intent is that \lstinline|yield| is invoked inside the prompt created by \lstinline|get-next|, by \lstinline|f|.  The \lstinline[language=Scheme]|abort/cc| throws the value\footnote{The use of back-tick and comma here is how Racket (like Scheme) exposes a shorthand for quasiquotation.} to the handler, which in this case is the identity function, returning the yielded element to site of the call to \lstinline|get-next|.
    \lstinline|finish| marks the iteration as complete.

    The main body after the let-bindings opens a new prompt, whose body is \emph{almost} \lstinline|f| with \lstinline|yield| and \lstinline|finish| provided.  This would permit the code that knows how to generate items --- \lstinline|f| --- to emit items incrementally and indicate its completion.  The actual body is slightly more complex.  The body must make an extra call to \lstinline|finish| after the call to \lstinline|f|, in case \lstinline|f| neglects to call \lstinline|finish| itself. The body must also initialize the resumption context.  It does this by capturing a continuation \lstinline|(begin $\bullet$ (f yield finish) (finish))|, and storing that in the \lstinline|resumption| slot.  The body of that \lstinline|call/cc| then aborts, yielding the function \lstinline|get-next| to the caller of \lstinline|iterate| (by way of applying the identity-function handler, as in our formal semantics).

    Clients can then obtain generators from \lstinline|iterate|, as in the example REPL session in Figure \ref{fig:repl}.

\begin{figure}
\begin{lstlisting}[language=Scheme,basicstyle=\scriptsize\ttfamily]
> (define foo ($\lambda$ (yield done) (yield "a") (yield "b") (done) (yield "never executed")))
> (define next (iterate foo))
> (println (next))
'(Some "a")
> (println (next))
'(Some "b")
> (println (next))
'None
> 
\end{lstlisting}
\caption{Example REPL session using the Racket iterator from Figure \ref{fig:gen1}.}
\label{fig:repl}
\end{figure}
    In general rather than \lstinline|foo|, which is not very interesting by itself, \lstinline|iterate| would be used with routines that yield successive elements of a data structure (list, tree, etc.), or perform come non-trivial computation only on demand (i.e., a form of stream).

    We can express nearly the same code in our core language.  There are only a couple small adjustments to make:
    \begin{itemize}
        \item We must explicitly use sum types for \lstinline|resumption| and the informal option type.
        \item We must introduce a separate prompt and tag to separate the prompt and abort that throws \lstinline|get-next| from the prompt and abort that yields values passed to \lstinline|yield|, since the two would need to have differing return types.
        (Another option would be to use sum types again, but this complicates client code significantly.)
    \end{itemize}
    In addition, our core language cannot give \lstinline|iterate| its own type, but must instead define it as a macro: our core language lacks \lstinline|new-prompt| to declare fresh prompt tags, and even with that, we would require type-level abstraction over tags to give \lstinline|f| a type.  A full language implementation would need to resolve these limitations, but for our current purposes the macro approach is adequate.  

    \begin{figure*}
\begin{lstlisting}[basicstyle=\ttfamily\scriptsize]
$\llbracket\texttt{iterate}~\texttt{init}~\texttt{gen}~\texttt{f}\rrbracket=$
    let resumption /* : ref (Uninit + cont unit (Pp,Cp,E*) (option $\tau$) + Done) */ = ref (inl Uninit) in
    let get-next /* : unit $\xrightarrow{E^*}$ (option $τ$) */ = 
        ($\lambda$ _. (case (! resumption) of
                      ([(inr (inl resume)) (% gen (resume tt) ($\lambda$ (v) v))]
                       [_ None]))
    in
    let yield /* : $\tau$ $\xrightarrow{(\{\mathsf{prophecy}~gen~(P_p,C_p,E^*)\leadsto(\mathsf{option}~\tau)~\mathsf{obs}~(\emptyset,\emptyset,I)\},\{\mathsf{abort}~gen~I\leadsto(\mathsf{option}~\tau)\},\bot)}$ unit */ =
        ($\lambda$ val. (call/cc gen ($\lambda$ res. (resumption := (inr (inl res))) (abort gen (Some val)))))
    in
    let finish /* : unit $\xrightarrow{(\emptyset,\{\mathsf{abort}~gen~I\leadsto(\mathsf{option}~\tau)\},\bot)}$ unit */ =
        ($\lambda$ _. (resumption := (inr (inr Done))) (abort gen None))
    in
      (% init /* :: typeof(get-next) */
         (begin
           (% gen /* :: option τ */
              (begin (call/cc gen ($\lambda$ k. (resumption := (inr (inl k))) (abort init get-next))) (f yield finish) (finish))
              ($\lambda$ v. v))
           get-next)
         ($\lambda$ v. v))
\end{lstlisting}
        \caption{A typed generator, parameterized (here implicitly) by two tags, \lstinline|init| and \lstinline|gen|.}
        \label{fig:gen2}
    \end{figure*}

    Figure \ref{fig:gen2} gives a version in our core language, assuming an instantiation with mutable references (with the identity effect for all uses) and a sum type, with type annotations. It makes the distinctions mentioned above, aborting fron inside the inner initial prompt (the \lstinline|gen| prompt) to the outer initial prompt (the \lstinline|init| prompt) to separate the result types, but still return \lstinline|get-next| to \lstinline|iterate|'s client only after the the resumption is initialized.  As in Figure \ref{fig:gen1}, the initially-captured continuation used for the first call to \lstinline|get-next| is still \lstinline|(begin $\bullet$ (f yield finish) (finish))|.

    Before we discuss typing uses of \lstinline|iterate|, let us consider what a desirable typing would entail.  First, note that the result of a ``call'' to \lstinline|iterate| is a closure (specifically \lstinline|get-next|), and assuming the underlying effect quantale ignores the reference manipulation during initialization, the immediate effect of evaluating a use of \lstinline|iterate| should be $(\emptyset,\emptyset,I)$.

    The assumed argument types for $f$ reflect the declared types for \lstinline|yield| and \lstinline|finish|.
    The main point to justify above is the latent effect assumed for $f$.
    Consider $E$ to be an upper bound on the underlying effect of $f$'s body between two successive calls to \lstinline|yield|.\footnote{Or between a \lstinline|yield| and a final call to \lstinline|finish|, or between \lstinline|finish| and any ``regular'' return by $f$.}

    \begin{figure*}
    \begin{mathpar}
    \mathsf{GenProphs}(P,\ell,E)=\begin{array}{l}
        \forall \ell',P',C',Q',P'',C'',Q'',\tau\ldotp
         \mathsf{prophecy}~\ell'~(P',C',Q')\leadsto\tau~\mathsf{obs}~(P'',C'',Q'')\in P \Rightarrow\\
        \qquad \ell'=\ell \land\tau=\mathsf{bool}\land \mathsf{GenProphs}(P',\ell,E) \land 
        C' \sqsubseteq \{\mathsf{abort}~\ell~(E^*)\leadsto(\mathsf{option}~\tau)\}\land
        Q'\sqsubseteq E^*\\
        \qquad \land\punblock{P''}{\ell}\sqsubseteq\punblock{P'}{\ell}
            \land C''\sqsubseteq C'
            \land Q''\sqsubseteq Q'
    \end{array}
    \and
        \inferrule[D-Iterate]{
            \Gamma\vdash f :
                \left(\begin{array}{l}
                     (\tau\xrightarrow{(\{\mathsf{prophecy}~gen~(P_p,C_p,E^*)\leadsto(\mathsf{option}~\tau)~\mathsf{obs}~(\emptyset,\emptyset,I)\},\{\mathsf{abort}~gen~I\leadsto(\mathsf{option}~\tau)\},\bot)}\mathsf{unit})\arcr
                    \qquad\xrightarrow{I}
        (\mathsf{unit}\xrightarrow{(\emptyset,\{\mathsf{abort}~gen~I\leadsto(\mathsf{option}~\tau)\},\bot)}\mathsf{unit})\arcr
        \qquad\qquad\xrightarrow{(P,C,E^*)}\mathsf{unit}
                \end{array}\right)
        \\
        C \sqsubseteq \{ \mathsf{abort}~gen~(E^*)\leadsto(\mathsf{option}~\tau) \}
        \\
        \mathsf{GenProphs}(P,\mathit{gen},E)
        }{
            \Gamma\vdash \llbracket\mathsf{iterate}~init~gen~f\rrbracket : \mathsf{unit}\xrightarrow{(\emptyset,\emptyset,E^*)}\mathsf{option}~\tau \mid (\emptyset,\emptyset,I)
        }
    \end{mathpar}
    \caption{Derived rule for generators.}
    \label{fig:iterate}
    \end{figure*}

    Following our approach above, we can give the derived rule in Figure \ref{fig:iterate} assuming all prophecies and control effects are related to the generator tag $\mathit{gen}$.
    Key to this derived rule above is the fact that \lstinline|f|'s body is constrained to have prophecies and control effects related only to $gen$.  Because all invocations of \lstinline|f| or continuations containing parts of \lstinline|f|'s body occur under prompts for $gen$, all of those effects are resolved inside calls to \lstinline|get-next|, leaving only the underlying effect. (This rule does assume no other effects --- such as aborts from exceptions --- escape the body of the generator.)
    \textsc{D-Iterate} also permits the prophecy ``emitted'' by the continuation capture for \lstinline|yield| to vary between uses of \textsc{D-Iterate}: $P_p$ and $C_p$ are the prophecies and control effects that the uses of \lstinline|yield| for that particular iterator construction should predict.  (The \textsf{GenProphs} antecedent in \textsc{D-Iterate} ensures the prophecy is validated.)  Intuitively, $P_p$ and $C_p$ should be chosen such that reusing them for the prophecy of every continuation captured by \lstinline|yield| predicts the effect of the code up to the next \lstinline|yield| or the end of the function --- including the effects of the next use of \lstinline|yield|, which itself must predict the effect of the code up to the next \lstinline|yield|\ldots and so on.

    With the derived rule in hand, we would hope that it is precise enought to validate the intuitive effect for common constructions.

    \begin{example}[Iterating a Pure Function]
        Consider the example of a generator that always immediately yields the boolean \lstinline|true|.  Under the assumptions made earlier (that reference mutations are ignored by the effect system), we may derive:
        {\scriptsize \[ \Gamma\vdash \llbracket\mathsf{iterate}~init~gen~(\lambda y,f\ldotp \llbracket\mathsf{loop}~(y~\mathsf{true})\rrbracket)\rrbracket : \mathsf{unit}\xrightarrow{(\emptyset,\emptyset,I)}\mathsf{option}~\mathsf{bool} \mid (\emptyset,\emptyset,I) \]}
        This follows from the rule above because the loop body $(y~\mathsf{true})$ has effect $\chi$ as defined in Figure \ref{fig:puregen_eff}, which is essentially the assumed latent effect for the \lstinline|yield| argument, assuming the presence of (equi-)recursively-defined prophecies.
        By the derived rule from Section \ref{sec:infcontrol}, the overall loop then has effect $\chi^*$, which is then the latent effect of the function.
        The conditions on $C$ and $P$ in \textsc{D-Iterate} are then clearly satisfied with $E=I$: the infinite union in the iteration above only creates prophecies satisfying \textsf{GenProphs}, in particular because all finite iterations of the effect produce prophecies less than the recursive prophecy (after the unblocking).

        To provide a bit more intuition for this, note that the prophecy set component $P_{\chi}$ of $\chi$ is itself less than the recursive prophecy $P_{fix}$ predicted by $P_{\chi}$: $P_\chi\sqsubseteq P_{fix}$.  This is enough to show that the prophecy component of $\chi^2$ (as shown in Figure \ref{fig:puregen_eff}) remains valid (in the sense of \textsf{validEffects}' prophecy validation).  This extends to any of the finite iterations introduced by the iteration operator of Section \ref{sec:contiter}.
        Because $P_{fix}$'s recursively-defined observations are again $P_{fix}$, then sequencing any finite iteration of $\chi$ with itself yields a finite approximation of $P_{fix}$.  The approximation is $\sqsubseteq P_{fix}$ in every case because at the point the approximation drops off with the ``base case'' observation $(\emptyset,\emptyset,I)$, this is less than the observation in $P_{fix}$ because the prophecy and control components are merely empty sets, and $I\sqsubseteq I$.
        
        So the overall latent effect of the iterator produced by \textsf{iterate} is $(\emptyset,\emptyset,I)$.

    \end{example}

    \begin{figure*}\scriptsize
    \[\textrm{let}~P_{fix}\equiv\{\mu{}P\ldotp\mathsf{prophecy}~gen~(\{P\},\{\mathsf{abort}~gen~I\leadsto\mathsf{bool}\},I)~\mathsf{obs}~(\{P\},\{\mathsf{abort}~gen~I\leadsto\mathsf{bool}\},\bot)\leadsto\mathsf{unit}\} \]
    \[ P_{fix} \equiv \{\mathsf{prophecy}~gen~(P_{fix},\{\mathsf{abort}~gen~I\leadsto\mathsf{bool}\},I)~\mathsf{obs}~(P_{fix},\{\mathsf{abort}~gen~I\leadsto\mathsf{bool}\},\bot)\leadsto\mathsf{unit}\}\]
    \[ \chi\equiv(\{\mathsf{prophecy}~gen~\textrm{\textcolor{blue}{$(P_{fix},\{\mathsf{abort}~gen~I\leadsto\mathsf{bool}\},I)$}}~\mathsf{obs}~(\emptyset,\emptyset,I)\}, \{\mathsf{abort}~gen~(\mathsf{option}~I\leadsto\mathsf{bool})\}, \bot) \]
    \[ 
        \begin{array}{lcl}
        P_{\chi^2}&\equiv& \{\mathsf{prophecy}~gen~\textrm{\textcolor{blue}{$(P_{fix},\{\mathsf{abort}~gen~I\leadsto\mathsf{bool}\},I)$}}~\mathsf{obs}~(\emptyset,\emptyset,I)\}\blacktriangleright\chi\\
        &\equiv& \{\mathsf{prophecy}~gen~\textrm{\textcolor{blue}{$(P_{fix},\{\mathsf{abort}~gen~I\leadsto\mathsf{bool}\},I)$}}~\mathsf{obs}~(P_\chi,C_\chi,Q_\chi)\}\\
        \end{array}
    \]
    \caption{The effect of the infinite loop body that always immediately yields a value.}
    \label{fig:puregen_eff}
    \end{figure*}

    This example assumes recursive prophecies, not present in our initial presentation, and which up to now we have not required.  While we can construct the semantics of any finite prefix of an execution by taking the union over all finite iterations per Section \ref{sec:contiter}, that construction ignores whether or not the observations in the resulting prophecies could actually be consistent with the predictions.  In order to write an actual prediction that over-approximates these observations, we require recursive prophecies.  

    \begin{example}[Iterating an Impure Looping Generator]
        One style of use for generators is to implement on-demand (pull) streams.  In a concurrent setting this may be implemented in a way where on each request the generator takes a lock in order to find the next element to produce.  To avoid unnecessary serialization, the lock must be released before yielding a new element, then reacquired.  There are two ways to implement this.  The first approach acquires, searches, and releases the lock on each call to \lstinline|get-next|.  In this case the body effect is pure as above.  An alternative is to assume the caller holds the lock initially, and the iterator should release and reaquire the lock. 

        Consider iterating the following generator function, which presumes the existence of some auxilliary state in the environment:\footnote{Technically this code should contain an inner loop because condition variable implementations permit spurious wake-ups, for a variety of reasons, including that even if the intended condition was true when the sleeping thread was signalled, it may have been falsified again before the thread has a chance to execute again.  This is a simplified example to focus on the effects.}
\begin{lstlisting}[language=Scheme]
($\lambda$ (yield finish)
    (begin (while (not-done?) 
             (begin (cond_wait l) 
                      (yield (first-elem))))
             (finish '())))
\end{lstlisting}
        The underlying effect of the body would be $(\{l\},\{l\})^*=(\{l\},\{l\})$ in the underlying locking effect quantale (it begins assuming $l$ is held, and finishes assuming $l$ is held).  This effect is observed in the execution prefix preceding both the call to \lstinline|finish| and prior to ``each'' call to \lstinline|yield|.
        The loop body has nearly the same effect as the body in the previous example, but with some uses of $I$ replaced by $(\{l\},\{l\})$, per Figure \ref{fig:condwait_eff}.
        The same argument for the prophecies always remaining valid extends to this case, taking advantage of the fact that $(\{l\},\{l\})^*=(\{l\},\{l\})$.
        This makes the latent effect of the generator function for the above also $(\{l\},\{l\})$.
    \end{example}
    \begin{figure*}
        {\scriptsize
    \[ \chi\equiv(\{\mathsf{prophecy}~gen~\textrm{\textcolor{blue}{$(P_{fix},\{\mathsf{abort}~gen~I\leadsto\mathsf{bool}\},(\{l\},\{l\}))$}}~\mathsf{obs}~(\emptyset,\emptyset,I)\}, \{\mathsf{abort}~gen~(\mathsf{option}~I\leadsto\mathsf{bool})\}, (\{l\},\{l\})) \]
        }
        \vspace{-2em}
        \caption{Loop body effect of the loop that waits and synchronizes for a value.}
        \label{fig:condwait_eff}
    \end{figure*}

    Because \textsc{D-Iterate} is a derived rule in a sound type-and-effect system, we know it is sound.  
    We could also go beyond the rule \textsc{D-Iterate} above, which assumes the only control effects and prophecies escaping the body of \lstinline|f| are related to the generator infrastructure.  This is naturally not always the case --- in a language like C\#, the body of a generator can throw exceptions. We do not explore the details of deriving rules for such combinations here, but the same methodology employed thus far still applies to that case.

    \fi

\section{Related Work}
\label{sec:relwork}
\ifTR
 Here we recall other related work not covered in Section \ref{sec:bg}.
 \fi

Recent years have seen great progress on semantic models for sequential effect systems~\cite{tate13,katsumata14,mycroft16}, centering on what are now known as \emph{graded monads}: monads indexed by some kind of monoid (to model sequential composition), commonly a partially-ordered monoid following Katsumata~\cite{katsumata14}.  Gordon~\cite{ecoop17} focused on capturing common structure for prior concrete effect systems, leading to the first abstract characterization of sequential effect systems with singleton effects, effect polymorphism, and iteration of sequential effects.  
\looseness=-1

To the best of our knowledge we are the first to use the term ``accumulator'' as we do to identify this as a reusable technique.  However accumulators have appeared before.  Koskinen and Terauchi's effect system~\cite{Koskinen14LTR} uses left-accumulators for safety and liveness properties (requiring an oracle for liveness).  Effects in their system are a pair of sets, one a set of finite traces (for terminating executions) and the other a set of infinite traces (for non-terminating executions).  The infinite traces left-accumulate: code that comes after a non-terminating expression in program-order never runs.  On the other hand, \emph{finite} executions from code \emph{before} an infinite execution extend the prefix of the infinite executions.
Earlier, Neamtiu et al.~\cite{neamtiu2008contextual} defined \emph{contextual} effects to track what (otherwise order-unaware) effects occurred before or after an expression,  to ensure key correctness properties for code using dynamic software updates.
\looseness=-1

Effect systems treating continuations are nearly as old as effect systems themselves~\cite{Jouvelot1989continuations}.
To the best of our knowledge, we are the first to integrate \emph{sequential} effects with exceptions, generators, or general delimited continuations --- or any control flow construct beyond while loops, including any form of continuation, tagged or otherwise.
As mentioned in Section \ref{sec:bg}, the original motivation for tags was to prevent encodings of separate control operators from interfering with each other~\cite{Sitaram1993}, which is critical for our goals, strictly more expressive than untagged continuations, and motivates important elements of the theory (blocking).
The only other work we know of focusing on effects with tagged delimited continuations is Pretnar and Bauer's variant~\cite{pretnar2014effect} of algebraic effects and handlers~\cite{pretnar2014effect} where operations may be handled by outer \textsf{handle} constructs (not just the closest construct as in other algebraic effects work). Their commutative effects ensures all algebraic operations are handled by some enclosing handler.
\looseness=-1

Tov and Pucella~\cite{tov2011} examined the interaction of \emph{untagged} delimited continuations with substructural types (a coeffect~\cite{Petricek2014CCC}).
Delbianco and Nanevski adapted Hoare Type Theory for untagged \emph{algebraic} continuations~\cite{Delbianco:2013:HRC:2500365.2500593},
which include prompts and handlers, but place handlers at the site of an abort rather than at the prompt in order to satisfy some useful computational equalities (see below).
As a consequence, encoding non-trivial control flow constructs in their system becomes significantly more complex; for example, simulating the standard semantics of throwing exceptions to the nearest enclosing catch block for the exception type would require catching, dispatching, and re-throwing at every prompt.  
This and lack of tagging would make compositional study of multiple control flow constructs / control operators difficult, and as our work shows the treatment of multiple tags is not a trivial extension of untagged semantics.
Atkey~\cite{atkey2009parameterised} considered denotational semantics for (untagged) \emph{composable} continuations in his parameterized monad framework for (denotational) sequential effect systems, essentially giving a denotation of a type-and-effect system for answer type modification~\cite{asai2007polymorphic,danvy1989functional} --- a kind of sequential effect which can be used to allow continuations to temporarily change the result type of a continuation, as long as it is known (via the effects) that it will be changed back. Thus Atkey considered answer type modification effects as an instance of a sequential effect specific to using control operators, rather than having an application-domain-focused effect (like exceptions or traces) work with continuations or giving an account of general sequential effects for control operators.  
Readers familiar with answer type modification may wonder about directly supporting it in our generic core language. We have not yet considered this deeply, but note that (i) directly ascribing answer type modifications to control operations would require assigning \emph{specific} answer type modification effects to the control operations, not just effects derived from primitives, but (ii) Kobori et al.~\cite{KoboriKK16} showed that tagged shift/reset can express answer type modification in a type system that does not track answer types explicitly, so such an extension may provide no additional power (treating convenience as another matter).
\looseness=-1

Algebraic effects with handlers~\cite{plotkin2009handlers} are a means to describe the \emph{semantics} of effects in terms of a set of operations (the effectful operations) along with handlers that interpret those operations as actions on some resource.  The combination yields an algebra characterizing equality of different effectful program expressions, hence the term ``algebraic''.
Languages with algebraic effects include an effect system to reason about which effects a computation uses, to ensure they are handled.  Some implementations even use Lindley and Cheney's effect adaptation~\cite{lindley2012row} of row polymorphism~\cite{wand1989type} to support effect inference~\cite{koka}.
Handlers for algebraic effects receive both the action to interpret and the continuation of the program following the effectful action.  Thus they can implement many control operators, including generators and cooperative multithreading~\cite{Leijen2017async}, as with the delimited continuations we study.
In an untyped setting without tagging, algebraic handlers can simulate (via macro translation) \lstinline|shift0|/\lstinline|reset0|~\cite{Forster:2017:EPU:3136534.3110257}, which can simulate prompts and handlers~\cite{shan2007static} (with correct space complexity, not only extensionally-correct behavior); with those limitations, handlers are as powerful as the constructs we study. For the common commutative effect system for handlers that ensures all operations are handled, Forster et al.~\cite{Forster:2017:EPU:3136534.3110257} prove that the translation from handlers to prompts (\lstinline|shift0|) is not type-and-effect preserving, and conjecture the reverse translation also fails to preserve types.  They conjectured that adding polymorphism to each system would enable a type-and-effect preserving translation (again, without tagging, for a commutative effect system), which was recently confirmed by Pir\'{o}g et al.~\cite{pirog2019typed} for a class of commutative effect systems.

The effect systems considered for algebraic effects thus far have only limited support for reasoning about sequential effects.
The types given for individual algebraic effects do support reasoning about the existence of a certain type of resource before and after the computation~\cite{Brady2013,bauer2013effect}.  However, the way this is done corresponds to a parameterized monad~\cite{atkey2009parameterised}, which Tate showed~\cite{tate13} crisply do not include all meaningful sequential effect systems. His examples that cannot be modeled as parameterized monads include examples that \emph{can} be modeled as effect quantales~\cite{quantalesjournal}, such as the non-reentrant locking effect system Tate uses to motivate aspects of his approach.

General considerations of sequential effect systems have not yet been explored for algebraic effects.  When it is considered, it seems likely ideas from our development (particularly prophecies) will be useful.  For example, Dolan et al.~\cite{dolan2017ocamleff} offer two reasons for dynamically enforcing linearity of continuations in their handlers: performance, but also avoiding the sorts of errors prevented by sequential effect systems, such as closing a file twice by reusing a continuation. 
\looseness=-1

It also seems plausible that our approach could be adapted to algebraic effects and handlers. With an effectively-tagged version of handlers~\cite{pretnar2014effect}, a similar macro-expression of control flow constructs and control operators is likely feasible, in particular adapting our notion of prophecy and observation to handlers: in this case, the continuations themselves are seen by handlers that are direct subexpressions of the handling construct itself, so prophecies might observe ``outside-in'' rather than our system's ``inside-out'' accumulation.

The approach we take to deriving type rules for control flow constructs and control operators is reminiscent of work done in parallel with ours by Pombrio and Krishnamurthi~\cite{Pombrio:2018:ITR:3192366.3192398}.  They address the problem of producing useful type rules when a language semantics and type rules are defined directly for a simpler core language, and a full source language is defined using syntactic sugar (i.e., macros) that expand into core language expressions with the intended semantics, such as the approach taken by $\lambda_{\textit{JS}}$~\cite{Guha10}.  There the issue is that type errors given in terms of the elaborated core terms are difficult to understand for developers writing in the unelaborated source language.  Pombrio and Krishnamurthi offer an approach to automatically lift core language type rules through the desugaring process to the source language, providing sensible source-level type errors.
Their work focuses on type systems without effects, but including such notions as subtyping and existential types.  They do not consider control operators (delimited continuations) or effects (neither commutative nor sequential).
Extending their approach to support the language features and types (effects) we consider would make our approach more useful to effect system designers, though this is non-trivial due to the many ways to combine sequential effects.
\looseness=-1

\section{Conclusions} 
We have given the first general approach to integrating arbitrary sequential effect systems with tagged delimited control operators, which allows lifting existing sequential effect systems without knowledge of control operators to automatically support tagged delimited control.  We have used this characterization to derive sequential effect system rules for standard control flow structures macro-expressed via continuations, including deriving known forms (loops) and giving the first characterization of exceptions and generators in sequential effect systems.
\looseness=-1

\ifTR
\subsection*{Acknowledgements}
We would like to thank the anonymous referees for LICS 2018, OOPSLA 2018, LICS 2019, OOPSLA 2019, POPL 2020, and ECOOP 2020, whose detailed and constructive feedback lead to significant improvements.
\fi
    
    \bibliographystyle{plainurl}
    \bibliography{effects,csg}

\begin{thebibliography}{10}

\bibitem{Abadi2006}
Martin Abadi, Cormac Flanagan, and Stephen~N. Freund.
\newblock Types for safe locking: Static race detection for java.
\newblock {\em ACM Trans. Program. Lang. Syst.}, 28(2), 2006.

\bibitem{abadi1988existence}
Mart{\'\i}n Abadi and Leslie Lamport.
\newblock The existence of refinement mappings.
\newblock In {\em LICS}, 1988.

\bibitem{abadi1991existence}
Mart{\'\i}n Abadi and Leslie Lamport.
\newblock The existence of refinement mappings.
\newblock {\em Theoretical Computer Science}, 82(2):253--284, 1991.

\bibitem{amtoft1999}
Torben Amtoft, Flemming Nielson, and Hanne~Riis Nielson.
\newblock {\em {Type and Effect Systems: Behaviours for Concurrency}}.
\newblock Imperial College Press, London, UK, 1999.

\bibitem{asai2007polymorphic}
Kenichi Asai and Yukiyoshi Kameyama.
\newblock Polymorphic delimited continuations.
\newblock In {\em APLAS}, pages 239--254, 2007.

\bibitem{atkey2009parameterised}
Robert Atkey.
\newblock {Parameterised Notions of Computation}.
\newblock {\em {Journal of Functional Programming}}, 19:335--376, July 2009.

\bibitem{bauer2013effect}
Andrej Bauer and Matija Pretnar.
\newblock An effect system for algebraic effects and handlers.
\newblock In {\em International Conference on Algebra and Coalgebra in Computer
  Science}, pages 1--16. Springer, 2013.

\bibitem{birkhoff}
Garrett Birkhoff.
\newblock {\em Lattice theory}, volume~25 of {\em Colloquium Publications}.
\newblock American Mathematical Soc., 1940.
\newblock Third edition, eighth printing with corrections, 1995.

\bibitem{blyth2006lattices}
Thomas~Scott Blyth.
\newblock {\em Lattices and ordered algebraic structures}.
\newblock Springer Science \& Business Media, 2006.

\bibitem{boyapati02}
Chandrasekhar Boyapati, Robert Lee, and Martin Rinard.
\newblock {Ownership Types for Safe Programming: Preventing Data Races and
  Deadlocks}.
\newblock In {\em {OOPSLA}}, 2002.

\bibitem{boyapati01}
Chandrasekhar Boyapati and Martin Rinard.
\newblock {A Parameterized Type System for Race-Free Java Programs}.
\newblock In {\em {OOPSLA}}, 2001.

\bibitem{Brady2013}
Edwin Brady.
\newblock Programming and reasoning with algebraic effects and dependent types.
\newblock In {\em Proceedings of the 18th ACM SIGPLAN International Conference
  on Functional Programming}, ICFP '13, pages 133--144. ACM, 2013.
\newblock \href {http://dx.doi.org/10.1145/2500365.2500581}
  {\path{doi:10.1145/2500365.2500581}}.

\bibitem{clements2001modeling}
John Clements, Matthew Flatt, and Matthias Felleisen.
\newblock Modeling an algebraic stepper.
\newblock In {\em European symposium on programming}, pages 320--334. Springer,
  2001.

\bibitem{Coyle:1991:BAI:122179.122181}
Christopher Coyle and Peter Crogono.
\newblock Building abstract iterators using continuations.
\newblock {\em SIGPLAN Not.}, 26(2):17--24, January 1991.
\newblock URL: \url{http://doi.acm.org/10.1145/122179.122181}, \href
  {http://dx.doi.org/10.1145/122179.122181} {\path{doi:10.1145/122179.122181}}.

\bibitem{crary1999typed}
Karl Crary, David Walker, and Greg Morrisett.
\newblock Typed memory management in a calculus of capabilities.
\newblock In {\em Proceedings of the 26th ACM SIGPLAN-SIGACT symposium on
  Principles of programming languages}, pages 262--275. ACM, 1999.
\newblock \href {http://dx.doi.org/10.1145/292540.292564}
  {\path{doi:10.1145/292540.292564}}.

\bibitem{danvy2006analytical}
Olivier Danvy.
\newblock An analytical approach to program as data objects, 2006.
\newblock DSc thesis, Department of Computer Science, Aarhus University.

\bibitem{danvy1989functional}
Olivier Danvy and Andrzej Filinski.
\newblock A functional abstraction of typed contexts.
\newblock Technical report, DIKU --- Computer Science Department, University of
  Copenhagen, July 1989.

\bibitem{Delbianco:2013:HRC:2500365.2500593}
Germ\'{a}n~Andr{\'e}s Delbianco and Aleksandar Nanevski.
\newblock Hoare-style reasoning with (algebraic) continuations.
\newblock In {\em ICFP}, 2013.

\bibitem{dolan2017ocamleff}
Stephen Dolan, Spiros Eliopoulos, Daniel Hillerstr{\"o}m, Anil Madhavapeddy,
  K.~C. Sivaramakrishnan, and Leo White.
\newblock Concurrent system programming with effect handlers.
\newblock In Meng Wang and Scott Owens, editors, {\em Trends in Functional
  Programming}, pages 98--117, Cham, 2018. Springer International Publishing.

\bibitem{Felleisen88}
Matthias Felleisen.
\newblock The theory and practice of first-class prompts.
\newblock In {\em Conference Record of the Fifteenth Annual {ACM} Symposium on
  Principles of Programming Languages, San Diego, California, USA, January
  10-13, 1988}, pages 180--190, 1988.
\newblock URL: \url{http://doi.acm.org/10.1145/73560.73576}, \href
  {http://dx.doi.org/10.1145/73560.73576} {\path{doi:10.1145/73560.73576}}.

\bibitem{felleisen1991expressive}
Matthias Felleisen.
\newblock On the expressive power of programming languages.
\newblock {\em Science of computer programming}, 17(1-3):35--75, 1991.

\bibitem{felleisen2009semantics}
Matthias Felleisen, Robert~Bruce Findler, and Matthew Flatt.
\newblock {\em Semantics engineering with PLT Redex}.
\newblock Mit Press, 2009.

\bibitem{FelleisenF87}
Matthias Felleisen and Daniel~P. Friedman.
\newblock A reduction semantics for imperative higher-order languages.
\newblock In {\em PARLE, Parallel Architectures and Languages Europe, Volume
  {II:} Parallel Languages, Eindhoven, The Netherlands, June 15-19, 1987,
  Proceedings}, pages 206--223, 1987.
\newblock URL: \url{https://doi.org/10.1007/3-540-17945-3\_12}, \href
  {http://dx.doi.org/10.1007/3-540-17945-3\_12}
  {\path{doi:10.1007/3-540-17945-3\_12}}.

\bibitem{objtyrace99}
Cormac Flanagan and Mart\'{\i}n Abadi.
\newblock {Object Types against Races}.
\newblock In {\em {CONCUR}}, 1999.

\bibitem{safelocking99}
Cormac Flanagan and Mart\'{\i}n Abadi.
\newblock {Types for Safe Locking}.
\newblock In {\em {ESOP}}, 1999.

\bibitem{flanagan2003atomicity}
Cormac Flanagan and Shaz Qadeer.
\newblock A type and effect system for atomicity.
\newblock In {\em {PLDI}}, 2003.

\bibitem{flanagan2003tldi}
Cormac Flanagan and Shaz Qadeer.
\newblock Types for atomicity.
\newblock In {\em {TLDI}}, 2003.

\bibitem{Flatt2007}
Matthew Flatt, Gang Yu, Robert~Bruce Findler, and Matthias Felleisen.
\newblock Adding delimited and composable control to a production programming
  environment.
\newblock In {\em {ICFP}}, 2007.

\bibitem{Forster:2017:EPU:3136534.3110257}
Yannick Forster, Ohad Kammar, Sam Lindley, and Matija Pretnar.
\newblock On the expressive power of user-defined effects: Effect handlers,
  monadic reflection, delimited control.
\newblock {\em Proc. ACM Program. Lang.}, 1(ICFP):13:1--13:29, August 2017.
\newblock URL: \url{http://doi.acm.org/10.1145/3110257}, \href
  {http://dx.doi.org/10.1145/3110257} {\path{doi:10.1145/3110257}}.

\bibitem{fuchs2011partially}
Laszlo Fuchs.
\newblock {\em Partially ordered algebraic systems}, volume~28 of {\em
  {International Series of Monographs on Pure and Applied Mathematics}}.
\newblock Dover Publications, 2011.
\newblock Reprint of 1963 Pergamon Press version.

\bibitem{ecoop17}
Colin~S. Gordon.
\newblock {A Generic Approach to Flow-Sensitive Polymorphic Effects}.
\newblock In {\em {ECOOP}}, 2017.

\bibitem{quantalesjournal}
Colin~S. Gordon.
\newblock {Polymorphic Iterable Sequential Effect Systems}.
\newblock Technical Report arXiv cs.PL/cs.LO 1808.02010, {Computing Research
  Repository (Corr)}, August 2018.
\newblock In Submission..
\newblock URL: \url{https://arxiv.org/abs/1808.02010}, \href
  {http://arxiv.org/abs/1808.02010} {\path{arXiv:1808.02010}}.

\bibitem{controltr}
Colin~S. Gordon.
\newblock {Sequential Effect Systems with Control Operators}.
\newblock Technical Report arXiv cs.PL 1811.12285, {Computing Research
  Repository (CoRR)}, December 2018.
\newblock URL: \url{https://arxiv.org/abs/1811.12285}, \href
  {http://arxiv.org/abs/1811.12285} {\path{arXiv:1811.12285}}.

\bibitem{ecoop13}
Colin~S. Gordon, Werner Dietl, Michael~D. Ernst, and Dan Grossman.
\newblock {JavaUI: Effects for Controlling UI Object Access}.
\newblock In {\em {ECOOP}}, 2013.

\bibitem{tldi12}
Colin~S. Gordon, Michael~D. Ernst, and Dan Grossman.
\newblock {Static Lock Capabilities for Deadlock Freedom}.
\newblock In {\em {TLDI}}, 2012.

\bibitem{gosling2014java}
James Gosling, Bill Joy, Guy~L Steele, Gilad Bracha, and Alex Buckley.
\newblock {\em {The Java Language Specification: Java SE 8 Edition}}.
\newblock Pearson Education, 2014.

\bibitem{projectloom}
OpenJDK~HotSpot Group.
\newblock {OpenJDK Project Loom: Fibers and Continuations}, 2019.
\newblock URL: \url{https://wiki.openjdk.java.net/display/loom/Main}.

\bibitem{Guha10}
Arjun Guha, Claudiu Saftoiu, and Shriram Krishnamurthi.
\newblock {The Essence of {JavaScript}}.
\newblock In {\em ECOOP}, 2010.

\bibitem{Haynes:1987:ECP:29873.30392}
Christopher~T. Haynes and Daniel~P. Friedman.
\newblock Embedding continuations in procedural objects.
\newblock {\em ACM Trans. Program. Lang. Syst.}, 9(4):582--598, October 1987.
\newblock URL: \url{http://doi.acm.org/10.1145/29873.30392}, \href
  {http://dx.doi.org/10.1145/29873.30392} {\path{doi:10.1145/29873.30392}}.

\bibitem{DBLP:conf/lfp/HaynesFW84}
Christopher~T. Haynes, Daniel~P. Friedman, and Mitchell Wand.
\newblock Continuations and coroutines.
\newblock In {\em {LISP} and Functional Programming}, pages 293--298, 1984.

\bibitem{DBLP:journals/cl/HaynesFW86}
Christopher~T. Haynes, Daniel~P. Friedman, and Mitchell Wand.
\newblock Obtaining coroutines with continuations.
\newblock {\em Comput. Lang.}, 11(3/4):143--153, 1986.
\newblock URL: \url{https://doi.org/10.1016/0096-0551(86)90007-X}, \href
  {http://dx.doi.org/10.1016/0096-0551(86)90007-X}
  {\path{doi:10.1016/0096-0551(86)90007-X}}.

\bibitem{Jouvelot1989continuations}
P.~Jouvelot and D.~K. Gifford.
\newblock Reasoning about continuations with control effects.
\newblock In {\em {PLDI}}, 1989.

\bibitem{katsumata14}
Shin-ya Katsumata.
\newblock Parametric effect monads and semantics of effect systems.
\newblock In {\em {POPL}}, 2014.

\bibitem{klein2012redex}
Casey Klein, John Clements, Christos Dimoulas, Carl Eastlund, Matthias
  Felleisen, Matthew Flatt, Jay~A. McCarthy, Jon Rafkind, Sam Tobin-Hochstadt,
  and Robert~Bruce Findler.
\newblock Run your research: On the effectiveness of lightweight mechanization.
\newblock In {\em Proceedings of the 39th Annual ACM SIGPLAN-SIGACT Symposium
  on Principles of Programming Languages}, POPL '12, pages 285--296, New York,
  NY, USA, 2012. ACM.
\newblock URL: \url{http://doi.acm.org/10.1145/2103656.2103691}, \href
  {http://dx.doi.org/10.1145/2103656.2103691}
  {\path{doi:10.1145/2103656.2103691}}.

\bibitem{KoboriKK16}
Ikuo Kobori, Yukiyoshi Kameyama, and Oleg Kiselyov.
\newblock Answer-type modification without tears: Prompt-passing style
  translation for typed delimited-control operators.
\newblock In Olivier Danvy and Ugo de'Liguoro, editors, {\em Proceedings of the
  Workshop on Continuations, WoC 2016, London, UK, April 12th 2015}, volume 212
  of {\em {EPTCS}}, pages 36--52, 2015.
\newblock URL: \url{https://doi.org/10.4204/EPTCS.212.3}, \href
  {http://dx.doi.org/10.4204/EPTCS.212.3} {\path{doi:10.4204/EPTCS.212.3}}.

\bibitem{Koskinen14LTR}
Eric Koskinen and Tachio Terauchi.
\newblock Local temporal reasoning.
\newblock In {\em {CSL/LICS}}, 2014.

\bibitem{KrishnamurthiHMGPF07}
Shriram Krishnamurthi, Peter~Walton Hopkins, Jay~A. McCarthy, Paul~T. Graunke,
  Greg Pettyjohn, and Matthias Felleisen.
\newblock Implementation and use of the {PLT} scheme web server.
\newblock {\em Higher-Order and Symbolic Computation}, 20(4):431--460, 2007.
\newblock URL: \url{https://doi.org/10.1007/s10990-007-9008-y}, \href
  {http://dx.doi.org/10.1007/s10990-007-9008-y}
  {\path{doi:10.1007/s10990-007-9008-y}}.

\bibitem{koka}
Daan Leijen.
\newblock {Koka: Programming with Row Polymorphic Effect Types}.
\newblock In {\em Proceedings 5th Workshop on Mathematically Structured
  Functional Programming (MSFP)}, 2014.

\bibitem{Leijen2017async}
Daan Leijen.
\newblock Structured asynchrony with algebraic effects.
\newblock In {\em Proceedings of the 2nd ACM SIGPLAN International Workshop on
  Type-Driven Development}, TyDe 2017, pages 16--29, New York, NY, USA, 2017.
  ACM.
\newblock URL: \url{http://doi.acm.org/10.1145/3122975.3122977}, \href
  {http://dx.doi.org/10.1145/3122975.3122977}
  {\path{doi:10.1145/3122975.3122977}}.

\bibitem{lindley2012row}
Sam Lindley and James Cheney.
\newblock Row-based effect types for database integration.
\newblock In {\em Proceedings of the 8th ACM SIGPLAN workshop on Types in
  language design and implementation}, pages 91--102. ACM, 2012.

\bibitem{lucassen88}
J.~M. Lucassen and D.~K. Gifford.
\newblock {Polymorphic Effect Systems}.
\newblock In {\em {POPL}}, 1988.

\bibitem{marino09}
Daniel Marino and Todd Millstein.
\newblock {A Generic Type-and-Effect System}.
\newblock In {\em {TLDI}}, 2009.
\newblock \href {http://dx.doi.org/10.1145/1481861.1481868}
  {\path{doi:10.1145/1481861.1481868}}.

\bibitem{csharpgen}
Microsoft.
\newblock {C\# Language Specification: Enumerable Objects}, 2018.
\newblock URL:
  \url{https://github.com/dotnet/csharplang/blob/master/spec/classes.md\#enumerable-objects}.

\bibitem{jsgen}
Mozilla.
\newblock {Mozilla Developer Network Documentation: function*}, 2018.
\newblock URL:
  \url{https://developer.mozilla.org/en-US/docs/Web/JavaScript/Reference/Statements/function*}.

\bibitem{mycroft16}
Alan Mycroft, Dominic Orchard, and Tomas Petricek.
\newblock Effect systems revisited --- control-flow algebra and semantics.
\newblock In {\em Semantics, Logics, and Calculi}. Springer, 2016.

\bibitem{neamtiu2008contextual}
Iulian Neamtiu, Michael Hicks, Jeffrey~S Foster, and Polyvios Pratikakis.
\newblock Contextual effects for version-consistent dynamic software updating
  and safe concurrent programming.
\newblock In {\em POPL}, pages 37--49, 2008.

\bibitem{nielson1993cml}
Flemming Nielson and Hanne~Riis Nielson.
\newblock From cml to process algebras.
\newblock In {\em {CONCUR}}, 1993.

\bibitem{Petricek2014CCC}
Tomas Petricek, Dominic Orchard, and Alan Mycroft.
\newblock Coeffects: A calculus of context-dependent computation.
\newblock In {\em {ICFP}}, 2014.

\bibitem{pirog2019typed}
Maciej Pir{\'o}g, Piotr Polesiuk, and Filip Sieczkowski.
\newblock Typed equivalence of effect handlers and delimited control.
\newblock In {\em 4th International Conference on Formal Structures for
  Computation and Deduction (FSCD 2019)}. Schloss Dagstuhl-Leibniz-Zentrum fuer
  Informatik, 2019.

\bibitem{plotkin2009handlers}
Gordon Plotkin and Matija Pretnar.
\newblock Handlers of algebraic effects.
\newblock In {\em European Symposium on Programming}, pages 80--94. Springer,
  2009.

\bibitem{Pombrio:2018:ITR:3192366.3192398}
Justin Pombrio and Shriram Krishnamurthi.
\newblock Inferring type rules for syntactic sugar.
\newblock In {\em Proceedings of the 39th ACM SIGPLAN Conference on Programming
  Language Design and Implementation}, PLDI 2018, pages 812--825, New York, NY,
  USA, 2018. ACM.
\newblock URL: \url{http://doi.acm.org/10.1145/3192366.3192398}, \href
  {http://dx.doi.org/10.1145/3192366.3192398}
  {\path{doi:10.1145/3192366.3192398}}.

\bibitem{pretnar2014effect}
Matija Pretnar and Andrej Bauer.
\newblock An effect system for algebraic effects and handlers.
\newblock {\em Logical Methods in Computer Science}, 10, 2014.

\bibitem{shan2007static}
Chung-chieh Shan.
\newblock A static simulation of dynamic delimited control.
\newblock {\em Higher-Order and Symbolic Computation}, 20(4):371--401, 2007.

\bibitem{Sitaram1993}
Dorai Sitaram.
\newblock Handling control.
\newblock In {\em {PLDI}}, 1993.

\bibitem{SitaramF90a}
Dorai Sitaram and Matthias Felleisen.
\newblock Control delimiters and their hierarchies.
\newblock {\em Lisp and Symbolic Computation}, 3(1):67--99, 1990.

\bibitem{SitaramF90b}
Dorai Sitaram and Matthias Felleisen.
\newblock Reasoning with continuations {II:} full abstraction for models of
  control.
\newblock In {\em {LISP} and Functional Programming}, pages 161--175, 1990.
\newblock URL: \url{http://doi.acm.org/10.1145/91556.91626}, \href
  {http://dx.doi.org/10.1145/91556.91626} {\path{doi:10.1145/91556.91626}}.

\bibitem{Skalka2008}
Christian Skalka.
\newblock Types and trace effects for object orientation.
\newblock {\em Higher-Order and Symbolic Computation}, 21(3), 2008.

\bibitem{skalka2008types}
Christian Skalka, Scott Smith, and David Van~Horn.
\newblock Types and trace effects of higher order programs.
\newblock {\em Journal of Functional Programming}, 18(2), 2008.

\bibitem{suenaga2008type}
Kohei Suenaga.
\newblock Type-based deadlock-freedom verification for non-block-structured
  lock primitives and mutable references.
\newblock In {\em {APLAS}}, 2008.

\bibitem{tate13}
Ross Tate.
\newblock The sequential semantics of producer effect systems.
\newblock In {\em {POPL}}, 2013.

\bibitem{pythongen}
Python~Development Team.
\newblock {Python Enhancement Proposal 255: Simple Generators}, 2001.
\newblock URL: \url{https://www.python.org/dev/peps/pep-0255/}.

\bibitem{tofte1997region}
Mads Tofte and Jean-Pierre Talpin.
\newblock Region-based memory management.
\newblock {\em Information and computation}, 132(2):109--176, 1997.

\bibitem{tov2011}
Jesse~A. Tov and Riccardo Pucella.
\newblock A theory of substructural types and control.
\newblock In {\em {OOPSLA}}, 2011.

\bibitem{wand1989type}
Mitchell Wand.
\newblock Type inference for record concatenation and multiple inheritance.
\newblock In {\em Logic in Computer Science, 1989. LICS'89, Proceedings.,
  Fourth Annual Symposium on}, pages 92--97. IEEE, 1989.

\end{thebibliography}
 
\ifTR
\clearpage
\appendix

\section{Hypothetical Typing Derivation for Infinite Loops}
\label{apdx:dloop}
Figure \ref{fig:infloop} gives a hypothetical typing derivation for an infinite loop, assuming no control effects escape the loop's body.  Choosing as in Section \ref{sec:infloop}, $Q_p=Q_e^*$ and $C_p=\{\mathsf{replace}~\ell : Q_e^*\leadsto\mathsf{unit}\}$ makes this derivation valid.  In that case, the final underlying effect in the derivation is equal to $Q_e\sqcup(Q_e^*)=Q_e^*$.

    \begin{figure*}\scriptsize
        \begin{mathpar}
            C = (Q_e\rhd C_p)\cup\{\mathsf{replace}~\ell:Q_e\rhd Q_p\leadsto\mathsf{unit}\}\and
        \mathscr{J}_{val}=\inferrule{
            \Gamma\vdash \mathsf{unit}<:\mathsf{unit}\\
            \Gamma\vdash (\emptyset, C, Q_e) \sqsubseteq (\emptyset,C_p,Q_p)
        }{
            \Gamma\vdash \mathsf{validEffects}(\{\mathsf{prophecy}~\ell~(\emptyset,C_p,Q_p)\leadsto\mathsf{unit}~\mathsf{obs}~(\emptyset,C,Q_e)\}, C, Q_e, \ell, \mathsf{unit})
        }\and
        \mathscr{J}_{cc}=
        \inferrule{\ldots}{
                    \Gamma\vdash (\mathsf{call/cc}~\ell~(\lambda k\ldotp k)) : \mu X\ldotp \mathsf{cont}~\ell~X~(\emptyset,C_p,Q_p)~\mathsf{unit} \mid (\{\mathsf{prophecy}~\ell~(\emptyset,C_p,Q_p)\leadsto\mathsf{unit}~\mathsf{obs}~(\emptyset,\emptyset,I)\},\emptyset,I)}
        \and
        \mathscr{J}_{body}=
            \inferrule{
                \mathscr{J}_{cc}\\
                \inferrule{
                    \inferrule{\mathsf{Assumption}}{\Gamma,cc:\ldots\vdash e : \tau \mid \chi_e}\\
                    \inferrule{\ldots}{\Gamma,cc:\ldots\vdash cc~cc : \mathsf{unit} \mid (\emptyset,\pblock{C_p}{\ell}\cup\{\mathsf{replace}~\ell:Q_p\leadsto\mathsf{unit}\},I)}
                }{
                    \Gamma,cc:(\mu X\ldotp \mathsf{cont}~\ell~X~(\emptyset,C_p,Q_p)~\mathsf{unit})\vdash (e;~cc~cc) : \mathsf{unit} \mid \chi_e\rhd(\emptyset,\pblock{C_p}{\ell}\cup\{\mathsf{replace}~\ell:Q_p\leadsto\mathsf{unit}\},I)
                }
            }{
                \Gamma\vdash (\mathsf{let}\;cc = (\mathsf{call/cc}~\ell~(\lambda k\ldotp k))~\mathsf{in}~(e; cc\;cc)) : \mathsf{unit} \mid (\{\mathsf{prophecy}~\ell~(\emptyset,C_p,Q_p)\leadsto\mathsf{unit}~\mathsf{obs}~(\emptyset,C,Q_e)\},C,Q_e)
            }
        \and
        \inferrule{
            \mathscr{J}_{body}\\
            \inferrule{\ldots}{\Gamma\vdash (\lambda\_\ldotp \mathsf{tt}) : \mathsf{unit}\xrightarrow{(\emptyset,\emptyset,I)}\mathsf{unit}}\\
            \mathscr{J}_{val}\\
        }{
            \Gamma\vdash (\%\;\ell\;(\mathsf{let}\;cc = (\mathsf{call/cc}~\ell~(\lambda k\ldotp k))~\mathsf{in}~(e; cc\;cc))\;(\lambda\_\ldotp \mathsf{tt})) : \mathsf{unit} \mid (\emptyset,\emptyset,Q_e\sqcup\bigsqcup \punblock{C}{\ell}|^I_\ell)
        }
        \end{mathpar}
        \caption{Typing infinite loops.  We assume $\Gamma\vdash e : \tau \mid \chi_e$, where $\chi_e=(\emptyset,\emptyset,Q_e)$.}
        \label{fig:infloop}
    \end{figure*}

\section{Context Typing and Substitutions}
For syntactic type safety, we must give types to terms that exist only at runtime, which include reified continuations.
For this we introduce the evaluation context type $\tau/\chi{\leadsto}\tau'/\chi'$, which characterizes an evaluation context with a hole of type $\tau$, which when filled by an expression of appropriate type and effect at most $\chi$ yields an expression producing $\tau'$ with overall effect $\chi'$.  This is only used by the context-typing judgment $\Sigma;\Gamma\vdash E : \tau/\chi{\leadsto}\tau'/\chi'$, which has no effect of its own, because evaluation contexts do not appear in expression positions during evaluation.
This judgment plays both a convenient administrative role in the soundness proof, and a role in typing the runtime form of continuations.

Note that we have \emph{avoided} explicitly tracking a notion of latent effect for a continuation while typing the main program expression.  There are two reasons for this.  First, doing this explicitly would make the type rules significantly more complex due to the need to identify various contexts and associate latent effects to them.  Second, it is unnecessary in the presence of prophecies: the observations made by a prophecy capture the latent effect between the point of the continuation capture and the enclosing prompt of the same tag.  And these are the only continuations for which a latent effect is useful.
Our use of prophecies permits the inference of these latent effects from the characterization above, in cases where they are required (see Lemma \ref{lem:valid_prophecies}).

The context typing judgment is defined in parallel with typing derivations, one case for each possible way of typing an evaluation context.  For example, there are two rules for typing the function-hole contexts:
\begin{mathpar}
    \inferrule[T-CtxtFunApp]{
        \Sigma;\Gamma\vdash E' :: \tau/\chi\leadsto(\sigma\xrightarrow{\chi_l}\sigma')/\chi'\\
        \Sigma;\Gamma\vdash e : \sigma \mid \chi_e
    }{
        \Sigma;\Gamma\vdash (E'~e) : \tau/\chi\leadsto\sigma'/(\chi'\rhd\chi_e\rhd\chi_l)
    }
    \and\inferrule[T-CtxtContApp]{
        \Sigma;\Gamma\vdash E' :: \tau/\chi\leadsto(\mathsf{cont}~\ell~\sigma~{(P,C,Q)}~\sigma')/\chi'\\
        \Sigma;\Gamma\vdash e : \sigma \mid \chi_e
    }{
        \Sigma;\Gamma\vdash (E'~e) : \tau/\chi\leadsto\sigma'/(\chi'\rhd\chi_e\rhd(\pblock{P}{\ell},C\cup\{\mathsf{replace}~\ell~Q\leadsto\tau'\},I))
    }
\end{mathpar}
The other context typing rules are defined similarly, each effectively exchanging one inductive hypothesis for a recursive context typing hypothesis with the same type and effect for the hole.
The base case is the natural rule for the hole, requiring that the type and effect of the value plugged into the hole are subtype or subeffect of the of the ``plugged'' context's type (since plugging an expression into the empty context is simply that expression):
\[
    \inferrule[T-CtxtHole]{
        \tau <: \tau'\\
        \chi \sqsubseteq \chi'
     }{\Sigma;\Gamma\vdash \bullet :: \tau/\chi\leadsto\tau'/\chi'}
\]
The context typing judgment is defined mutually with the term typing, as the type rule for continuation values refers back to the context typing judgment:

\begin{mathpar}
    \inferrule[T-ContC]{
        \Sigma;\Gamma\vdash E :: \tau/(\emptyset,\emptyset,I)\leadsto\tau_0'/\chi_0\\
        \tau_0'<:\tau'\\
        \punblock{P_0}{\ell}\sqsubseteq\punblock{P}{\ell}\\
        \punblock{C_0}{\ell}\sqsubseteq \punblock{C}{\ell}\\
        Q_0\sqsubseteq Q
    }{
        \Sigma;\Gamma\vdash (\mathsf{cont}_\ell^{\tau'}~E) : \mathsf{cont}~\ell~\tau~\chi~\tau' \mid I
    }
\end{mathpar}

\subsection{Context Decomposition}
Because the preservation proof will destructure full expressions into evaluation contexts and redexes, and we will require both local typing information about the redex and information about replacing the redex within the context, we must be able to decompose the typing derivation of a (filled) evaluation context in parallel with the operational semantics.
\begin{lemma}[Context Typing Decomposition]
    \label{lem:context_typing_decomp}
If
\begin{itemize}
    \item $\Sigma;\Gamma\vdash E[e] : \tau \mid \chi$
\end{itemize}
then there exists a $\tau'$, $\chi'$, such that:
\begin{itemize}
\item $\Sigma;\Gamma\vdash e : \tau' \mid \chi'$
\item $\Sigma;\Gamma\vdash E :: \tau'/\chi'{\leadsto}\tau/\chi$
\end{itemize}
\end{lemma}
\begin{proof}
By induction on $E$ (with other variables universally quantified in the inductive hypothesis).
\begin{itemize}
    \item Case $E=\bullet$: Here $E[e]=e$, so $\tau'=\tau$, $\chi'=\chi$, and $\chi''=I$.
    \item Case $E=(E'[e]~e')$:
        Here we have two cases, for application of functions or application of continuations.  We present the function application case; the continuation application is similar.
        Given:
        \begin{itemize}
        \item $\Sigma;\Gamma\vdash(E'[e]~e') : \tau \mid \chi$
        \end{itemize}
        By inversion on the typing derivation, for the function application case:
        \begin{itemize}
        \item $\Sigma;\Gamma\vdash E'[e] : \tau_{e'}\xrightarrow{\chi_{latent}}\tau \mid \chi_{f}$
        \item $\Sigma;\Gamma\vdash e' : \tau_{e'} \mid \chi_{e'}$
        \item $\chi=\chi_f\rhd\chi_{e'}\rhd\chi_{latent}$
        \end{itemize}
        Via the inductive hypothesis there exists some $\tau'$, $\chi'$, such that:
        \begin{itemize}
            \item $\Sigma;\Gamma\vdash e : \tau'\mid \chi'$
            \item $\Sigma;\Gamma\vdash E' :: \tau'/\chi'{\leadsto}(\tau_{e'}\xrightarrow{\chi_{latent}}\tau)/\chi_f$
        \end{itemize} 
        We can then invert on context typing and use \textsc{T-App} to prove
        \[\Sigma;\Gamma\vdash E :: \tau'/\chi'{\leadsto}\tau/\chi\]
        The inversion on typing produces a second case, for applying continuations, which proceeds similarly.
    \item Case $E=(v~E'[e])$: Similar to the other function application context.
    \item Case $E=(\%~\ell~E'[e]~v)$: Similar to previous cases.  
    \item Case $E=(\mathsf{call/cc}~\ell~E'[e])$: Similar to previous cases.
\end{itemize}
\end{proof}

Note that when decomposing contexts, the redex is typed in the same type environment as the surrounding evaluation context.  This is a consequence of the fact that no evaluation context reaches ``under'' binders like inside a lambda expression.

\subsection{Valid Prophecies}
\label{apdx:valid_prophecies}
One of the most subtle parts of the effect system is the use of prophecies to capture the residual effect of various evaluation contexts, in a non-local manner.  Intuitively, a prophecy captures all possible effects from the point of the prophecy (the point where the continuation capture would be the next expression to reduce) up to an enclosing prompt.  
In particular, notice that post-composing an effect after one with a prophecy performs the same ``transformations'' on the $C$ and $Q$ components of the prophecy as on those components of the actual effect, effectively type-checking a context twice simultaneously.
The lemma below makes the intuition precise, and shows that the type system in fact matches that intuition.

\begin{lemma}[Valid Prophecies]
\label{lem:valid_prophecies}
For any $\Sigma$, $\Gamma$, $E$, $\tau$, $\chi$, $\chi'$, $\sigma$, $\ell$, $\chi_\ell$, $\gamma$, if
\begin{itemize}
\item $\Sigma;\Gamma\vdash E : \tau/\chi\leadsto\sigma/\chi'$ and
\item $\mathsf{prophecy}~\ell~\chi_\ell\leadsto\gamma~\mathsf{obs}~(\emptyset,\emptyset,I)\in P_\chi$
\item $E$ contains no prompts for $\ell$
\end{itemize}
then there exists a $P$, $C$, and $Q$ such that:
\begin{itemize}
\item $\mathsf{prophecy}~\ell~\chi_\ell\leadsto\gamma~\mathsf{obs}~(P,C,Q)\in P_{\chi'}$
\item $\Sigma;\Gamma\vdash E :: \tau/I\leadsto\sigma/(P,C,Q)$
\end{itemize}
\end{lemma}
\begin{proof}
By induction on $E$.  We present a demonstrative inductive case and the one interesting case.
\begin{itemize}
\item Case $E=(E'~e)$:
	Inversion on the context typing produces a case for function application, and a case for continuation application.  We show the former; the latter is similar.  By that inversion and the inductive hypothesis:
	\begin{itemize}
	\item $\mathsf{prophecy}~\ell~\chi_\ell\leadsto\gamma~\mathsf{obs}~(P',C',Q')\in P_{E'}$
    \item $\Sigma;\Gamma\vdash E' :: \tau/I\leadsto(\tau'\xrightarrow{\chi_f}\sigma)/(P',C',Q')$
    \item $\Sigma;\Gamma\vdash e : \tau' \mid \chi_e$
    \item $P_E=(P_{E'}\blacktriangleright\chi_e\blacktriangleright\chi_f)\cup(P_e\blacktriangleright\chi_f)\cup P_f$
    \end{itemize}
    Applying \textsc{T-CtxtFunApp} with the ``plugged'' type for $E'$ produces the expected result type ($\sigma$) and a result effect \[(P,C,Q)\overset{def}{=}(P',C',Q')\rhd\chi_e\rhd\chi_f\]
    And given the prophecy observing $(P',C',Q')$ in $P_{E'}$, we know the $(P,C,Q)$ above is present in $P_E$: $P_{E'}\blacktriangleright\chi_e\blacktriangleright\chi_f$ will contain 
    \[
        \mathsf{prophecy}~\ell~\chi_\ell\leadsto\gamma~\mathsf{obs}~(P',C',Q') \blacktriangleright\chi_e\blacktriangleright\chi_f
    \]
\item Case $E=(\%~\ell'~E'~h)$:
	By assumption, $\ell\neq\ell'$.  By the inversion on context typing and the inductive hypothesis:
    \begin{itemize}
	\item $\mathsf{prophecy}~\ell~\chi_\ell\leadsto\gamma~\mathsf{obs}~(P',C',Q')\in P_{E'}$
    \item $\Sigma;\Gamma\vdash E' ::\tau/I\leadsto \sigma/(P',C',Q')$
    \item $\Sigma;\Gamma\vdash h : \sigma'\xrightarrow{(\emptyset,\emptyset,Q_h)}\sigma$
    \item $P_E = \punblock{P_{E'}}{\ell'}\setminus^{Q_h}\ell'$
    \item $C_E = \punblock{C_{E'}}{\ell'}\setminus^{Q_h}\ell'$
    \item $Q_E = Q_{E'}\sqcup\bigsqcup \punblock{C_{E'}}{\ell'}|^{Q_h}_{\ell'}$
    \item $\mathsf{validEffects}(P_{E'},C_{E'},Q_{E'},\ell',\sigma,\sigma')$ ($\sigma'$ is the argument type of the handler)
    \end{itemize}
    Note that the changes from the body effect to prompt effect imposed by \textsc{T-CtxtPrompt} preserve the prophecy of interest (suitably modified itself).  For choices:
	\begin{itemize}
	\item $P=\punblock{P'}{\ell'}\setminus^{Q_h}\ell'$
    \item $C=\punblock{C'}{\ell'}\setminus^{Q_h}\ell'$
    \item $Q=Q'\sqcup\bigsqcup(\punblock{C'}{\ell'}|^{Q_h}_{\ell'})$
    \end{itemize}
    we may apply \textsc{T-CtxtPrompt} to give a context typing for $(\%~\ell~E'~h)$ (The choices for $P$, $C$, and $Q$ directly imply \textsf{validEffects}).  Now we must show \\\mbox{$\prophecy{\ell~\chi_\ell\leadsto\gamma}{(P,C,Q)}\in P_E$}.  Because of the prophecy we know of in $P_{E'}$, the following is in $P_E$:
    \[\begin{array}{rl}
        &\punblock{\mathsf{prophecy}~\ell~\chi_\ell\leadsto\gamma~\mathsf{obs}~(P',C',Q')}{\ell'}\setminus^{Q_h}\ell'  \\
        =& \mathsf{prophecy}~\ell~\chi_\ell\leadsto\gamma~\mathsf{obs}~(\punblock{P'}{\ell'}\setminus^{Q_h}\ell',\punblock{C'}{\ell'}\setminus^{Q_h}\ell',Q'\sqcup\bigsqcup(\punblock{C'}{\ell'}|^{Q_h}_{\ell'})) \\
        =& \mathsf{prophecy}~\ell~\chi_\ell\leadsto\gamma~\mathsf{obs}~(P,C,Q))
    \end{array}
    \]
\end{itemize}
\end{proof}

\subsection{Evaluation Contexts and Residual Effects}
The canonical evaluation rule in the presence of evaluation contexts is \textsc{E-Context}, which replaces an evaluation context's hole with a new expression.  In cases where the effect of the redex is preserved --- such as the last step of applying a function, where the direct effects of the function and parameter are both $I$ and the unfolding merely lifts the function's latent effect to be the direct effect of the applied function --- context typing via \textsc{T-Context} suffices.  However, the other major class of reductions are the cases where some primitive with a non-identity effect is evaluated, leaving a value in its place.  Since in general $I$ --- the effect of all values --- is not a least element in an effect quantale's partial order, this will not be a subeffect of the original hole effect.

In these cases, as in prior sequential effect systems, the sequential composition of the just-evaluated effect \emph{followed by} the residual effect should be a subeffect of the original.  In the presence of our control effects, this intuition still holds, but the details are more complex because we must essentially prove that many subeffect relationships holding before the reduction remain true on ``suffixes'' of control effects.

\begin{lemma}[Reduct Effects]
    \label{lem:reduct_effects}
If
\begin{itemize}
\item $\Sigma;\Gamma\vdash E :: \tau/\chi\leadsto\tau'/\chi'$
\item $\Sigma;\Gamma\vdash e : \tau \mid \chi_e$
\item $\Sigma;\Gamma\vdash q\rhd\chi_e\sqsubseteq\chi$
\end{itemize}
then there exists a $\chi''$ such that
\begin{itemize}
\item $\Sigma;\Gamma\vdash E[e] : \tau' \mid \chi''$
\item $\Sigma;\Gamma\vdash q\rhd\chi''\sqsubseteq\chi'$
\end{itemize}
\end{lemma}
\begin{proof}
By induction on E, leaving all else universally quantified.  Most cases are straightforward uses of the inductive hypothesis. The exception is the case for prompts, which relies on verifying that the transformations to the body's prophecy and control effects preserves subtyping.
\begin{itemize}
\item Case $E=\cdot$: Trivial.
\item Case $E=(E'~e')$:
	By inversion on context typing, showing only the function application case and omitting the similar continuation application case: 
	\begin{itemize}
	\item $\Sigma;\Gamma\vdash E' :: \tau/\chi\leadsto[\sigma\xrightarrow{\gamma}\tau']/\chi_f$
	\item $\Sigma;\Gamma\vdash \chi'=\chi_f\rhd\chi_{e'}\rhd\gamma$
	\end{itemize}
	From the inductive hypothesis, we may conclude:
	\begin{itemize}
	\item $\Sigma;\Gamma\vdash E'[e] : [\sigma\xrightarrow{\gamma}\tau'] \mid \chi_f'$
	\item $\Sigma;\Gamma\vdash q\rhd\chi_f'\sqsubseteq\chi_f$
    \end{itemize}
    Therefore, let $\chi''=\chi_f'\rhd\chi_{e'}\rhd\gamma$.  Then the typing result for $E[e]$ holds by \textsc{T-App} (note we've suppressed the details of handling the choice of $e'$ being a value or the function having non-dependent type).
    For the subeffect obligation:
    \begin{itemize}
        \item $\Sigma;\Gamma\vdash q\rhd\chi''=q\rhd \chi_f'\rhd\chi_{e'}\rhd\gamma \sqsubseteq \chi_f\rhd\chi_{e'}\rhd\gamma=\chi'$
    \end{itemize}
    follows from the definition of $\chi''$, the definition of $\chi'$ from inversion, and the inductive result.
\item Case $E=(v~E')$: Similar to other cases.
\item Case $E=(\%~\ell~E'~h)$: 
    By inversion on context typing
    \begin{itemize}
    \item $\Sigma;\Gamma\vdash E' :: \tau/\chi\leadsto\tau'/\chi_b$
    \item $\chi_b=(P_b,C_b,Q_b)$
    \item $\chi'=(\punblock{P_b}{\ell}\setminus\ell, \punblock{C_b}{\ell}\setminus\ell, Q_b\sqcup\bigsqcup \punblock{C_b}{\ell}|^{Q_h}_\ell)$
    \item $\Sigma;\Gamma\vdash h : \sigma\xrightarrow{(\emptyset,\emptyset,Q_h)}\tau \mid I$
    \item $\Sigma;\Gamma\vdash\mathsf{validEffects}(P_b, C_b, Q_b, \ell, \tau', \sigma)$
    \end{itemize}
    By the inductive hypothesis:
    \begin{itemize}
        \item $\Sigma;\Gamma\vdash E'[e] : \tau' \mid \chi_b'$
        \item $q\rhd \chi_b'\sqsubseteq\chi_b$
    \end{itemize}
    Destructure $\chi_b'$ as $(P_b',C_b',Q_b')$.
    Let $\chi''=(\punblock{P_b'}{\ell}\setminus\ell,\punblock{C_b'}{\ell}\setminus\ell,Q_b'\sqcup\bigsqcup \punblock{C_b}{\ell}|^{Q_h}_\ell)$.
    Because $q\rhd(P_b', C_b', Q_b') \subseteq (P_b, C_b, Q_b)$, 
    we may conclude
    $q\rhd(\punblock{P_b'}{\ell}\setminus\ell, \punblock{C_b'}{\ell}\setminus \ell, Q_b'\sqcup\bigsqcup \punblock{C_b}{\ell}|^{Q_h}_\ell)\sqsubseteq(\punblock{P_b}{\ell}\setminus\ell, \punblock{C_b}{\ell}\setminus\ell, Q_b\sqcup\bigsqcup \punblock{C_b}{\ell}|^{Q_h}_\ell)$.
\item Case $E=(\mathsf{call/cc}~\ell~E')$: Similar to other cases.
\end{itemize}
\end{proof}

\section{Soundness (Type Safety)}
\label{apdx:soundness}

We prove syntactic type safety for a lightly type-annotated source language.
We require some type annotations because the dynamic semantics form a labelled transition system where the label is the effect of the reduction.  Because some of the control effects contain types, we must add those to the term language.
Specifically:
\begin{itemize}
    \item $\mathsf{abort}^\rho$ is the abort operator labelled with the type of the value passed to the handler.
    \item $\mathsf{call/cc}^\rho_{\chi}$ ``predicts'' a continuation result type (the type of the enclosing prompt) of $\rho$ and latent continuation effect of $\chi$ (as in the continuation type, this is a flattened \emph{underlying} effect)
\end{itemize}
In each case, the runtime typing rule is modified to enforce the correct relationship between the term and the type.  These type-annotations are straightforward to produce from a valid source typing derivation.

The structure of our proof borrows heavily from Gordon's modular proof of type safety for a lambda calculus defined with respect to an unspecified effect quantale, where the language is parameterized by a selection of primitives, new values, new types, a notion of state, state types, and operations and relations giving types to the new primitives, values, and states.  The proof is also parameterized by a lemma that amounts to one-step type preservation for executing fully-applied primitive operations --- these are the only parts of the language that may modify the state (the rest of the framework is defined without knowledge of the state's internals).  Sufficient restrictions are placed on these parameters to ensure various traditional lemmas continue to hold (for example, ensuring that the primitives do not add a third boolean, so the typical Canonical Forms lemma stands).

We impose an additional requirement on the parameters for primitive typing, beyond what Gordon requires.  Whereas Gordon requires that for any primitive, only its final effect is non-$I$, we also require that all control behaviors and block sets for any primitive added are empty.  This ensures that all primitives are in fact local operations.  We retain Gordon's model of dynamic primitive behavior, which reduces fully-applied primitives to a new value, transforms the (pluggable) state, and produces a dynamic effect --- in the \emph{underlying} effect quantale, rather than the control-enhanced one.

\begin{lemma}[Redex Preservation]
    \label{lem:redex_preservation}
If 
\begin{itemize}
\item $\vdash\Gamma$
\item $\vdash\sigma:\Sigma$
\item $\Sigma;\Gamma\vdash e : \tau \mid \chi$
\item $\sigma;e\overset{q}{\rightarrow}\sigma';e'$
\end{itemize}
then there exists $\Sigma'$ and $\chi'$ such that 
\begin{itemize}
\item $\Sigma\le\Sigma'$
\item $\vdash\sigma':\Sigma'$
\item $\Sigma';\Gamma\vdash e' : \tau \mid \chi'$
\item $q\rhd \chi'\sqsubseteq\chi$
\end{itemize}
\end{lemma}
\begin{proof}
    By induction on the reduction step, followed by inversion on the typing derivation in each case.  The proof is nearly identical to Gordon's proofs~\cite{ecoop17,quantalesjournal} beyond the addition of equivariant recursive types, non-dependent products and sum types, subsumption, and a new straightforward case for \textsc{E-PromptVal}.
\end{proof}

\begin{lemma}[Context Substitution]
    \label{lem:ctxt_sub}
If 
\begin{itemize}
\item $\Sigma;\Gamma\vdash E :: \tau/\chi\leadsto\sigma/{\chi'}$
\item $\Sigma;\Gamma\vdash e : \tau \mid \chi$
\end{itemize}
then $\Sigma;\Gamma\vdash E[e] : \sigma \mid \chi'$
\end{lemma}
\begin{proof}
    By straightforward induction on the context typing.
\end{proof}

Finally, we are ready to tackle the central preservation lemma.
\begin{lemma}[Context Reduction]
    \label{lem:preservation}
If
\begin{itemize}
\item $\vdash\Gamma$
\item $\vdash\sigma:\Sigma$
\item $\Sigma;\Gamma\vdash e : \tau \mid \chi$
\item $\sigma;e\overset{q}{\Rightarrow}\sigma';e'$
\end{itemize}
then there exists a $\Sigma'$ $\chi'$ such that 
\begin{itemize}
\item $\Sigma\le\Sigma'$
\item $\vdash\sigma':\Sigma'$
\item $\Sigma';\Gamma\vdash e' : \tau \mid \chi'$
\item $q\rhd{\chi'}\sqsubseteq{\chi}$
\end{itemize}
\end{lemma}
\newcommand{\case}[1]{\item Case \textsc{#1}:\xspace}
\begin{proof}
By induction on the derivation of $\sigma;e\overset{q}{\Rightarrow}\sigma';e'$.
In each case below \emph{other than} the case for the context reduction (\textsc{E-Context}), nothing will change the state.  So in all cases other than \textsc{E-Context}, we choose $\Sigma'\overset{def}{=}\Sigma$, and the semantics ensure $\sigma=\sigma'$, and so in all cases $\vdash\sigma':\Sigma'\Leftrightarrow\vdash\sigma:\Sigma$.
Thus we present only the expression- and effect-related details below.  Moreover, we continue to use $\sigma$ as an additional meta-variable for types, rather than as a meta-variable for states, to minimize the number of prime marks that must be counted by the reader (or author) of the proof.
\begin{itemize}
    \case{E-Context} This case follows from 
    context decomposition (Lemma \ref{lem:context_typing_decomp}), 
    redex reduction (Lemma \ref{lem:redex_preservation}),
    and the reduct effects lemma (Lemma \ref{lem:reduct_effects}).
    In this case,
    \begin{itemize}
        \item $\sigma;e\xrightarrow{q}\sigma';e'$
        \item $\Sigma;\Gamma\vdash E[e] : \tau \mid \chi$
        \item $\vdash \Gamma$
        \item $\vdash \sigma : \Sigma$
    \end{itemize}
    By context decomposition (Lemma \ref{lem:context_typing_decomp}):
    \begin{itemize}
        \item $\Sigma;\Gamma\vdash e : \tau_e \mid \chi_e$
        \item $\Sigma;\Gamma\vdash E :: \tau_e/\chi_e\leadsto\tau/\chi$
    \end{itemize}
    By redex preservation (Lemma \ref{lem:redex_preservation}), there exist $\Sigma'$ and $\chi_e\prime$ such that:
    \begin{itemize}
        \item $\Sigma \le \Sigma'$
        \item $\vdash \sigma' : \Sigma$
        \item $\Sigma';\Gamma\vdash e' : \tau_e \mid \chi_e\prime$
        \item $\Sigma';\Gamma\vdash q \rhd \chi_e\prime \sqsubseteq \chi_e$
    \end{itemize}
    By weakening and repetition of context decomposition:
    \begin{itemize}
        \item $\Sigma';\Gamma\vdash E :: \tau_e/\chi_e\leadsto\tau/\chi$
    \end{itemize}
    Then by the reduct effects lemma (Lemma \ref{lem:reduct_effects}), we can derive the appropriate result for some $\chi'$:
    \begin{itemize}
        \item $\Sigma';\Gamma\vdash E[e'] : \tau \mid \chi'$
        \item $\Sigma';\Gamma\vdash q\rhd \chi'\sqsubseteq\chi$
    \end{itemize}
    \case{E-Abort}
        In this case, 
        \begin{itemize}
            \item $e=E[(\%~\ell~E'[(\mathsf{abort}^\rho~\ell~v)]~h)]$,
            \item $q=I$
            \item $e'=E[h~v]$
            \item $E'$ contains no prompts for $\ell$
        \end{itemize}
        By context decomposition (Lemma \ref{lem:context_typing_decomp}), there exist $\tau_{prompt}$, $\chi_{prompt}$, where:
        \begin{itemize}
            \item $\Sigma;\Gamma\vdash(\%~\ell~E'[(\mathsf{abort}^\rho~\ell~v)]~h) : \tau_{prompt} \mid \chi_{prompt}$
            \item $\Sigma;\Gamma\vdash E :: \tau_{prompt}/\chi_{prompt}\leadsto\tau/\chi$
        \end{itemize}
        By inversion on the prompt's typing:
        \begin{itemize}
            \item $\Sigma;\Gamma\vdash E'[(\mathsf{abort}^\rho~\ell~v)] : \tau_{prompt} \mid (P_{body},C_{body},Q_{body})$
            \item $\Sigma;\Gamma\vdash h : \sigma_h\xrightarrow{(\emptyset,\emptyset,Q_h)}\tau_{prompt} \mid (\emptyset,\emptyset,I)$
            \item $\Sigma;\Gamma\vdash\mathsf{validEffects}(P_{body},C_{body},Q_{body}, \ell, \tau_{prompt}, \sigma_h)$
            \item $\chi_{prompt}=(\punblock{P_{body}}{\ell}\setminus^{Q_h}\ell, \punblock{C_{body}}{\ell}\setminus^{Q_h}\ell, Q_{body}\sqcup\bigsqcup \punblock{C_{body}}{\ell}|^{Q_h}_\ell)$
        \end{itemize}
        Applying context decomposition again to $E'$ and its contents:
        \begin{itemize}
            \item $\Sigma;\Gamma\vdash (\mathsf{abort}^\rho~\ell~v) : \tau_{abrt} \mid \chi_{abrt}$
            \item $\Sigma;\Gamma\vdash E' :: \tau_{abrt}/\chi_{abrt}\leadsto\tau_{prompt}/\chi_{prompt}$
        \end{itemize}
        By inversion on the typing of \textsf{abort} and value typing:
        \begin{itemize}
            \item $\Sigma;\Gamma\vdash v : \sigma \mid I$
            \item $\Sigma;\Gamma\vdash (\emptyset, \{\mathsf{abort}~\ell~I\leadsto\sigma\},I)\sqsubseteq\chi_{abrt}$
        \end{itemize}
        Because abort effects for $\ell$ are propagated through $E'$ (which contains no $\ell$-prompts) without accumulating new prefixes (since $E'$ is an evaluation context):
        \begin{itemize}
            \item $\mathsf{abort}~\ell~I\leadsto\sigma\in C_{body}$
        \end{itemize}
        By inversion on the $\mathsf{validEffects}$ conclusion above:
        \begin{itemize}
            \item $\Sigma;\Gamma\vdash \sigma <: \sigma_h$
        \end{itemize}
        Then by \textsc{T-App}:
        \begin{itemize}
            \item $\Sigma;\Gamma\vdash h~v : \tau_{prompt} \mid (\emptyset,\emptyset,Q_h)$
        \end{itemize}
        This is a subeffect of $\chi_{prompt}$, because $(I\rhd Q_h) \in \punblock{C_{body}}{\ell}|^{Q_h}_\ell$ as a result of the abort control effect.  This allows the use of context substitution (Lemma \ref{lem:ctxt_sub}) to place this result back in $E$.
    \case{E-CallCC}
        In this case
        \begin{itemize}
        \item $e=E[(\%~\ell~E'[(\mathsf{call/cc}_\chi^\rho~\ell~k)]~h)]$
        \item $q=I$
        \item $e'=E[(\%~\ell~E'[(k~(\mathsf{cont}^\rho~\ell~E'))]~h)]$
        \end{itemize}
        By context decomposition (Lemma \ref{lem:context_typing_decomp}):
        \begin{itemize}
            \item $\Sigma;\Gamma\vdash(\%~\ell~E'[(\mathsf{call/cc}_\chi^\rho~\ell~k)]~h) : \tau_{prompt} \mid \chi_{prompt}$
            \item $\Sigma;\Gamma\vdash E :: \tau_{prompt}/\chi_{prompt}\leadsto\tau/\chi$
        \end{itemize}
        By inversion on the prompt typing:
        \begin{itemize}
        \item $\Sigma;\Gamma\vdash E'[(\mathsf{call/cc}_\chi^\rho~\ell~k)] : \tau_{prompt} \mid (P_{body}, C_{body}, Q_{body})$
        \item $\Sigma;\Gamma\vdash h : \sigma_h\xrightarrow{(\emptyset,\emptyset,Q_h)}\tau_{prompt} \mid (\emptyset,\emptyset,I)$
        \item $\Sigma;\Gamma\vdash\mathsf{validEffects}(P_{body},C_{body},Q_{body}, \ell, \tau_{prompt}, \sigma_h)$
        \item $\chi_{prompt}=(\punblock{P_{body}}{\ell}\setminus^{Q_h}\ell, \punblock{C_{body}}{\ell}\setminus^{Q_h}\ell, Q_{body}\sqcup\bigsqcup \punblock{C_{body}}{\ell}|^{Q_h}_\ell)$
        \end{itemize}
        Applying context decomposition again to E' and its contents:
        \begin{itemize}
        \item $\Sigma;\Gamma\vdash(\mathsf{call/cc}_\chi^\rho~\ell~k) : \tau_{hole} \mid \chi_{hole}$
        \item $\Sigma;\Gamma\vdash E' :: \tau_{hole}/\chi_{hole}\leadsto\tau_{prompt}/\chi_{prompt}$
        \end{itemize}
        where $k$ is a value.
        By inversion on the \textsf{call/cc} typing and value typing:
        \begin{itemize}
        \item $\Sigma;\Gamma\vdash k : (\mathsf{cont}~\ell~\tau_{hole}~\chi_{proph}~\sigma)\xrightarrow{\chi_k}\tau_{hole} \mid (\emptyset,\emptyset,I)$
        \item $\chi_{hole}=\chi_k\rhd(\{ \mathsf{prophecy}~\ell~\chi_{proph}\leadsto\sigma~\mathsf{obs}~(\emptyset,\emptyset,I) \}, \emptyset, I)$
        \end{itemize}

        By the valid prophecies lemma (Lemma \ref{lem:valid_prophecies}):
        \begin{itemize}
            \item $\mathsf{prophecy}~\ell~\chi_{proph}\leadsto\sigma~\mathsf{obs}~(P_{obs}, C_{obs}, Q_{obs})\in P_{body}$
            \item $\Sigma;\Gamma\vdash E :: \tau_{hole}/(\emptyset,\emptyset,I) \leadsto \tau_{prompt}/(P_{obs}, C_{obs}, Q_{obs})$
        \end{itemize}
        From the \textsf{validEffects} assumption above, conclude:
        \begin{itemize}
            \item $\punblock{P_{obs}}{\ell}\sqsubseteq\punblock{P_{proph}}{\ell}$
            \item $\punblock{C_{obs}}{\ell}\sqsubseteq \punblock{C_{proph}}{\ell}$
            \item $Q_{obs}\sqsubseteq Q_{proph}$
            \item $\Sigma;\Gamma\vdash \tau_{prompt} <: \sigma$
        \end{itemize}
        These are exactly the requirements for the typing of continuation values via \textsc{T-ContC}, thus:
        \begin{itemize}
        \item $\Sigma;\Gamma\vdash(\mathsf{cont}^\rho~\ell~E') : \mathsf{cont}~\ell~\tau_{hole}~\chi_{proph}~\rho \mid I$
        \end{itemize}
        The use of \textsc{T-App} to apply $k$ to the continuation at its assumed type, along with Context Substitution (Lemma \ref{lem:ctxt_sub}) completes the case.
    \case{E-InvokeCC}
    In this case:
    \begin{itemize}
        \item $e=E[(\%~\ell~E'[((\mathsf{cont}^\rho~\ell~E'')~v)]~h)]$
        \item $e'=E[(\%~\ell~(E''[v])~h)]$
        \item $q=(\ell\mapsto\replace{I\leadsto\sigma},I)$
    \end{itemize}
    By Context Decomposition (Lemma \ref{lem:context_typing_decomp}):
    \begin{itemize}
        \item $\Sigma;\Gamma\vdash (\%~\ell~E'[((\mathsf{cont}^\rho~\ell~E'')~v)]~h) : \tau' \mid \chi_{prompt}$
        \item $\Sigma;\Gamma\vdash E :: \tau_{prompt}/\chi_{prompt}\leadsto\tau/\chi$
    \end{itemize}
    By inversion on the prompt typing:
        \begin{itemize}
        \item $\Sigma;\Gamma\vdash E'[((\mathsf{cont}^\rho~\ell~E'')~v)] : \tau_{prompt} \mid (P_{body}, C_{body}, Q_{body})$
        \item $\Sigma;\Gamma\vdash h : \sigma_h\xrightarrow{(\emptyset,\emptyset,Q_h)}\tau_{prompt} \mid (\emptyset,\emptyset,I)$
        \item $\Sigma;\Gamma\vdash\mathsf{validEffects}(P_{body},C_{body},Q_{body}, \ell, \tau_{prompt}, \sigma_h)$
        \item $\chi_{prompt}=(\punblock{P_{body}}{\ell}\setminus^{Q_h}\ell, \punblock{C_{body}}{\ell}\setminus^{Q_h}\ell, Q_{body}\sqcup\bigsqcup \punblock{C_{body}}{\ell}|^{Q_h}_\ell)$
        \end{itemize}
    By context decomposition for $E'$:
    \begin{itemize}
        \item $\Sigma;\Gamma\vdash((\mathsf{cont}^\rho~\ell~E'')~v) : \tau_{hole} \mid \chi_{hole}$
        \item $\Sigma;\Gamma\vdash E' :: \tau_{hole}/\chi_{hole}\leadsto\tau_{prompt}/\chi_{prompt}$
    \end{itemize}
    By inversion on the continuation application, and value typing:
    \begin{itemize}
        \item $\Sigma;\Gamma\vdash (\mathsf{cont}^\rho~\ell~E'') : \mathsf{cont}~\ell~\gamma~(P_{E''}, C_{E''}, Q_{E''})~\rho \mid I$
        \item $\Sigma;\Gamma\vdash v : \gamma \mid I$
        \item $\chi_{hole}=(\pblock{P_{E''}}{\ell}, \pblock{C_{E''}}{\ell}\cup\{ \mathsf{replace}~\ell:Q_{E''}\leadsto\sigma \}, I)$
    \end{itemize}
    By inversion on typing of the continuation (\textsc{T-ContC}):
    \begin{itemize}
        \item $\Sigma;\Gamma\vdash E'' :: \gamma/(\emptyset,\emptyset,I)\leadsto\sigma/(P_{E''}', C_{E''}', Q_{E''}')$
        \item $\punblock{P_{E''}'}{\ell}\sqsubseteq\punblock{P_{E''}}{\ell}$
        \item $\punblock{C_{E''}'}{\ell}\sqsubseteq\punblock{C_{E''}}{\ell}$
        \item $Q_{E''}'\sqsubseteq Q_{E''}$
    \end{itemize}
    Thus we conclude by Lemma \ref{lem:ctxt_sub} for the $v$ at hand:
    \begin{itemize}
        \item 
        $\Sigma;\Gamma\vdash E''[v] : \sigma\mid(P_{E''}', C_{E''}', Q_{E''}')$
    \end{itemize}
    Because unblocked control effects for $\ell$ and control effects blocked until $\ell$ are carried through $E'$ without change (since it is an evaluation context containing no $\ell$-prompts)
    \begin{itemize}
        \item $\pblock{C_{E''}}{\ell}\cup\{ \mathsf{replace}~\ell:Q_{E''}\leadsto\sigma \} \subseteq C_{body}$ 
    \end{itemize}
    Likewise, $\pblock{P_{E''}}{\ell}\subseteq P_{body}$ (originating from $\chi_{hole}$). 
    (Note that knowing the \textsf{replace} effect above is contained in $C_{body}$ means that the original prophecy validation ensured that $\sigma<:\tau_{prompt}$.)

    This is \emph{almost} enough to make the effect of $E''[v]$ a subeffect of $(P_{body},C_{body},Q_{body})$, allowing the use of subsumption to give it the same effect as the context it is replacing ($E'[\ldots]$).  But the underlying effect $E''[v]$ is not necessarily a subeffect of $Q_{body}$.

    However, if we first apply \textsc{T-Prompt} to $(\%~\ell~E''[v]~h)$, its effect --- $(\punblock{P_{E''}}{\ell}\setminus^{Q_h}\ell, \punblock{C_{E''}}{\ell}\setminus^{Q_h}\ell, Q_{E''}\sqcup\bigsqcup \punblock{C_{E''}}{\ell}|^{Q_h}_\ell)$ --- will be a subeffect of $\chi_{prompt}$.  For this to hold, we require that $\punblock{P_{E''}'}{\ell}\setminus^{Q_h}\ell\sqsubseteq \punblock{P_{body}}{\ell}\setminus^{Q_h}\ell$.  Fortunately, if the \emph{unmasked} $P_{E''}$ is already less than a set masked by $-\setminus^{Q_h}\ell$, then masking it by $-\setminus^{Q_h}\ell$ again will have no effect; and we just proved this above.
    Thus $\punblock{P_{E''}}{\ell}\setminus^{Q_h}\ell\sqsubseteq \punblock{P_{body}}{\ell}\setminus^{Q_h}\ell$.
    Similarly we may conclude that $\punblock{C_{E''}'}{\ell}\setminus^{Q_h}\ell \sqsubseteq\punblock{C_{body}}{\ell}\setminus^{Q_h}\ell$.
    And because of the replacement present in $C_{body}$, $Q_{E''}\sqsubseteq Q_{body}\sqcup\bigsqcup \punblock{C_{body}}{\ell}|^{Q_h}_\ell$.
    The rest of the case (plugging into $E$) follows from subsumption and Context Substitution (Lemma \ref{lem:ctxt_sub}).
\end{itemize}
\end{proof}

\section{Compositional Continuations}

In the main paper we do not discuss compositional (\lstinline|call/comp|) continuations, because they are not required for the macro-expansions we study, and because their metatheory is effectively a simplification of that for non-compositional (\lstinline|call/cc|) continuations.
However they are still useful, both because they give semantic completeness in some denotational models~\cite{SitaramF90b} (when untyped) and because they remedy some of the space consumption issues with using \lstinline|call/cc| to simulate other control operators.
Here we give additional operational rules and type rules for compositional continuations, and outline the extensions to the soundness proof.

\begin{figure*}
\begin{mathpar}
\begin{array}{l}
    e ::= \ldots \mid \mathsf{comp}~E\\
    v ::= \ldots \mid \mathsf{comp}~E\\
    \tau ::= \ldots \mid \mathsf{comp}_\ell~\tau~(P,C,\opt{Q})~\tau\\
    p ::= \ldots \mid \mathsf{cprophecy}~\ell~\chi~\leadsto\tau~\mathsf{obs}~\chi'
\end{array}
\and
\inferrule[\small E-InvokeComp]{ }{
    \sigma;((\mathsf{comp}~E'')~v)
    \overset{I}{\Rightarrow}
    \sigma;E''[v]
}
\and
\inferrule*[left=\small E-CallComp]{ E'~\textrm{contains no prompts for}~\ell }{
    \sigma;E[(\%~\ell~E'[(\mathsf{call/comp}~\ell~k)]~h)]
    \overset{I}{\Rightarrow}
    \sigma;E[(\%~\ell~E'[(k~(\mathsf{comp}~E'))]~h)]
}
\and
\inferrule*[left=T-CallComp]{
    \Gamma\vdash e : (\mathsf{comp}_\ell\;\tau\;\chi_k\;\gamma)\overset{\chi}{\rightarrow}\tau \mid \chi_e
}{
    \Gamma\vdash (\mathsf{call/comp}\;\ell\;e): \tau \mid (\chi_e\rhd\chi)\rhd(\{\mathsf{cprophecy}~\ell~\chi_k\leadsto\gamma~\mathsf{obs}~(\emptyset,\emptyset,I)\},\emptyset,I)
}
\and
\inferrule*[left=T-AppComp]{
    \Gamma\vdash k : \mathsf{comp}\;\tau'\;(P,C,\opt{Q})~\tau \mid \chi_k\\
    \Gamma\vdash e : \tau' \mid \chi_e
}{
    \Gamma\vdash (k\;e) : \tau \mid \chi_k\rhd\chi_e\rhd({P},{C},\opt{Q})
}
\and
    \inferrule[T-CompC]{
        \Sigma;\Gamma\vdash E :: \tau/(\emptyset,\emptyset,I)\leadsto\tau_0'/\chi_0\\
        \tau_0'<:\tau'\\
        \chi_0\sqsubseteq\chi
    }{
        \Sigma;\Gamma\vdash (\mathsf{comp}_\ell^{\tau'}~E) : \mathsf{comp}~\tau~\chi~\tau' \mid I
    }
\end{mathpar}
\caption{Operational semantics and type rules for compositional continuations.}
\label{fig:comp}
\end{figure*}

These extensions are given in Figure \ref{fig:comp}.  
\lstinline|call/comp tag f| captures the same continuation as \lstinline|call/cc tag f|, but in a \emph{composable} form.  \lstinline|f| is invoked with the corresponding \emph{composable} continuation, which when invoked \emph{extends} the curent evaluation context similarly to typical function application, rather than replacing it up to the enclosing appropriately-tagged prompt.  
This is seen most clearly by contrasting the semantics in Figure \ref{fig:comp} to those in Figure \ref{fig:opsem}.

Compositional continuations add a new value expression to represent compositional continuations.  Note that unlike non-compositional continuations, compositional continuations are tagged only to support the type soundness proof; the tag is not operationally required because their application (\textsc{E-CallComp}) is completely local, so no prompt matching is required.  As a result, application of compositional continuations is typed essentially the same as function application, merely using a (compositional) continuation type rather than a function type.  This includes \emph{unmodified} use of the continuation's latent effect and using the compositional continuation's result type as the result of the application, since no context is discarded.  \lstinline|call/comp| is typed nearly the same as \lstinline|call/cc| because the mechanics of capturing the continuation are the same, but because the resulting continuation does not discard any surrounding context when invoked, \textsc{T-AppComp} does not block prophecies or control effects.  As a result, \lstinline|call/comp| is typed with a new type of prophecy, a \textsf{cprophecy}, for which \textsc{V-Effects} must perform additional validation:
\[
    \forall \chi_{proph},\tau',P_p,C_p,\opt{Q_p}\ldotp 
        \mathsf{cprophecy}~\ell~\chi_{proph}\leadsto\tau'~\mathsf{obs}~(P_p,C_p,\opt{Q_p})\in \punblock{P}{\ell}\Rightarrow
        (P_p,C_p,\opt{Q_p})\sqsubseteq\chi_{proph}\land
    \tau<:\tau'
\]
Note that validation of compositional prophecies does use normal subeffecting to compare the predicted and observed effects; recall that unblocking prophecies and control effects when validating regular prophecies was necessary so a prophecy could observe a \emph{blocked} version of the predicted effect from an invocation.  Compositional continuation invocation does not block components, so no such adaptation is required.  $P$ --- the prophecies from the body of the prompt being validated --- does need to be unblocked as in the other case, as these compositional prophecies being checked may still have arisen from invoking non-composable continuations in the body, in which case the \textsf{cprophecy} may be blocked when it is passed to \textsc{V-Effects}.

It would indeed be possible, operationally, to represent composable continuations as merely function abstractions (i.e., $(\lambda v\ldotp E[v])$), but using a distinct class of values both simplifies value typing, and follows the semantics of Flatt et al.~\cite{Flatt2007} more closely.

The soundness proof extends straightforwardly.  The case for compositional continuation application structured like the function application case (though using context substitution instead of variable substitution).  The case for compositional continuation capture is structured like the non-compositional capture, in particular reusing Lemma \ref{lem:valid_prophecies}.

\fi
    \end{document}